\documentclass[12pt,titlepage]{article}
\usepackage{longtable, booktabs, overpic,url}
\usepackage{graphicx, lipsum, footmisc, enumitem, bbm}
\usepackage[top=1in, bottom=1in, left=1in, right=1in]{geometry}
\usepackage{setspace}
\usepackage{graphicx}
\usepackage{amssymb}

\let\emptyset\varnothing
 
\usepackage{cancel}
   
\newcommand\hcancel[2][black]{\setbox0=\hbox{$#2$}
\rlap{\raisebox{.45\ht0}{\textcolor{#1}{\rule{\wd0}{1pt}}}}#2}

\usepackage{tikz}                           
\usetikzlibrary{decorations.pathreplacing} 
\usepackage{amsthm}                         
\usepackage{amssymb}                       
\usepackage{amsmath}                       
\usepackage{xfrac}                          
\usepackage{comment}                        
\usepackage{authblk}                       

\newtheorem{prop}{Proposition}
\newtheorem{lemma}{Lemma}
\newtheorem{corollary}{Corollary}

\newtheorem{definition}{Definition}

\renewcommand{\ni}{\noindent}
\newcommand{\nn}{\nonumber}
\newcommand{\bs}{\bigskip}%

\newcommand{\EA}{\scalebox{0.6}{E}}

\newcommand{\bean}{\begin{eqnarray*}}
\newcommand{\eean}{\end{eqnarray*}}
\newcommand{\bea}{\begin{eqnarray}}
\newcommand{\eea}{\end{eqnarray}}
\newcommand{\be}{\begin{equation}}
\newcommand{\ee}{\end{equation}}
\newcommand{\bi}{\begin{itemize}}
\newcommand{\ei}{\end{itemize}}
\newcommand{\bal}{\begin{align}}
\newcommand{\eal}{\end{align}}

\definecolor{dgreen}{RGB}{34,139,34}

\begin{document}

\renewcommand{\thefootnote}{\fnsymbol{footnote}}

\title{With a Grain of Salt: Uncertain Veracity of External News and Firm Disclosures\footnote{We thank Yaakov Amihud, Daniel Andrei, Snehal Banerjee, Sugato Bhattacharyya, Edwige Cheynel, Jesse Davis, Eti Einhorn, Felix Feng, Sivan Frenkel, Rachel Flam, Ilan Guttman, Valentin Haddad, Mirko Heinle,  Bernard Herskovic, Moritz Hiemann, Navin Kartik, Ilan Kremer, Bart Lipman, Dan Luo, Nadya Malenko, Jordan Martel, Assaf Romm, Jacob Sagi, Ulrich Sch\"affer, Avanidhar  Subrahmanyam, Tsahi Versano, Ran Weksler, Alex Zentefis as well as  participants at  2023 Spring Meeting of the Finance Theory Group, 2020 Tel Aviv University Accounting Conference, 2020 Early Insights in Accounting Conference,  2020 Accounting and Economics Society, 2019 Junior Accounting Theory Conference, 2023 Bocconi Accounting Conference, Columbia Business School, Hebrew University of Jerusalem, Texas A\&M, and Reichman University for helpful comments and suggestions on this and previous versions of this paper. Beatrice Michaeli acknowledges financial support from Laurence and Lori Fink Center at University of California in Los Angeles.
}
}

  \author{ 
Jonathan Libgober\thanks{Jonathan Libgober is at the Department of Economics, USC Dornsife College of Letters, Arts and Sciences, and can be reached at libgober@usc.edu.}
  \hspace{1cm}
  Beatrice Michaeli\thanks{Beatrice Michaeli is at the Anderson School of Management, University of California in Los Angeles, and can be reached at beatrice.michaeli@anderson.ucla.edu.} 
  	\hspace{1cm} 
		Elyashiv Wiedman\thanks{Elyashiv Wiedman is at the School of Business Administration, Hebrew University of Jerusalem, and can be reached at elyashiv.wiedman@mail.huji.ac.il.} 
}

\maketitle

\clearpage

\vspace{2cm}

\begin{center}
\textbf{\large With a Grain of Salt: Uncertain Veracity of External News and Firm Disclosures}
\end{center}

\singlespacing

\bs
\bs

\ni \textbf{Abstract:} We examine how uncertain veracity of external news influences investor beliefs, market prices and corporate disclosures. Despite assuming independence between the news' veracity and the firm’s endowment with private information, we find that favorable news is taken “with a grain of salt” in equilibrium—more precisely, perceived as less likely veracious—which reinforces investor beliefs that nondisclosing managers are hiding disadvantageous information. Hence more favorable external news could paradoxically lead to lower market valuation. That is,  amid management silence, stock prices  may be non-monotonic in the positivity of external news. In line with mounting empirical evidence, our analysis implies asymmetric price reactions to news and price declines following firm disclosures. We further predict that external news that is  more likely veracious may increase or decrease the probability of disclosure  and link these effects to empirically observable characteristics.


\bigskip

 \bs  \ni \textbf{Keywords:} uncertain veracity, voluntary disclosure, price non-monotonicity, asymmetric price reaction, probability of disclosure

\bigskip

\bs  \ni \textbf{JEL:} D82, D83, G12

\clearpage

\renewcommand{\thefootnote}{\arabic{footnote}}

\doublespacing

\section{Introduction}

\noindent In their assessment of firms, investors rely not only on information included in corporate disclosures but also on insights from external news  such as analyst forecasts, mainstream and social media, peer reports and restatements, announcements of fiscal policies, regulations or changes in macroeconomic factors. The veracity of such news is not always straightforward: all too often it stems from rumors, errors, or otherwise uninformed, unreliable, or fraudulent sources.\footnote{This issue is particularly alarming when it comes to  social media where posts are often ``fake, out of context and straight propaganda" (Kelly 2023). For a growing empirical evidence and discussion of rumors and false information in various settings see Clarkson,  Joyce and Tutticci (2006); Schindler (2007); Marett and Joshi (2009); Kapferer (2013); Kimmel (2013); Ahern and Sosyura (2015); Kohlbeck and Vakilzadeh (2020); Alperovich, Cumming, Czellar and Groh (2021); Davis, Khadivar and Walker (2021); Liu and Moss (2022); Cai, Quan and Zhu (2023).} For instance, in November 2022, a statement asserted that the pharmaceutical company Eli Lilly and Co. intended to offer its insulin at no cost, causing an estimated decline of 15 billion dollars in the firm's market valuation. This statement was subsequently proven to be  false (Barr 2022; Harwell 2022). Similarly, in September 2008, a report---which was later deemed erroneous---about United Airlines' bankruptcy led to a 73\% loss of firm value within a mere 10 minutes of its emergence (Maynard 2008; Schaper 2008). The takeaway here is that external news needs to be taken ``with a grain of salt"  if its \emph{veracity  is uncertain}. But how does such uncertainty impact the incentives of companies to disclose their private information to the public, in response to or in anticipation of external news that might not be veracious? And what are the implications for investors?


This paper provides answers to these questions in the context of a disclosure model. We study the implications of uncertainty regarding the veracity of external news and elucidate that it can influence investor beliefs, corporate disclosure behavior and market prices in subtle ways. Several of our results diverge significantly from prior work set in disclosure settings that does not provide as central a role for belief updating about veracity. Among these results, most notable are our observations that investors are more skeptical about favorable news and prices can be non-monotonic 
in the positivity of external news. Our contribution is to clarify why such phenomena naturally occur \emph{as a direct consequence of} market anticipation that a given piece of news might not be veracious.  

Our model features a manager (``she") who may be endowed with private information about her firm value and, if informed, can voluntarily and truthfully disclose it to investors. The manager's objective is to maximize her firm's market price. The investors evaluate the firm after considering the manager's disclosure and news from an external source (``signal") with uncertain and unobservable veracity. In the main part of the paper, we assume that if the signal is veracious, it reflects the firm value; otherwise it is uninformative. This 
``truth-or-noise'' signal structure 
allows for the uncertainty over the veracity to have the greatest impact possible so that we can study its implications most clearly. 
It also fits the examples we have in mind: social media ``tweets" and ``posts" as well as mainstream media broadcasts and articles may be substantiated or based on untrue rumors and deliberate misinformation; reports may be truthful or false; announcements may be correct or erroneous; 
independent third-party experiments and trials may or may not be performed based on the appropriate procedure thus yielding reliable or unreliable results. As we illustrate in an extension, the truth-or-noise structure is not crucial for our main results, as long as the  \emph{veracity is uncertain}. 

Even though the manager's endowment with private information and the veracity of the signal are  independent events, we find that the investor beliefs about them are \emph{endogenously intertwined} since both are updated jointly following nondisclosure. 
Specifically, external news that is relatively unfavorable strengthens the joint beliefs that the news is veracious and that the nondisclosing manager is informed; relatively favorable news has the opposite effect. This pattern is driven by the manager's (endogenous) strategic disclosure behavior and what can be inferred from it: In an equilibrium where good       information is disclosed and  bad is withheld, the investors understand that an informed manager who did not disclose must have observed sufficiently disadvantageous information. Thus, when faced with a silent manager, the investors perceive a relatively unfavorable signal as more likely veracious, which in turn strengthens their beliefs that the silent manager is informed. Conversely: after observing a favorable signal, the investors cannot rationally believe that the signal is veracious and that the silent manager is informed, at the same time.

Overall, our results imply that investors are more skeptical about favorable external news. 
 This shift in investor beliefs about veracity has two implications: First, the market price may be  \emph{non-monotonic}  in the sense that relatively more favorable external news may paradoxically lead to lower market prices---this is because such news is viewed as less likely veracious. Second, the price is more sensitive to external news that is less favorable---this is because such news is perceived as more likely veracious. This finding  is consistent with mounting evidence on asymmetric asset price reactions to external news in various settings (e.g., Aggarwal and Schirm 1998; Goldberg and Grisse 2013; Blot, Hubert and Labondance 2020; Capkun, Lou, Otto, and Wang 2022; Xu and You 2023).

 Moving to the manager's equilibrium behavior, we first focus on the case of early disclosure where the manager decides whether to reveal her information before the external signal (e.g., because the firm's conference call is scheduled prior to media broadcast or release of macroeconomic news; alternatively, because the manager can only disclose before the so called ``quiet period" in the lead-up to an Initial Public Offering or the closing of a business quarter). A manager who has no information has no choice but remain silent; an informed manager can reveal or withhold her firm's value. Because the external signal has not arrived yet, the informed manager considers the expected (future) market price, conditional on the value she observed. Even though her expectation of the nondisclosure price is non-monotonic in the firm value, we still obtain a uniqueness result due to a key property: the price in this case exhibits a steep decline around a specific signal realization, making disclosure of any value beyond this realization even more attractive.  We also find that the expected nondisclosure price exceeds low firm values in the region before the steep decline. In this region, the manager benefits from the uncertain veracity  and relies on the external source to reveal the same information that she observed. 
Overall, we find that the presence of external sources  discourages (early) corporate disclosures and this effect is exacerbated when the signal is ex ante more likely veracious.

We next consider late disclosure: that is, a case where the manager's decision succeeds the external news (e.g., because macroeconomic  announcements, mainstream
media articles and broadcasts, or social media posts are released unexpectedly or before the
scheduled manager’s conference call). In this case the manager observes the external signal and decides whether to react by disclosing her information or to remain silent.
We find that the positivity of the external news can be divided into two regions: In the lower region, external news is relatively unfavorable and encourages disclosure. This effect is exacerbated when the likelihood for veracity is higher. The opposite is true in the higher region where external news is relatively more favorable and a higher veracity likelihood discourages disclosure.
Our predictions are  consistent with empirical evidence that a firm's likelihood to disclose information is influenced by the positivity of the information revealed by peers (Capkun, Lou, Otto, and Wang 2022). 


Our third disclosure setting assumes that managers can choose the timing of their disclosure at some rescheduling cost (e.g., because of reduced
preparation time, reallocation of resources and manpower to meet new timeline, cancellation of previously scheduled announcement, or disruption
of internal planning, decision-making processes and investor relations). We note that if advancing disclosure is more costly than delaying it, the manager at least weakly prefers to delay because external news might be favorable and there is lower cost associated with waiting. However, if advancing is cheaper, managers observing advantageous information prefer to disclose early and avoid the (higher) cost of delaying---thus, we predict that disclosures of good information occur before the arrival of external news. 

Prior to concluding the paper, we consider more general signal structures and incorporate frequent adjustment of market prices. We find that the latter may encourage disclosures. Our results in this extension may offer an explanation of a finding in Sletten (2012) where a firm disclosing in response to a peer's restatement experiences a decrease in its stock price. Our empirical predictions are formally summarized at the end of the study.

\section{Related Literature}

Our paper is related to  the literature on voluntary disclosure, initiated by Grossman (1981) and Milgrom (1981), and especially to the work about the impact of external news on  disclosure when the market is uncertain whether firms have private information.\footnote{Uncertain information endowment is one of the  frictions that prior literature finds to prevent information unraveling (Dye 1985). In the absence of any known friction, Einhorn (2018) finds that unraveling may also be prevented when there are competing sources of information in the market.}
Several papers (Frenkel, Guttman, and Kremer 2020; Dye and Sridhar 1995) consider settings where the arrival of external news depends on the firm's information endowment.
In these studies the market updates beliefs about the manager's information endowment because of the assumed correlation---without it, there is no effect on beliefs. In contrast, we posit that the arrival of external news is uncorrelated with the manager's endowment but the news' veracity is uncertain. This substantially alters how investors update beliefs as well as how market prices react in our equilibrium, and opens the door for non-monotonicity. Furthermore, in the above-mentioned studies, the content of external news is either orthogonal to firm value with certainty (as in Dye and Sridhar 1995)  or perfectly informative (as in Frenkel, Guttman and Kremer 2020). Under specific assumptions about the  correlation between external news and firm value, 
equilibria may only exist in mixed strategies (a possibility explored in recent work by Frenkel, Guttman, and Kremer 2023).

In several papers (e.g., Acharya, DeMarzo, and Kremer 2011; Menon 2020), the firm value and the external signal are imperfectly and positively correlated, with the correlation \emph{known and fixed}. Thus, while the arrival of the  signal affects the (conditional) distribution of the firm value, the \emph{investor beliefs about this correlation remain constant regardless of the  signal realization}. As a result, market prices in these models are \emph{monotonic} (see Section \ref{structuredisc} for in-depth discussion). In contrast, we study a setting where investors are \emph{uncertain about the veracity} of the external signal, i.e., about the correlation between external news and firm value. The arrival of the signal in our model leads to a non-trivial ``grain of salt"  pattern of beliefs' updating and, as a result, non-monotonicity in market price.
To our knowledge, this result is a novel implication of the form of endogenous updating that we study.

Some of our results---that  market prices can be non-monotonic and their reactions to news can be asymmetric---are reminiscent of those identified in several prior studies: however, \emph{the economic forces driving the results are substantially different}. In Veronesi (1999) markets react asymmetrically because risk-averse investors hedge against changes in their uncertainty about the fundamental and overreact (underreact) to bad (good) news when times are good (bad). In Banerjee and Green (2015) asymmetric price reactions emerge because the traders have mean-variance preferences: 
favorable news reduces the expected fundamental value and amplifies the negative risk effect, whereas unfavorable news increases it and counteracts the risk. In their model non-monotonicity occurs because for extremely favorable news the risk effect outweighs the mean effect thereby resulting in a price decrease. 

In Bond, Goldstein and Prescott (2010),  discontinuity  in  the fundamental arises due to corrective actions that firms undertake  after inferring information reflected in stock prices. Our main setting also features discontinuity: it is microfounded using a model of disclosure with external signals. In our case this imposes constraints in equilibrium: for instance, ruling out the possibility of a steep increase in the price, whereas in their model a jump upwards can occur depending on the impact of the intervention. In addition, the discontinuity we obtain is in the case of \emph{anticipatory} disclosure, and we do not obtain a discontinuity in the case where the manager's decision is corrective---i.e., when disclosure may be late. Lastly,  in Genotte and Leland (1990) and Barlevy and Veronesi (2003) prices can be discontinuous due to exogenous insurance of investment portfolio or presence of uninformed investors who are willing to buy at relatively high prices and endogenously act as insurance.

\section{Model and Benchmark}

\label{setting}

We extend the voluntary disclosure framework with uncertain information endowment (Dye 1985) by considering an additional source of public information with uncertain veracity.

\textbf{Players and payoffs.} The model entails a manager (``she") who runs a firm with  value $v\in[v_{min},v_{max}]$ and is a price-maximizer. 
A group of risk-neutral investors observes all publicly available information $\Omega$, forms beliefs, and prices the firm at $P(\Omega) = \mathbb{E}[v \lvert \Omega]$. In the main part of the paper we assume that the price is set once, after $\Omega$ is realized. In Section \ref{pref} we relax this assumption and allow investors to adjust the market price more frequently.

\textbf{Information structure.} The common prior belief is that $v$ is drawn from a cumulative differentiable distribution function $G$ (and probability distribution function $g$), which we assume is log-concave with prior expectation $\mathbb{E}[v]=\mu$.\footnote{The log-concave family includes multiple common distributions: normal, uniform, exponential, chi, beta, gamma, logistic etc. Many of our results do not require log-concavity but, to facilitate consistency, we invoke this assumption throughout.} The publicly available  $\Omega$ consists of an external signal (on occasion, just ``signal")  and the manager's voluntary disclosure: 

\begin{itemize}
 
\item[1.] \emph{External signal.} There is a  public signal $s$ with unobservable veracity $\phi\in \{V,N\}$. 
Throughout the main analysis we assume that with probability $q\in(0,1)$, the signal is veracious ($\phi=V$) and reflects the firm value, $s=v$. Otherwise, the signal is non-veracious ($\phi=N$) in a sense that  $s=x$ where $x\in[v_{min},v_{max}]$ is drawn from the same probability distribution as the firm value.\footnote{The support of a non-veracious external signal is the same as that of a veracious one---if this were not the case, the investors could learn whether the external signal is veracious for some signal values.} Such ``truth-or-noise" information structures are ubiquitous  in various settings.\footnote{For example, in Lewis and Sappington (1994) product buyers receive a signal that is either fully informative or uninformative about their taste; in Kanodia, Singh and Sperro (2005, Section 5) the accounting measurement
provides a truth-or-noise signal of the firm's investment; in Ottaviani and Sorensen (2006) an expert's signal is either informative or uninformative---later communicated by a cheap talk to the sender; in Banerjee and Green (2015) uninformed traders face informed traders whose information is either truth or noise; in Guttman and Marinovic (2018) a borrower privately observes a truth-or-noise signal and reports it with bias to avoid a covenant violation; in Banerjee, Davis and Gondhi (2020) a manager's signal either perfectly reflects an investment profitability or is pure noise; in Goncalves, Libgober, and Willis (2023) the retracted information could have been correct or false. See also the application examples in Johnson and Myatt (2006) and Ganuza and Penalva (2010).}
In our disclosure model this structure fits well the  examples  we have in mind (e.g., social media ``tweets" and ``posts" as well as mainstream media broadcasts and articles may be substantiated or based on untrue rumors and deliberate misinformation; reports may be truthful or false; announcements may be correct or erroneous; 
independent third-party experiments and trials may or may not be performed based on the appropriate procedure thus yielding reliable or unreliable results). 
Furthermore, the truth-or-noise structure streamlines our analysis and allows us to convey a more straightforward  intuition but, as we show in Section \ref{structuredisc}, is not crucial for our main results---they hold as long the signal \emph{veracity is uncertain}. 

  \item[2.] \emph{Manager's information endowment and voluntary disclosure.}  With probability $p\in (0,1)$, the manager is informed ($\kappa=I$) about the firm value, and otherwise is uninformed ($\kappa=U$). 
  We assume that $\kappa$ is independent of $v$ and $\phi$ and unknown to investors. An uninformed manager cannot credibly communicate the lack of information and has no choice but remain silent ($d=\emptyset$). An informed manager can voluntarily disclose the firm value  ($d=v$) to investors before the arrival of the external signal (which we refer to as ``early disclosure'')   or after (which we refer to as ``late disclosure''). 
  Following the voluntary disclosure literature, any disclosed value is verifiable and thus truthful.   We further assume that in equilibrium, a manager who is indifferent between disclosure decisions chooses to remain silent.\footnote{Assuming a particular tie-breaking rule is common in the literature as a way to avoid equilibrium multiplicity. Our results characterize all equilibria, even without this assumption.} We describe the outcomes without assumptions on how indifference is broken in Section \ref{first}. In this section we also show that, except for the indifference point, the manager always has strict preferences. Thus, \emph{while we allow for mixed strategies in our analysis, these do not arise in equilibrium}.

\end{itemize}

\noindent To avoid confusion, we refer to relatively high/low $v$ as ``\emph{good}/\emph{bad} (manager's) information or value" and to relatively high/low $s$ as ``\emph{favorable}/\emph{unfavorable} (external) news or signal."

\textbf{Timeline.} The timeline of events is illustrated in Figure 1. 
At date 1, the manager observes the firm value $v$ with probability $p$. At date 3, the external signal $s$ is realized. Depending on the specific setting (early case in Section \ref{first}, late case in Section \ref{second} or dynamic case in Section \ref{freqd}) the manager can disclose the observed value at date 2 
and/or at date 4. 
At date 5, the investors price the firm.

\begin{figure}[t] 
\setlength{\unitlength}{1.cm}
\begin {picture}(14,4)\thicklines
\put(0.,2){\vector(1,0){17}}

\put(0,1.75){\line(0,1){0.5}}
\put(0,2.5){1}
\put(0,1.25){Manager's}
\put(0,0.75){information}
\put(0,0.25){endowment}

\put(3.6,1.75){\line(0,1){0.5}}
\put(3.6,2.5){2}
\put(2.,3.1){(``early disclosure")}

\put(7.2,1.75){\line(0,1){0.5}}
\put(7.2,2.5){3}
\put(7.2,1.25){External signal}
\put(7.2,0.75){realization}

\put(10.8,1.75){\line(0,1){0.5}}
\put(10.8,2.5){4}
\put(9.5,3.1){(``late disclosure")}

\put(5.2,0){\textcolor{blue}{\textbf{Manager's disclosure}}}

\put(3.6,0.1){\color{blue}\line(0,0){0.3}}
\put(3.6,0.5){\color{blue}\vector(0,0){0.5}}

\put(10.8,0.1){\color{blue}\line(0,0){0.3}}
\put(10.8,0.5){\color{blue}\vector(0,0){0.5}}

\put(3.6,0.1){\color{blue}\line(1,0){0.3}}
\put(4,0.1){\color{blue}\line(1,0){0.3}}
\put(4.4,0.1){\color{blue}\line(1,0){0.3}}
\put(4.8,0.1){\color{blue}\line(1,0){0.3}}

\put(9.6,0.1){\color{blue}\line(1,0){0.3}}
\put(10,0.1){\color{blue}\line(1,0){0.3}}
\put(10.4,0.1){\color{blue}\line(1,0){0.4}}

\put(14.4,1.75){\line(0,1){0.5}}
\put(14.4,2.5){5}
\put(14.4,1.25){Investors}
\put(14.4,0.75){price} 
\put(14.4,0.25){the firm}

\end {picture}
\begin{center}
\bs
\textbf{Figure 1}: Timeline of events
\end{center}
\end{figure}

\textbf{Benchmark.} As a benchmark (superscript ``B"), suppose there is no signal (as in Dye 1985; Jung and Kwon 1988) or the signal is non-veracious with certainty ($q=0$ so that investors ignore it). 
The  price following disclosure is $P(v)=v$. The price following nondisclosure is
\bea
P(\emptyset) &= &\Pr(I \lvert \emptyset) \cdot \mathbb{E}[v \lvert v \leq v^B] + \Pr(U \lvert \emptyset)\cdot \mu
\label{PNDD}
\eea
when the investors conjecture that an informed manager discloses all values above a threshold $v^B$.\footnote{Here, $\Pr(I \lvert \emptyset) =  \frac{p \cdot G(v^B)}{1-p+p\cdot G(v^B)}$ and $\Pr(U \lvert \emptyset)=1-\Pr(I \lvert \emptyset)$. Note that these conditional probabilities also depend on the disclosure threshold---but we suppress it to avoid clutter.}  In equilibrium, when the manager observes $v=v^B$, she is indifferent between disclosing and withholding her information. That is, $P(\emptyset)=P(v=v^B)$.

\setcounter{lemma}{-1}

\begin{lemma}\emph{\textbf{(Dye 1985; Jung and Kwon 1988)}}
\label{lm1}
When  $q=0$, 
there exists a unique threshold $v^B\in(v_{min},\mu)$, such that the manager discloses if $v > v^B$ and withholds otherwise. The threshold $v^B$ 
 is decreasing in the probability of information endowment, $p$.
\end{lemma}

The key insight in Dye (1985) is that informed managers can pretend to be uninformed because the market is unaware about their information endowment. 
In this paper, we show that external signals with uncertain veracity affect the ability of managers to pretend to be uninformed in two ways. First, because these signals may provide  information about firm value, they directly affect the market price and the managers' disclosure decisions. Second,  the signals also influence the investor beliefs regarding the managers' information endowment and the signal veracity. As a result, the signals indirectly affect the market price and corporate disclosures. This second channel drives many of our observations, such as non-trivial beliefs' updating and prices' non-monotonicity. 

\section{Investor Beliefs and Market Prices} \label{sect:updating}

We solve the model by backward induction. First, we consider the investor beliefs and market price that arise  at date 5, when the signal realization and the manager's disclosure or nondisclosure are given. To proceed with the analysis, for now we conjecture that the disclosure follows a threshold rule:
in particular, the manager withholds values that are lower than some---\emph{fixed at that point of time}---threshold $\widehat{v} \in [v_{min},v_{max}]$ and discloses otherwise.\footnote{In our late setting the manager's date-4 choice of disclosure threshold depends on the realized at date 3 signal. However, when the beliefs and price are formed at date 5 the signal had already been realized and the manager had already disclosed all values above that---\emph{fixed at date 5}---threshold. } 
In Section \ref{main}, when we solve the date-2 and date-4 disclosure decisions, we formally prove that the equilibrium is indeed of threshold type and is unique.

Similar to the benchmark case, the market price when the manager discloses  is determined only by the disclosed value, 
\bea
P(s, v) = \mathbb{E}[v\lvert v,s]=v.
\label{PD}
\eea
The investors disregard the external signal because the manager observes the firm value precisely and her disclosure is truthful.  
When the firm is silent, the price depends on the investor beliefs about the manager's information endowment and the signal veracity. Specifically, the investors consider four possible events: (i) the manager is uninformed and the signal is non-veracious; (ii) the manager is uninformed and the signal is veracious; (iii) the manager is informed and the signal is non-veracious; (iv) the manager is informed and the signal is veracious. Therefore the market price can be expressed as:
\bea
P(s, \emptyset) = \mathbb{E}[v \lvert s, \emptyset] 
&=& \Pr(U, N \lvert s, \emptyset) \cdot \mathbb{E}[v \lvert U, N, s, \emptyset] 
+  \Pr(U, V \lvert s, \emptyset) \cdot \mathbb{E}[v \lvert U, V, s, \emptyset]
\nn
\\
&&+\Pr(I , N \lvert s, \emptyset)\cdot  \mathbb{E}[v\lvert I, N, s, \emptyset ]  
+  \Pr(I, V \lvert s, \emptyset) \cdot \mathbb{E}[v \lvert I, V, s, \emptyset].  \qquad 
\label{PNDgeneral0}
\eea

To simplify \eqref{PNDgeneral0}, consider the case where the investors believe that the signal is non-veracious and the manager is uninformed. Then, because there is nothing to be learned from the signal and the manager's silence, the market expectation about firm value   is simply the prior, $\mathbb{E}[v \lvert U, N, s, \emptyset]=\mathbb{E}[v]=\mu$.
If investors believe the signal is veracious, their expectation of the firm value is simply the signal, 
$\mathbb{E}[v \lvert \kappa, V, s, \emptyset]=s$, regardless of whether they believe the manager is informed ($\kappa=I$) or not ($\kappa=U$). This is because a veracious signal perfectly reflects the firm value. 
The last case---when investors believe the manager is informed and the signal is non-veracious---introduces sensitivity of the price function to the disclosure threshold. Note that, in this case, there is nothing to be learned from the signal, and so the investors disregard it. The  nondisclosure decision, however, indicates that the manager prefers to withhold the observed value. 
Thus the market expectation about the firm value in this case is $\mathbb{E}[v\lvert I, N, s, \emptyset ]=\mathbb{E}[v\lvert v \leq \widehat{v} ]$.
We can simplify the nondisclosure  price in \eqref{PNDgeneral0}, 
\bea
P(s, \emptyset) &= &\Pr(U, N \lvert s, \emptyset) \cdot \mu 
+  \Pr(U, V \lvert s, \emptyset) \cdot s
\nn
\\
&&+ \Pr(I, N \lvert s, \emptyset) \cdot  \mathbb{E}[v\lvert v \leq \widehat{v} ]  
+  \Pr(I, V \lvert s, \emptyset) \cdot s. \quad  \quad 
\label{PNDgeneral5}
\eea

Our next results describe the probabilities the market assigns to each event $(\phi,\kappa)$, and illustrate how the beliefs of the investors about the occurrence of these events are affected by the observed signal in a nontrivial way (despite the fact that they are initially independent).

\begin{lemma}
\label{beliefs}
    Let $\gamma(s,\widehat{v}) \equiv q(1-\mathbbm{1}_{s>\widehat{v}}\cdot p)+(1-q)\big(1-p+p G(\widehat{v}))$.     For a given threshold $\widehat{v}$, the investor beliefs conditional on $d=\emptyset$ and $s$ are given by 
\bean
\Pr(I, V \lvert s,\emptyset)&=& \frac{q p \cdot \mathbbm{1}_{s\leq\widehat{v}}  }{\gamma(s, \widehat{v})},        \quad \quad
\Pr(I,N\lvert s,\emptyset)= \frac{(1-q) p G(\widehat{v})}{ \gamma(s,\widehat{v})}, \\
\Pr(U,V \lvert s,\emptyset)&=& \frac{q (1-p) }{ \gamma(s,\widehat{v})}, \quad \quad
\Pr(U,N\lvert s,\emptyset)= \frac{(1-q) (1-p) }{ \gamma(s,\widehat{v})}.
   \eean
\end{lemma}  

\begin{prop} \emph{\textbf{(Non-constant joint beliefs)}}
\label{joint}
For given threshold $\widehat{v}$ the joint investor beliefs about $\phi$ and $\kappa$ depend on the realization of the external signal $s$: 
\bean
\Pr(\kappa, \phi\lvert s\leq\widehat{v}, \varnothing)&>& 
\Pr(\kappa, \phi\lvert s>\widehat{v}, \varnothing) \text{ if $\phi=V$ and $\kappa=I$,}
\\
\Pr(\kappa, \phi\lvert s\leq\widehat{v}, \varnothing)&<& 
\Pr(\kappa, \phi\lvert s>\widehat{v}, \varnothing) \text{ otherwise.}
\eean
\end{prop}

\begin{figure}[t]
\setlength{\unitlength}{.5cm}
\vspace{-0.6cm}
\begin {picture}(11.5,11.5)\thicklines
\put(0,0){\includegraphics[width=.41\textwidth]{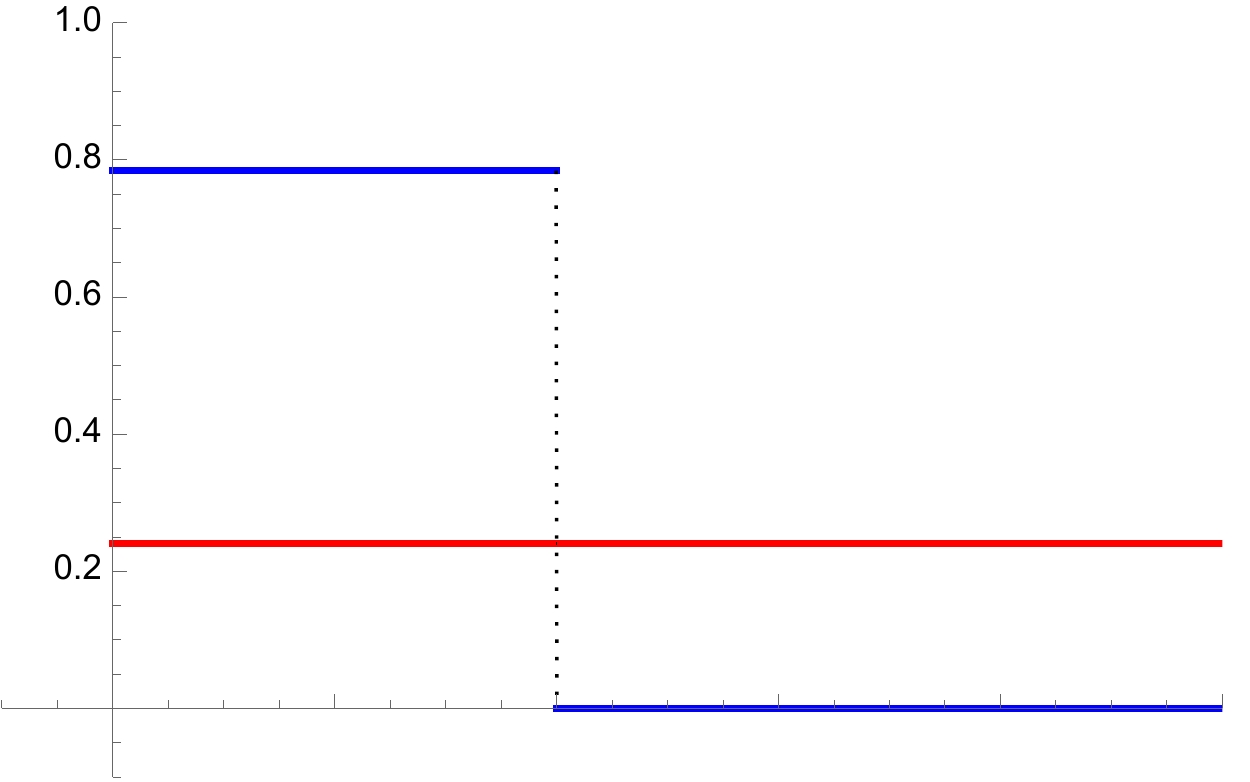}}
\put(18,0){\includegraphics[width=.41\textwidth]{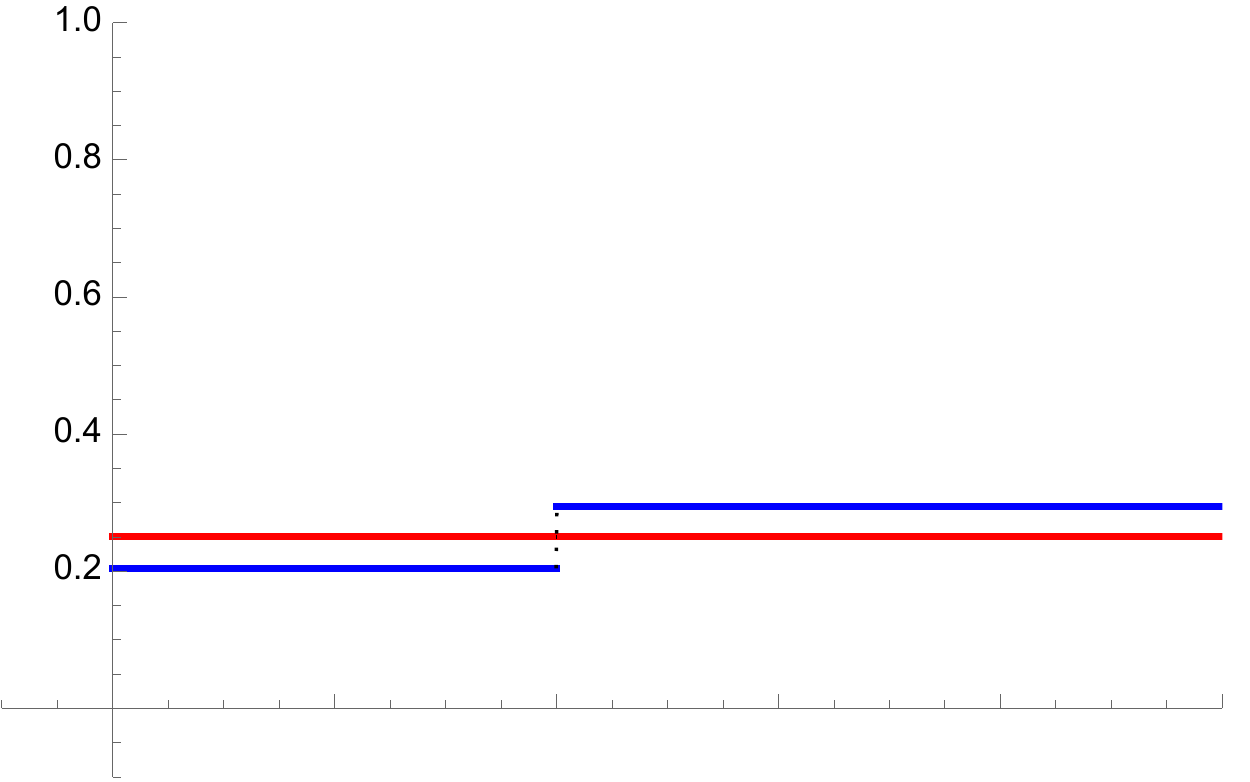}}
\put(0,-12){\includegraphics[width=.41\textwidth]{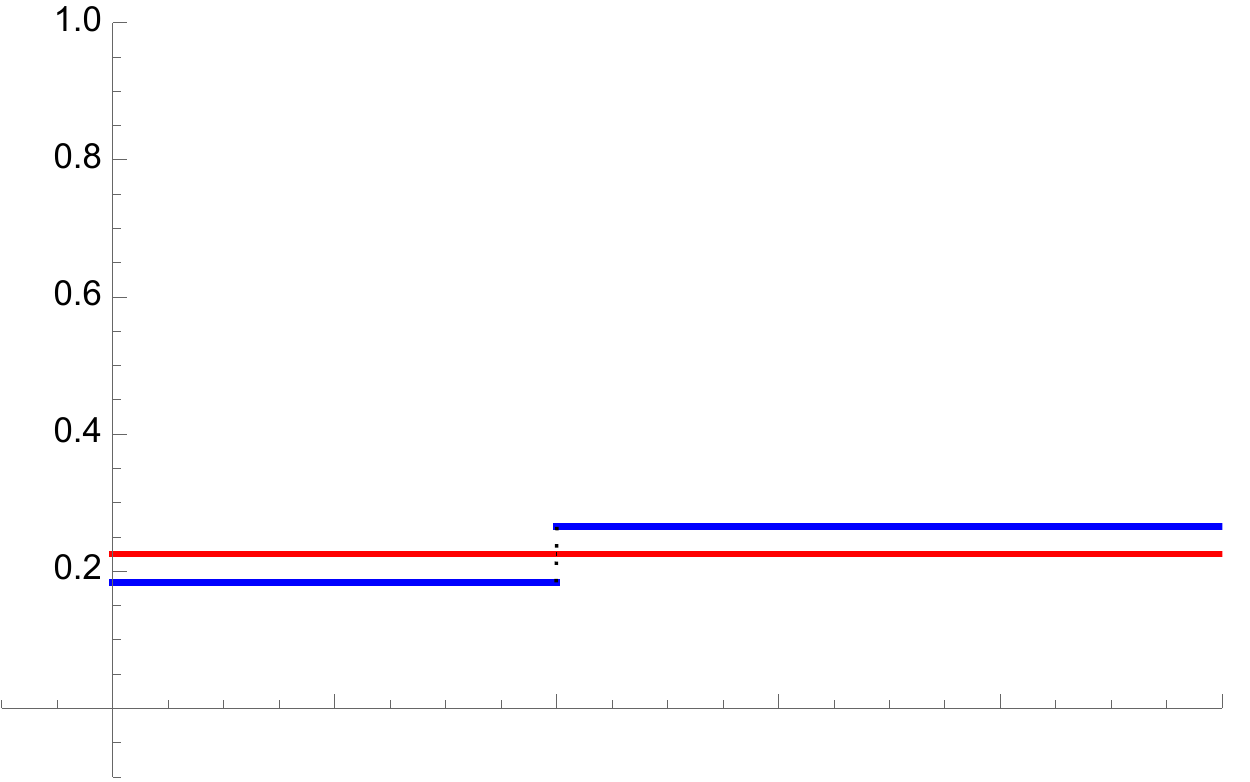}}
\put(18,-12){\includegraphics[width=.41\textwidth]{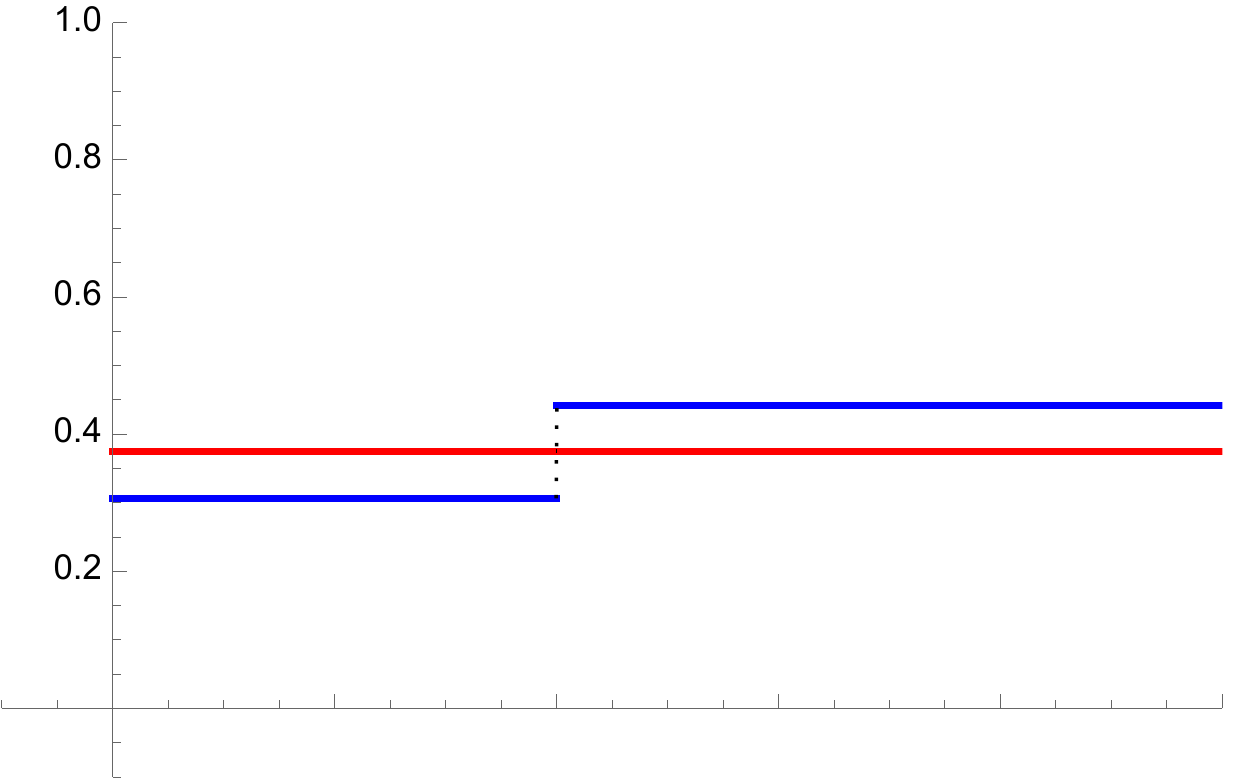}}
\put(13,-0.1){\footnotesize{$s$}}
\put(13,-12.1){\footnotesize{$s$}}
\put(31,-0.1){\footnotesize{$s$}}
\put(31,-12.1){\footnotesize{$s$}}
\put(0,9){\footnotesize{Beliefs}}
\put(18,9){\footnotesize{Beliefs}}
\put(0,-3){\footnotesize{Beliefs}}
\put(18,-3){\footnotesize{Beliefs}}
\put(6,-0.1){\footnotesize{$s=\widehat{v}$}}
\put(6,-12.1){\footnotesize{$s=\widehat{v}$}}
\put(24,-0.1){\footnotesize{$s=\widehat{v}$}}
\put(24,-12.1){\footnotesize{$s=\widehat{v}$}}
\put(13.5,0.7){\color{blue}\scriptsize{$\Pr(V,I\lvert s, \emptyset)$}}
\put(13.5,2.5){\color{red}\scriptsize{$\Pr(V)\Pr(I\lvert \emptyset)$}}
\put(31.5,3){\color{blue}\scriptsize{$\Pr(V,U\lvert s, \emptyset)$}}
\put(31.5,2.4){\color{red}\scriptsize{$\Pr(V)\Pr(U\lvert \emptyset)$}}
\put(13.5,-9.2){\color{blue}\scriptsize{$\Pr(N,I\lvert s, \emptyset)$}}
\put(13.5,-9.7){\color{red}\scriptsize{$\Pr(N)\Pr(I\lvert \emptyset)$}}
\put(31.5,-8){\color{blue}\scriptsize{$\Pr(N,U\lvert s, \emptyset)$}}
\put(31.5,-8.7){\color{red}\scriptsize{$\Pr(N)\Pr(U\lvert \emptyset)$}}
\put(6,-1){\small{\textbf{panel (a)}}}
\put(6,-13){\small{\textbf{panel (c)}}}
\put(24,-1){\small{\textbf{panel (b)}}}
\put(24,-13){\small{\textbf{panel (d)}}}
\end{picture}
\bigskip
\vspace{6.5cm}
\begin{center}
\textbf{Figure 2}: Joint investor beliefs about $\phi$ and $\kappa$ at $s=\widehat{v}$ \\ {\footnotesize{Numerical example with uniform distribution, $v_{min}=0$, $v_{max}=1$, $\Pr(I)=p=0.6$, $\Pr(V)=q=0.4$. \\ Under these parameter values, $\widehat{v}=0.4$, $\Pr(U\lvert \emptyset)=0.625$ and $\Pr(I\lvert \emptyset)=0.375$.}}
\end{center}

\end{figure}

At the heart of our results, especially the price non-monotonicity, is how the market makes inferences about the signal's veracity based on the manager's disclosure behavior.
The intuition behind the steep decrease in joint beliefs about the event $(\phi=V,\kappa=I)$ at $s=\widehat{v}$ is particularly insightful. In an equilibrium  with disclosure threshold $\widehat{v}$, the investors understand that an informed manager who did not disclose must have observed $v \leq \widehat{v}$.\footnote{We formally prove the existence of unique threshold equilibrium in Section \ref{main}.} Thus, if the external signal exceeds the disclosure threshold ($s>\widehat{v}$), it will be rationally inconsistent for the investors to believe that the signal is veracious and that the manager is informed at the same time. Thus, they infer that $\Pr(I, V \lvert s >\widehat{v}, \emptyset)=0< \Pr(I\lvert \emptyset)\Pr(V)$. 
In contrast, investors perceive a signal below the disclosure threshold ($s\leq \widehat{v}$) as more likely to be veracious, which strengthens their joint beliefs that the nondisclosing manager is also informed---in this case, $\Pr(I, V \lvert s \leq \widehat{v}, \emptyset)> \Pr(I\lvert \emptyset)\Pr(V)$.\footnote{It can also be show that $\Pr(V \lvert  s\leq \widehat{v},\emptyset)>\Pr(V)> \Pr(V \lvert  s>\widehat{v},\emptyset)$ and $\Pr(I\lvert s\leq \widehat{v},\emptyset)>\Pr(I \lvert \emptyset)>\Pr(I\lvert s>\widehat{v},\emptyset)$.}  To summarize, unfavorable $s$ strengthens the beliefs of investors that the external news is veracious and that the nondisclosing manager is informed. Conversely, favorable $s$ weakens the investor beliefs: such news is less likely to accurately describe the value of a firm run by an informed, nondisclosing manager. This result is graphically illustrated in panel (a) of Figure 2.
The rest of the beliefs, i.e., those about events other than $(\phi=V, \kappa=I)$,  are illustrated in panels (b), (c) and (d): in all of these cases more favorable external signals strengthen the joint beliefs.\footnote{The non-trivial effect of external signals on beliefs arises \emph{only} in the presence of strategic disclosure. To see why, note that if the manager cannot disclose for exogenous reasons (or, equivalently, has no information with certainty, $p=0$), the lack of disclosure and the signal realization carry no information about $\phi$ and $\kappa$ and so
beliefs in this hypothetical scenario remain constant: $\Pr(\phi,\kappa\lvert s,\emptyset)=\Pr(\kappa \lvert \emptyset)\Pr(\phi\lvert s)=\Pr(\kappa)\Pr(\phi)$, for any $s$, $\phi$ and $\kappa$.}

Putting our observations together, it is straightforward to provide a formal description of the market price following manager's silence $d=\emptyset$ and external signal $s$. 
\begin{lemma}
\label{price}
Let $\gamma(s,\widehat{v}) \equiv q(1-\mathbbm{1}_{s>\widehat{v}}\cdot p)+(1-q)\big(1-p+p G(\widehat{v}))$. For a given disclosure threshold $\widehat{v}$, the nondisclosure price is:
\bea
P(s, \emptyset) = \frac{q(1-\mathbbm{1}_{s>\widehat{v}}\cdot p) s+(1-q)((1-p)\mu+p G(\widehat{v}) \mathbb{E}[v\lvert v\leq \widehat{v}])}{\gamma(s,v)}.
 \label{priceformula}
\eea
\end{lemma}

\begin{figure}[t]
\setlength{\unitlength}{.5cm}
\vspace{-0.6cm}
\begin {picture}(11.5,11.5)\thicklines
\put(0,0){\includegraphics[width=.3\textwidth]{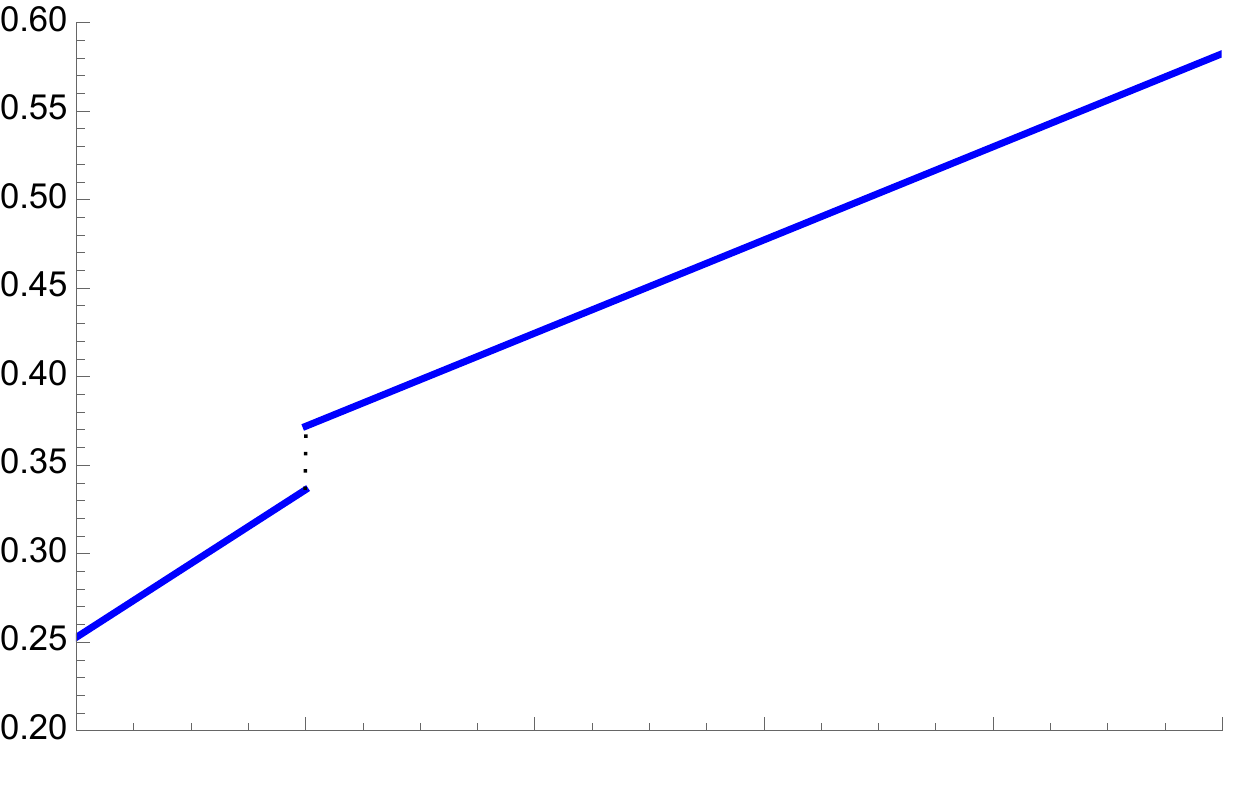}}
\put(12,0){\includegraphics[width=.3\textwidth]{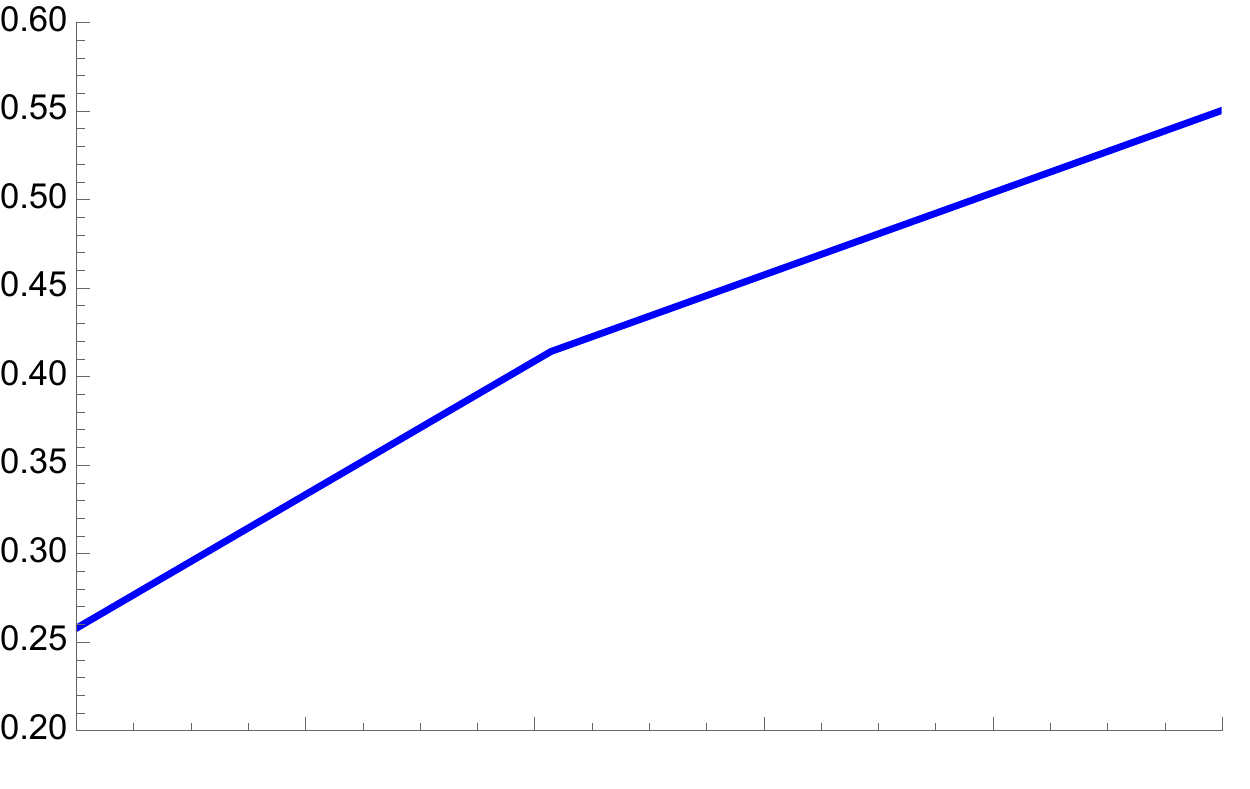}}
\put(24,0){\includegraphics[width=.3\textwidth]{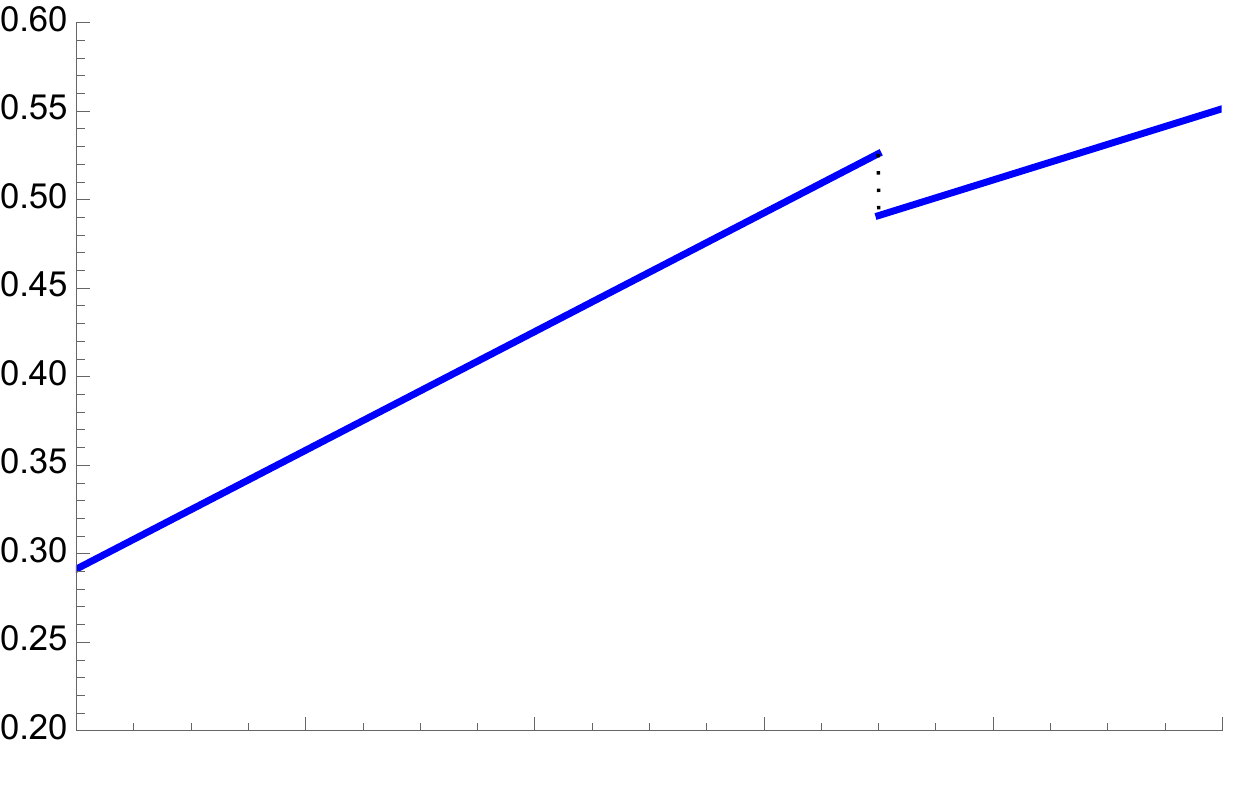}}
\put(10,-0.1){\footnotesize{$s$}}
\put(22,-0.1){\footnotesize{$s$}}
\put(34,-0.1){\footnotesize{$s$}}
\put(0,7){\footnotesize{Price}}
\put(12,7){\footnotesize{Price}}
\put(24,7){\footnotesize{Price}}
\put(3,-0.5){\footnotesize{$\widehat{v}=0.2<v^B$}}
\put(15,-0.5){\footnotesize{$\widehat{v}=v^B$}}
\put(27,-0.55){\footnotesize{$\widehat{v}=0.7>v^B$}}
\put(3,-1.5){\small{\textbf{panel (a)}}}
\put(15,-1.5){\small{\textbf{panel (b)}}}
\put(27,-1.5){\small{\textbf{panel (c)}}}
\end{picture}
\bigskip
\vspace{0.5cm}
\begin{center}
\textbf{Figure 3}: 
Nondisclosure price for given $\widehat{v}$ as a function of the external signal $s$\\ {\footnotesize{Numerical example with uniform distribution, $v_{min}=0$, $v_{max}=1$, $p=0.5$, $q=0.3$. Here, $v^B=0.42$.}}
\end{center}

\end{figure}

\noindent These observations drive the following pair of results, which clarify how the market's inference about veracity substantively influences the features of the market price: 

\begin{prop} \emph{\textbf{(Asymmetric price sensitivity)}}
\label{sensitivity}
For given threshold $\widehat{v}$ the nondisclosure market price is more sensitive to unfavorable external news, $\frac{\partial P(s, \emptyset \lvert  s\leq \widehat{v})}{\partial s}> \frac{\partial P(s,\emptyset \lvert  s> \widehat{v})}{\partial s}>0$.
\end{prop}

\begin{prop} \emph{\textbf{(Price non-monotonicity)}}
\label{pricenon}
For given threshold the nondisclosure price is non-monotonic in $s$ if $\widehat{v}> v^B$.
\end{prop}

\ni Three observations are worth emphasizing. First, because investors are uncertain about the veracity of the signal, they never price the firm at $s$ (except for a knife-edge case). Second, the investors' reaction is more sensitive to sufficiently unfavorable signals ($s<\widehat{v}$) than to sufficiently favorable ones ($s>\widehat{v}$).
In our setting, the reason is that unfavorable signals are perceived to be more likely veracious than  favorable signals (Proposition \ref{joint}). Put differently, investors are  more skeptical about favorable external news.  

Third, \emph{for a given threshold}, the nondisclosure price might exhibit a steep increase (Figure 3 panel a) or steep decrease (Figure 3 panel c). In particular, when 
the manager remains silent, the market price  is increasing in the signal but (except in the knife-edge case of $\widehat{v}=v^B$ illustrated in panel b) there is a steep change at $s=\widehat{v}$. This is because the investor beliefs are shifted when the signal reaches the disclosure threshold $\widehat{v}$.
As we show in Section 4, the steep change \emph{can never be upwards in equilibrium} (i.e., panel a cannot happen).\footnote{For some intuition, note that if the price were to sharply increase at a value where the manager were indifferent between disclosure decisions, then those to the right of this value would find it optimal to not disclose, just as the managers to the left of it, eliminating the possibility of a steep increase in equilibrium.} With a steep decline, favorable external news does not necessarily increase stock prices, i.e., the nondisclosure price is \emph{non-monotonic} in $s$.  As a result, in the neighborhood of the disclosure threshold, better external news leads to a lower nondisclosure price  if $\widehat{v}>v^B$ (as is the case of panel c). 
The potential for price non-monotonicity  arises because of the rational way in which the market makes inferences about the signal veracity and the manager's information endowment (Lemma \ref{beliefs} and Proposition \ref{joint}). To appreciate the importance of this updating, consider naive investors who, regardless of the external signal value and the manager silence, fix their joint beliefs at $\Pr(\phi)\Pr(\kappa\lvert \emptyset)$. Then, the nondisclosure price in \eqref{PNDgeneral5} would simply be given by 
 $
 \Pr(V )\Pr(U\lvert \emptyset) \cdot s + \Pr(V )\Pr(I\lvert \emptyset) \cdot s + \Pr(N )\Pr(U\lvert \emptyset) \cdot \mu +\Pr( N )\Pr(I\lvert \emptyset)  \cdot  \mathbb{E}[v\lvert v \leq \widehat{v}], 
$
 which is monotonic in $s$.\footnote{Within the confines of our model, nondisclosure prices  are discontinuous when $\widehat{v}\neq v^B$. This emerges because the CDF of $s$ conditional on $\phi=V$ is discontinuous at $s=v$. While other specifications of uncertain veracity may  ``smooth out'' the discreteness of the price function, they will nevertheless still involve a steep decrease or increase at $\widehat{v}\neq v^B$, with non-monotonicity when $\widehat{v}>v^B$ (Section  \ref{structuredisc}), which is one of the main insights of our work. Aside from being economically substantive, discontinuity and non-monotonicity usually  present technical difficulties with showing equilibrium existence (or fully characterizing equilibria). As we show in Section \ref{main}, a unique equilibrium in our game exists.}

\section{Disclosure Incentives} 
\label{main}

In Section \ref{sect:updating} we described the mechanics of belief updating and price characteristics for a given arbitrary threshold $\widehat{v}$. To do so, we conjectured that the manager's disclosure disclosure strategy follows a threshold rule. As we will show below, a unique threshold equilibrium indeed exists in each of our disclosure settings: early, late, and dynamic. 


\subsection{Early Disclosure}
\label{first}

We first consider the  case of early disclosure (superscript ``E") where the manager has to decide whether to disclose her private information at date 2, before the arrival of the external signal, and cannot delay the decision to date 4. Such case might arise, for instance, when the manager is slated for a prearranged conference call ahead of a media broadcast or the release of macroeconomic news. Similarly, it may emerge wherein external news is expected to arrive during a ``quiet period" in the lead-up to an Initial Public Offering (IPO) or to the close of a business quarter---under these circumstances awaiting to disclose after the external news  would breach the manager's legal obligation to remain silent during the ``quiet period."



A manager who is not endowed with information has no choice but to remain silent. An informed manager has to consider the market price following nondisclosure, $P(s, \emptyset)$.  Because this price depends on a signal that is not yet available at the time of the early disclosure decision, the manager has to form an expectation about it, conditional on the observed $v$. Note that $v$ carries information about the signal because the latter may, with some probability,  accurately reflect the firm value.
From the manager's perspective at date 2, the future external signal will be either veracious or non-veracious. Thus a straightforward way to think about the expected price is to present it as 
\[ \mathbb{E}[P(s,\emptyset)\lvert v] = \Pr(V \lvert   v) \cdot \mathbb{E}[P(s, \emptyset)\lvert V,v]+\Pr(N\lvert  v) \cdot \mathbb{E}[P(s, \emptyset)\lvert N,v].\]
 Because the signal is not realized yet, the observed  value $v$  carries no information about the signal veracity; thus $\Pr(\phi \lvert  v)=\Pr(\phi), \ \phi \in\{V,N\}$.  The manager expects that a veracious signal equals the observed value $v$. Hence, with probability $\Pr(V)=q$, the expected nondisclosure price is $\mathbb{E}[P(s, \emptyset)\lvert V,v]=P(s=v, \emptyset)$. A non-veracious signal equals $x$ so that, with probability $\Pr(N)=1-q$, the expected price is  $\mathbb{E}[P(s, \emptyset)\lvert N,v]=\mathbb{E}[P(s=x, \emptyset)]$.

When deciding whether to disclose, an informed manager compares her expectation of the nondisclosure price  with the price in case of disclosure. Because both prices depend on $v$, it is not immediately obvious that a threshold equilibrium exists. 
However, despite this endogeneity, the increasing difference property holds, conditional on the signal veracity. 

The expected nondisclosure price faced by the manager is increasing in the observed $v$ (via the expectation $\mathbb{E}[s\lvert v]$), up until the point at which the market expects disclosure to become more beneficial than keeping silent---at this point, the price sharply decreases, making disclosure even more attractive.

\begin{prop}\label{pr2} \emph{\textbf{(Early disclosure equilibrium)}}
In the early case there exists a unique disclosure equilibrium where the manager discloses at date 2 all values above a threshold $v^E\in (v^B, \mu)$. In the equilibrium, the date-5 nondisclosure price is non-monotonic in $s$ with a steep decline at $s=v^E$. 
\end{prop}
\begin{figure}[t]
\setlength{\unitlength}{.5cm}
\begin {picture}(11.5,11.5)\thicklines
\put(5,0){\includegraphics[width=0.6\textwidth]{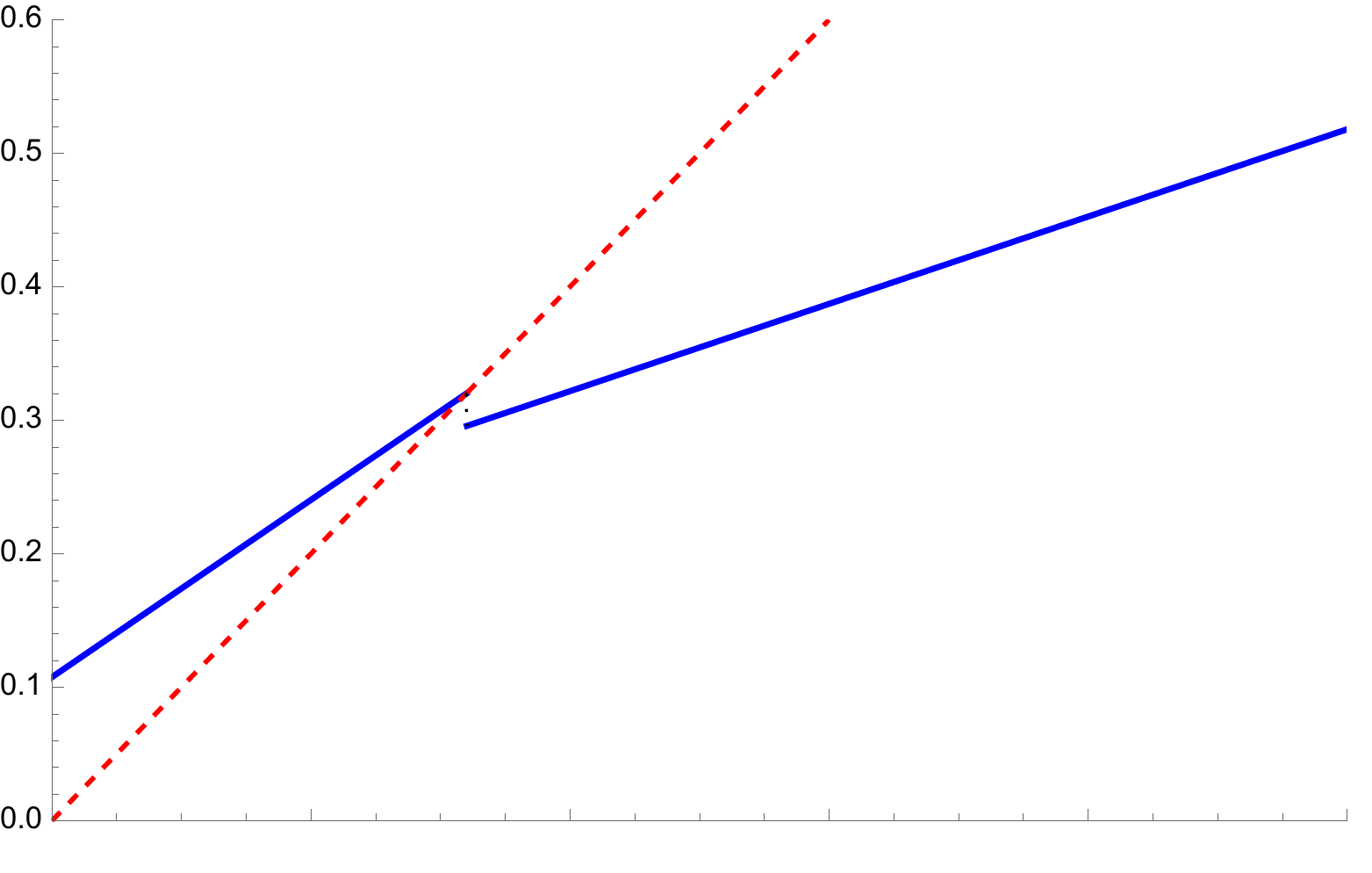}}
\put(24.5,-0.3){\small{Value $v$}}
\put(4,12.9){\small{Price}}
\put(11.1,-0.3){\small{$v^E>v^B$}} 
\put(14.1,13.1){\small{\color{red}$P(s,v)=v$}}
\put(25,11.1){\small{\color{blue}$\mathbb{E}[P(s,\emptyset)\lvert v ]$}}
\put(11.8,0.9){\color{gray}\dashbox{0.1}(0,6)}
\end{picture}
\begin{center}
\textbf{Figure 4}: Market price and equilibrium (early) disclosure threshold \\ {\footnotesize{Numerical example with uniform distributions, $v_{min}=0$, $v_{max}=1$, $p=0.90$ and $q=0.75$.
The disclosure price is illustrated with the dashed red line, and the expected nondisclosure price with the solid blue line.}}
\end{center}
\end{figure}

\ni 
 Managers  who withhold under the benchmark continue to withhold in the presence of external news. In addition, managers observing $v \in  [v^B, v^E]$ also withhold their information: we predict that external news of uncertain veracity crowds out voluntary disclosure compared with the benchmark. 
   Another way to interpret our result is to say that some managers are better off when their information is revealed by a potentially non-veracious external source rather than when they directly disclose it. 
This is driven by the investors' uncertainty about the veracity of external news and the manager's information endowment.
Specifically, if investors believe that the external signal is veracious, they react the same way as they would have if the manager had disclosed the information. However,  if they believe  that the external source is non-veracious,
the investors assign a higher likelihood that the manager is uninformed and hence put more weight on the prior expectation. 
Thus some managers  observing relatively low values (such as $v=v^B<\mu$) benefit from relying on the  external source to reveal those values. 

Further extending our observation that $v^E>v^B$, we confirm numerically for the uniform distribution that the early threshold is increasing in the prior probability that the signal is veracious (see Figure 5 for graphical illustration). Put differently, external signals that are more likely veracious crowd out managerial disclosures more intensely.


Our result that $v^E>v^B$, together with Proposition \ref{pricenon}, drives the non-monotonicity of the nondisclosure market price \emph{in equilibrium}. This non-monotonicity reflects the higher skepticism that favorable external signals are viewed with---that is, that they are taken ``with a grain of salt"---as the fact that the manager did not disclose means that the signals are more likely to reflect noise. 

\begin{figure}[t]
\setlength{\unitlength}{.5cm}
\begin {picture}(11.5,11.5)\thicklines
\put(5,0){\includegraphics[width=0.6\textwidth]{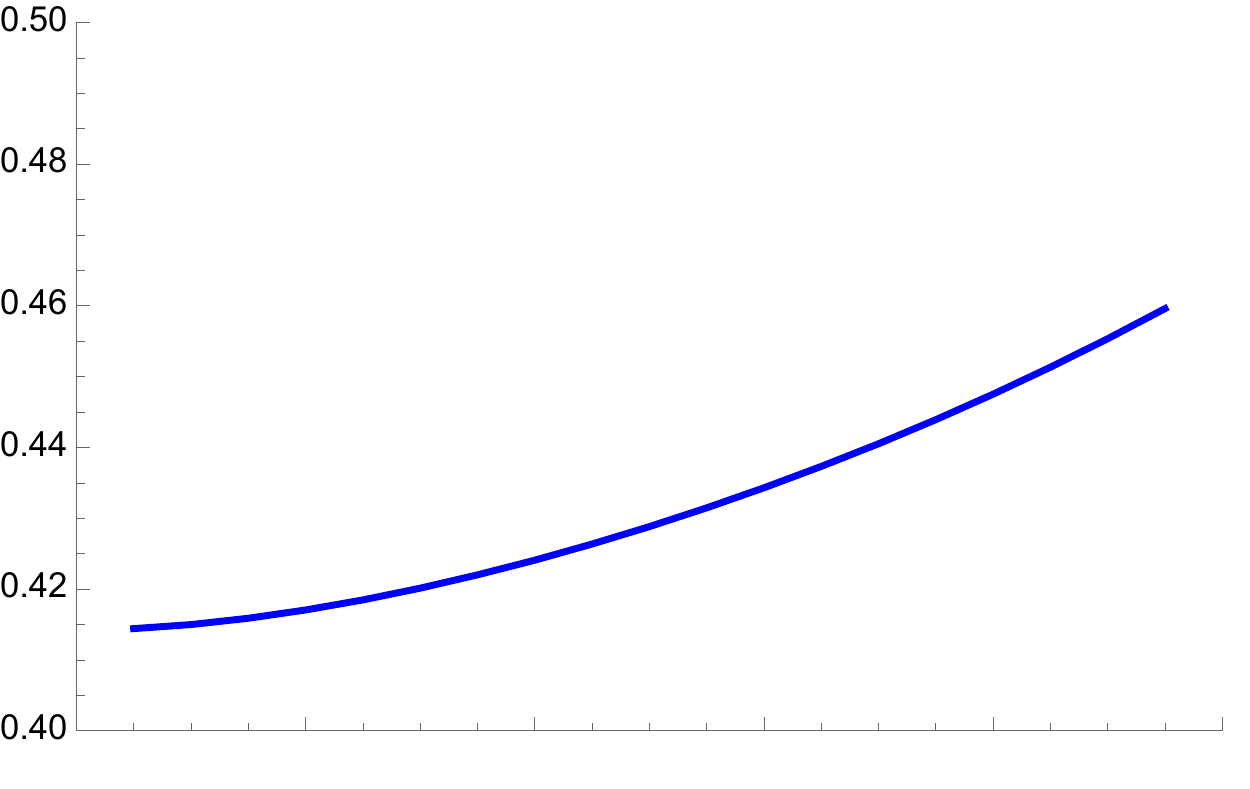}}
\put(24.5,0.3){\small{$q$}}
\put(3,13.3){\small{Early threshold}}
\put(24.4,8){\color{blue}\small{$v^E$}}

\end{picture}
\begin{center}
\textbf{Figure 5}: Early disclosure threshold as a function of $q$ \\ 
{\footnotesize{Numerical example with uniform distribution, $v_{min}=0$, $v_{max}=1$, $p=0.5$}}
\end{center}
\end{figure}

Before concluding this subsection, we comment on the assumption that an indifferent manager remains silent. This assumption plays no role in the proof of our results, but avoids the possibility of multiplicity. Notice that in Figure 4, when the manager is indifferent, the expected price function is left continuous because the market assumes indifferent managers do not disclose (as per our solution concept), so that the same disclosure behavior emerges at $v^{E}$, as well as for all $v < v^{E}$. If the indifferent manager \emph{discloses}, the expected price would be \emph{right}-continuous, and only intersect the 45 degree line \emph{after} the steep decline. This would yield a different price function. More generally, varying the tiebreaking probability (mixing strategy) would result in the 45-degree line being somewhere between the left limit of the function and the right limit of the function at $v^{E}$. Notably, however, the set of equilibria can be fully determined by assuming different tiebreaking rules, i.e., the equilibrium is unique as a function of the probability the indifferent manager discloses, and every equilibrium only involves one indifferent manager.\footnote{The following  simple argument suggests the left-continuous equilibrium is more compelling: Suppose the manager could, at time 2, change the value from $v$ to $v- \xi$, for some arbitrarily small $\xi$. Insisting that such arbitrarily small devaluations are not uniformly strictly profitable essentially amounts to a requirement that the price function is left continuous. This singles out the equilibrium  satisfying the property that an indifferent manager remains silent, as imposed by our model.}

\subsection{Late Disclosure}
\label{second}

We continue with late disclosure scenario (superscript ``L") in which the manager observes the external signal and only then decides whether to disclose her information (at date 4). This could happen for example when macroeconomic news, announcements, 
mainstream media articles and broadcasts, or social media posts are released unexpectedly or before the scheduled manager's conference call.\footnote{The manager may also observe the firm value only after the arrival of external news. The results in Section \ref{second} hold qualitatively under such alternative timeline.}

An informed manager can respond at date 4  by revealing the observed $v$ if the anticipated date-5 disclosure price  in \eqref{PD} exceeds the nondisclosure one  in Lemma \ref{price}.

\begin{prop}\label{lm4} \emph{\textbf{(Late disclosure equilibrium)}}
In the late case there exists a unique equilibrium where for any signal realization $s$ the manager discloses at date 4 if the observed value exceeds a threshold $v^{L}(s)<1$. This threshold is increasing in $q$ if $s > v^B$ but decreasing if $s < v^B$. Furthermore, if $s \lesseqgtr v^B$ then $v^L(s)\gtreqless s$ and $v^L(s)\lesseqgtr v^B$.
\end{prop}

\begin{figure}[t]
\setlength{\unitlength}{.5cm}
\begin {picture}(11.5,11.5)\thicklines
\put(5,0){\includegraphics[width=0.6\textwidth]{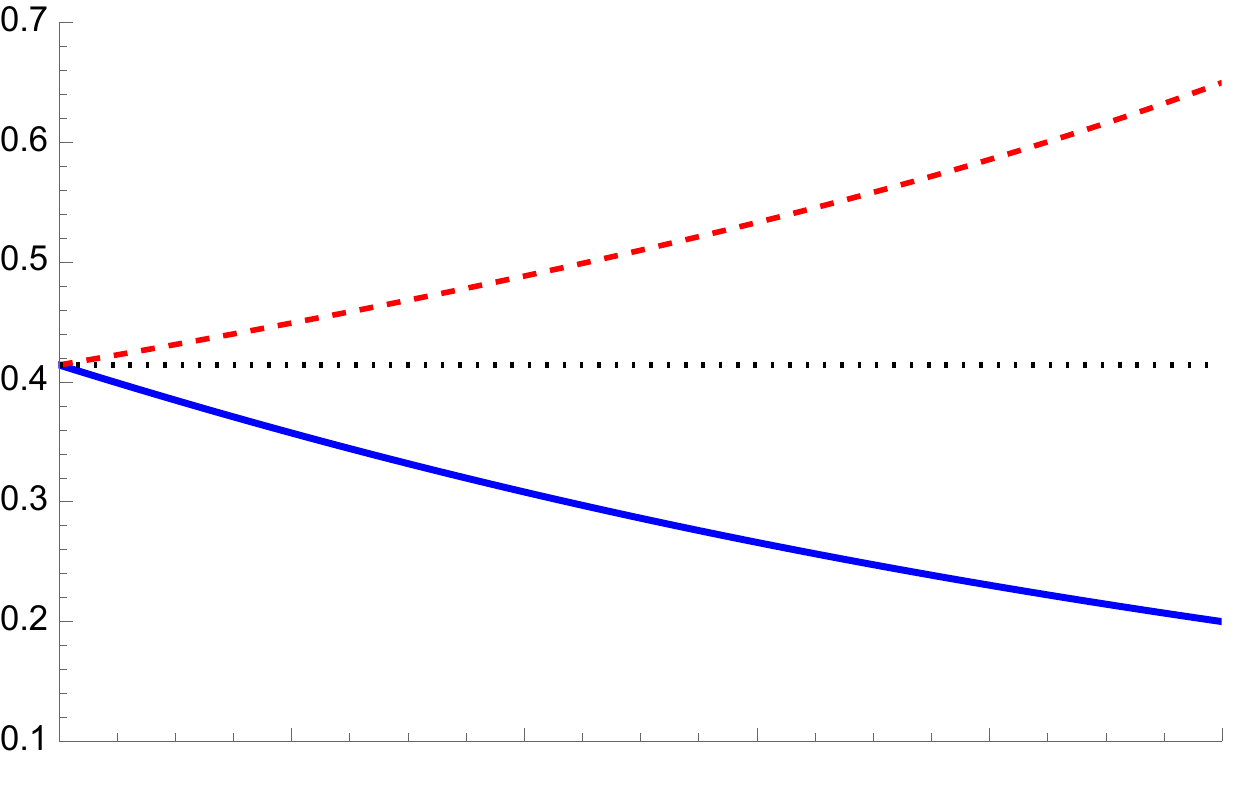}}
\put(24.5,-0.1){\small{$q$}}
\put(4,13.3){\small{Late threshold}}
\put(25.2,2.8){\color{blue}\small{$v^L(s=0.20)$}}
\put(25.2,6.8){\small{$v^L(s=v^B)$}}

\put(25.2,12){\color{red}\small{$v^L(s=0.65)$}}

\end{picture}
\begin{center}
\textbf{Figure 6}: Late disclosure threshold as a function of $q$ \\ 
{\footnotesize{Numerical example with uniform distribution, $v_{min}=0$, $v_{max}=1$,  $p=0.5$. Here, $v^B=0.42$.}}
\end{center}
\end{figure}

\ni To understand the role of the benchmark threshold $v^B$ in our result, note that with probability $q\in(0,1)$ the market expectation is a convex combination of the price when $s$ is veracious for sure, $\lim_{q\to 1} P(s,\emptyset)=s$, and when it is non-veracious for sure, $\lim_{q\to 0}  P(s,\emptyset)=P(\emptyset)=v^B$. This implies that when $s=v^B$, the  nondisclosure price (and thereby the late disclosure threshold) is \emph{independent} of the signal's perceived veracity, since it is equal to $v^B$ no matter what this conjecture is. It turns out that $s=v^B$ is the only signal with this property.

For any $s\neq v^B$ the effect of $q$ is non-trivial as graphically illustrated in Figure 6. In particular, if the signal is sufficiently favorable, $s>v^B$, higher $q$ implies that investors believe the signal is more likely veracious and put a higher weight on this more favorable signal, which increases the nondisclosure price and strengthens the incentives of the manager to remain silent, i.e., $v^L(s)$ increases. The opposite is true when the signal is sufficiently unfavorable, $s<v^B$. Then, higher $q$ means investors put a higher weight on this unfavorable signal which reduces the nondisclosure price and stimulates the manager to respond. That is, $v^L(s)$ decreases. Put differently, whether  external news encourages or discourages disclosure depends on $s$. A manager facing a sufficiently unfavorable (favorable) signal is more (less) likely to disclose compared with the benchmark case.\footnote{Only a signal $s=v^B$ has no effect on the probability of disclosure.} 

It seems intuitive that managers respond and ``correct" unfavorable external signals. Our result, however, shows that this intuition is not always true.  If the external signal is sufficiently low ($s < v^B$),  the disclosure threshold exceeds it ($v^L(s)>s$); i.e., the manager withholds values that are \emph{more favorable} than the ones revealed by the external signal.  Conversely, if the external signal is sufficiently high ($s > v^B$), the disclosure threshold falls short of it ($v^L(s)<s$); i.e., the manager discloses values that are \emph{less favorable} than the ones revealed by the external signal.  This result, graphically illustrated in Figure 7, may at first seem perplexing: why would anyone choose not to correct unfavorable news? The answer is simple: the manager discloses when the observed value exceeds the nondisclosure price. Because the latter is not identical to $s$ (recall that investors are skeptical about the signal's veracity and under-react to it), remaining silent in the face of some unfavorable $s$ and ``speaking up" to correct downwards some favorable external news  may be beneficial. 

\begin{figure}[t]
\setlength{\unitlength}{.5cm}
\begin {picture}(11.5,11.5)\thicklines
\put(5,0){\includegraphics[width=0.6\textwidth]{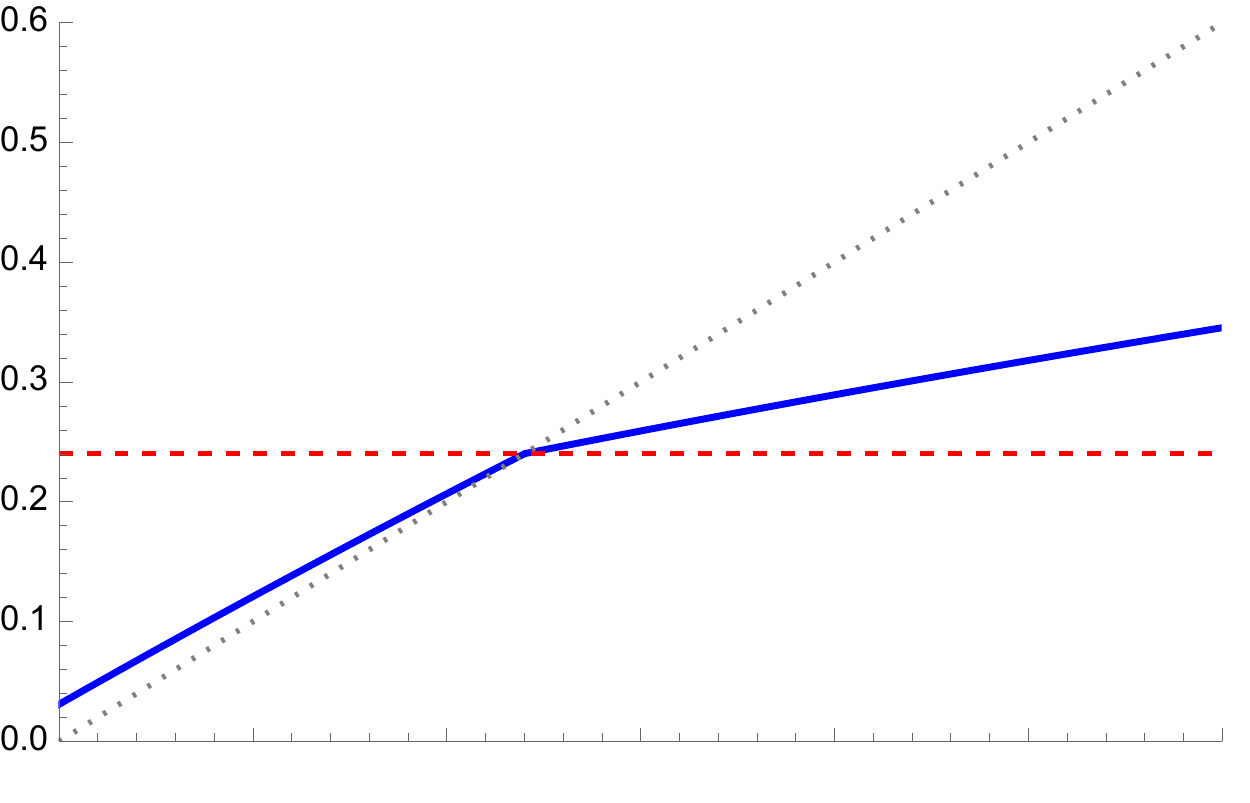}}
\put(24.5,-0.3){\small{Signal $s$}}
\put(4,13.3){\small{Threshold}}
\put(24.8,12){\color{gray}\small{$s$}}
\put(24.8,7.5){\color{blue}\small{$v^L(s)$}} 
\put(24.8,5.5){\color{red}\small{$v^B$}}

\end{picture}
\begin{center}
\textbf{Figure 7}: Late disclosure threshold as a function of $s$ \\ 
{\footnotesize{Numerical example with uniform distribution, $v_{min}=0$, $v_{max}=1$,  $p=0.9$ and $q=0.6$.}}
\end{center}
\end{figure}


\begin{prop} \emph{\textbf{(Late threshold)}}
\label{ssss}
In the late disclosure equilibrium, the disclosure threshold  has a kink at $s=v^B$ such that
$1>\frac{\partial}{\partial s}v^L(s \lvert s\leq v^B)>\frac{\partial}{\partial s}v^L(s \lvert s> v^B)>0$.  
\end{prop}

We find that $v^L(s)$ is more sensitive to unfavorable external news, $s<v^B$ than to favorable news (Figure 7). This is driven by the investors' skepticism about favorable signals that are discounted in price formation.\footnote{The equilibrium threshold is continuous in $s$ even though the nondisclosure price differs for favorable and unfavorable external news. The intuition for this result can be found in Proposition \ref{lm4}. The exact signal realization that distinguishes between the two pricing options is $s=v^B$. If the external signal is below this value, it is also above the late disclosure threshold. And vice versa, if $s$ is above this value it is below $v^L(s)$. At the critical value, it holds that $s=v^B=v^L(s)$ at which point the two pricing options coincide.}

\subsection{Dynamic Disclosure}
\label{freqd}

Sections \ref{first} and \ref{second} assumed that corporate disclosures are either early or late for exogenous reasons. We now consider a dynamic model in which the manager can choose the timing of her disclosure. Furthermore, we assume that advancing disclosure to an earlier date comes at rescheduling cost $c^E \geq 0$ (e.g., due to schedule constraints, reduced preparation time, reallocation of resources and manpower to meet accelerated timeline) and delaying it to a later date comes at rescheduling cost $c^L\geq 0$ (e.g., due to reputation loss for withholding the information early on, cancellation of previously scheduled announcement, or disruption of internal planning, decision-making processes and investor relations activities).
We first observe that the equilibrium disclosure continues to follow a threshold rule. 

\begin{definition} 
\label{def:threshold}
A \textbf{threshold equilibrium} in the dynamic game is a number $v^{E} \in (v^{B}, v_{max}]$ and an increasing, continuous function $v^{L}(s)$ such that the manager (i) discloses early if $v > v^{E}$ and (ii) discloses late if disclosure was not early and $v >v^{L}(s)$; otherwise, she remains silent.
\end{definition}

\begin{prop}  \emph{\textbf{(Dynamic disclosure equilibrium)}}
\label{thm:overallthreshold}
There exists a unique threshold equilibrium of the dynamic game.
\end{prop} 


\noindent The existence of the equilibrium is enabled by our observation that in the dynamic case the market price can never sharply increase in $s$:\footnote{As in prior work, the simultaneous presence of costs  and uncertain information endowment may result in multiplicity of equilibria (e.g., Bagnoli and Bergstrom 2005; Kartik, Lee and Suen 2019). This possibility is ruled out in our case because $g(\cdot)$ is from the log-concave family of distributions.} 

\begin{corollary}
\label{noup}
For threshold equilibria (as in Definition \ref{def:threshold}), the nondisclosure price may exhibit a steep decline at $s=\min\{v^E,v^L(s)\}$ but never a steep increase. 
\end{corollary}

The equilibrium in the dynamic case is more involved for several reasons. First, the manager must keep track of the second-period decision, which may be uncertain. Second, the manager has the option to remain silent and delay her decision. 
Importantly, this option value is non-constant in the manager’s value and itself is a function of the market’s equilibrium inference. As a result, an endogenous feedback between the disclosure decision and the relative benefit from nondisclosure arises. One complication this introduces is that we are not guaranteed any differentiability properties of the price function given an arbitrary conjecture regarding the manager's first period strategy, even though this does emerge in equilibrium (except for knife-edge cases). Thus, we use primitive arguments to show that the intuitions outlined in Section \ref{first} and Section \ref{second} hold generally, and that they then imply the increasing differences property.\footnote{In the proofs we show that the manager's payoff will satisfy an increasing difference property both in the case where the signal is veracious as well as in the case where it is not. This uses the fact that the \emph{expected} price given a non-veracious signal is independent of the manager's value. Despite the steep decline, the option value implied at this history does not increase at a rate faster than the value itself.}

In what follows, we briefly discuss several straightforward corollaries that consider the relative magnitude of rescheduling costs. We begin by 
 briefly noting how rescheduling costs affect disclosure incentives.
\begin{corollary}
\label{cost}
Suppose the thresholds $v^E$ and $v^L(s)$ are uniquely defined (as described in Proposition \ref{thm:overallthreshold}); then they are increasing in $c^{E}$ and $c^{L}$, respectively. There exist cutoffs $\overline{c}^j \in (0,v_{max}]$, $j=\{E,L\}$ such that, if the  costs exceed these cutoffs, disclosing early or late, respectively, is never beneficial. Cutoff $\overline{c}^j$ is increasing in $c^j$, $j=\{E,L\}$. 
\end{corollary}
\noindent The higher the rescheduling cost, the larger the set of values for which disclosure is not beneficial and the less negatively the market reacts to the manager's silence (see Verrecchia 1983 for similar intuition with disclosure costs). 
For the reminder of the paper we assume none of the costs is prohibitively large, $c^j<\overline{c}^j$ for $j=\{E,L\}$.

Consider a scenario in which delaying disclosure is  cheaper than advancing it, $c^{E}> c^{L}$ (e.g., because rescheduling earnings announcement to an earlier date  creates significant schedule complications, immensely reduces preparation time, requires very costly reallocation of resources and manpower to meet the accelerated timeline). 
 A manager who observes $v<v^{s=v_{max}}\equiv \min\{1, c^{E}+ P(s=v_{max},d=\emptyset)\}$ may benefit (but not lose) from delaying her decision: she obtains a market price of $v$ net of $c^{E}$ by disclosing immediately and a weakly higher price (as the external signal may be favorable and prompt the manager to remain silent) net of $c^{E}> c^{L}$ if she delays. A manager who observes $v > v^{s=v_{max}}$ knows that, if she does not disclose at date 2, she will certainly disclose at date 4, regardless of the external signal. Thus the manager may benefit from delaying as she faces a lower cost.
 
 \begin{corollary}
\label{Emore costly}
When 
$c^{E}> c^{L}$ the manager remains silent at date 2 and discloses at date 4 if $v>v^L(s)$. 
\end{corollary}

Now consider the case where delaying is costlier, $c^{E}< c^{L}$ (e.g., due to significant costs for cancellation of previously scheduled announcement, serious disruption of internal planning, decision-making processes, and investor relations activities, or damaging reputation loss for withholding the information early on). At date 2, when deciding whether to delay, an informed manager compares her payoff from sharing her information immediately, $P(s, v)-c^{E}=v-c^{E}$, with her expected payoff from delaying the decision to date 4, $\mathbb{E}[P(s, d^{E}=\emptyset, d^{L})\lvert v] - \Pr(d^{L}=v)c^{L}$, where  $d^{L}$ represents late disclosure and $d^{E}$ early disclosure.
\begin{corollary}
\label{Lmore costly}
When $c^{E}< c^{L}$ 
the manager discloses at date 2 if $v> v^E$ and, if the signal realization is such that $v^L(s)<v^E$,  discloses  $v\in (v^L(s), v^E)$ at date 4. 
\end{corollary}
\noindent Managers observing sufficiently high firm values prefer to disclose early and avoid the higher cost of delaying; the rest prefer to delay 
until after the arrival of the external signal.

The comparative statics of $v^L(s)$ and $v^E$ with respect to $q$ in Sections \ref{second} and \ref{first}, respectively, can shed light on how prior beliefs that external news is veracious impact the timing of corporate disclosures. In particular, external news that is ex ante more likely veracious (higher $q$) increases the managers' incentives to delay disclosure. The withheld values will be revealed later, unless the external news turns out to be sufficiently favorable.

\section{Extensions}

\subsection{Signal Structure}
\label{structuredisc}

We now show that the truth-or-noise signal structure is not crucial for our main results: prices continue to be non-monotonic when the veracious signal is not perfectly informative. 
Ultimately, what drives our key insights and sets them apart from those in prior literature 
is the \emph{players' uncertainty} about $\phi$. We elaborate on this point below.

\textbf{Non-monotonicity with uncertain veracity.}
Let $\phi \in \{\phi_{1}, \phi_{2}\}$, with $\Pr(\phi=\phi_{1})=q$ and $\Pr(\phi=\phi_{2})=1-q$.\footnote{A more general specification could allow for $\Phi$ to be the set of \emph{veracity states}. Given $\phi \in \Phi$ and $v$, the signal $s$ is distributed according to $H(\cdot \mid v, \phi)$ on the real line. 
Proof of price non-monotonicity under this more general structure is available upon request.}
Veracity $\phi_{2}$ implies the signal is uninformative ($s=x$), and veracity $\phi_1$ implies that the signal reflects the value with additively separable noise, $\varepsilon$, which is distributed according to a log-concave  distribution with mean zero and variance $1/\tau$. A possible interpretation of this structure is: external news that are not based on rumor are informative but not perfectly so. To simplify the analysis and the exposition in this section we assume that the supports of $v$ and $\varepsilon$ are infinite.\footnote{That is, we now let $v_{min}\to -\infty$ and $v_{max}\to +\infty$. If $v_{min}$ and $v_{max}$ were finite, signal realizations outside $[v_{min},v_{max}]$ would imply the veracity is $\phi_1$ with certainty.  In addition, signal realizations around the bounds would carry information about the realized noise $\varepsilon$. These are updating patterns that one needs to account for when constructing the posterior expectations. While our results hold also with finite support, we assume infinite support to streamline and simplify the exposition in this extension.\label{updating}}

Recall that price non-monotonicity arises when $\widehat{v}>v^B$. This property also continues to hold under the alternative structure. 

\begin{prop}\emph{\textbf{(Non-monotonicity with uncertain veracity)}}
\label{gennonmon}
 Under the signal structure specified in Section \ref{structuredisc} and for given threshold $\widehat{v}>v^B$, there exists a finite cutoff $\widehat{\tau}$ such that a sufficient condition for the nondisclosure price to be non-monotonic in the external signal $s$ is that $\tau>\widehat{\tau}$.
    \end{prop}
\noindent Intuitively, $\tau \to 0$ implies that signal with $\phi_1$ is very noisy and is ignored by the market. On the other hand, $\tau > \widehat{\tau}$ implies that the precision of that signal is sufficiently high. Thus, as before, an increase in the signal positivity may result in investors facing a silent manager to doubt more strongly that the signal veracity is $\phi_1$ and suspecting that its veracity is instead $\phi_2$. Thereby the market price may decrease. 
Figure 8 panel (a) graphically illustrates this. One can  see that the price  is non-monotonic because of a price decline at $\widehat{v}$. The higher the signal precision, the more pronounced is this non-monotonicity. In fact, when the precision is perfect (as in the case of the main model), there is an atom of the distribution at $s=v$ and the price becomes discontinuous as in Figure 3 panel (c). While with less-than-perfect precision the price discontinuity may smooth out, Proposition \ref{gennonmon} and Figure 9 (a)  show that the \emph{non-monotonicity persists} as long as the veracity of the signal is uncertain.


\begin{figure}[t]
\setlength{\unitlength}{.5cm}
\vspace{-0.6cm}
\begin {picture}(11.5,11.5)\thicklines
\put(0,0){\includegraphics[width=.43\textwidth]{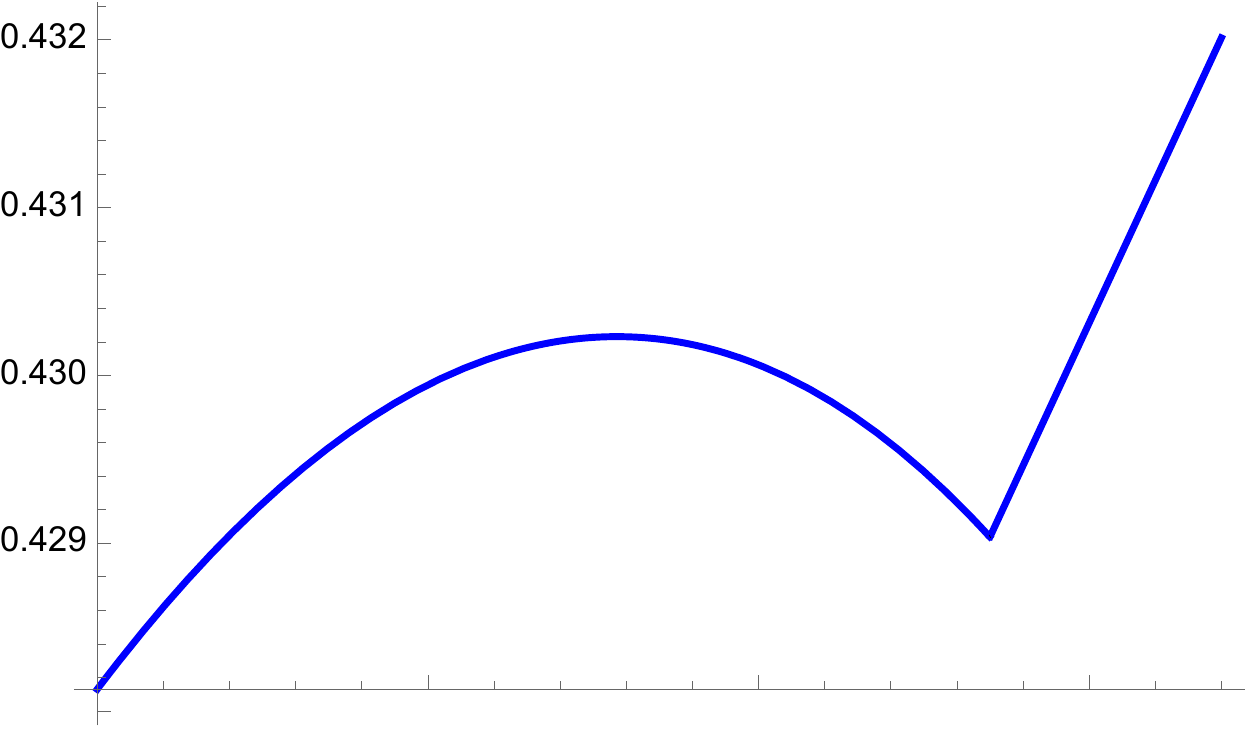}}
\put(17,0){\includegraphics[width=.43\textwidth]{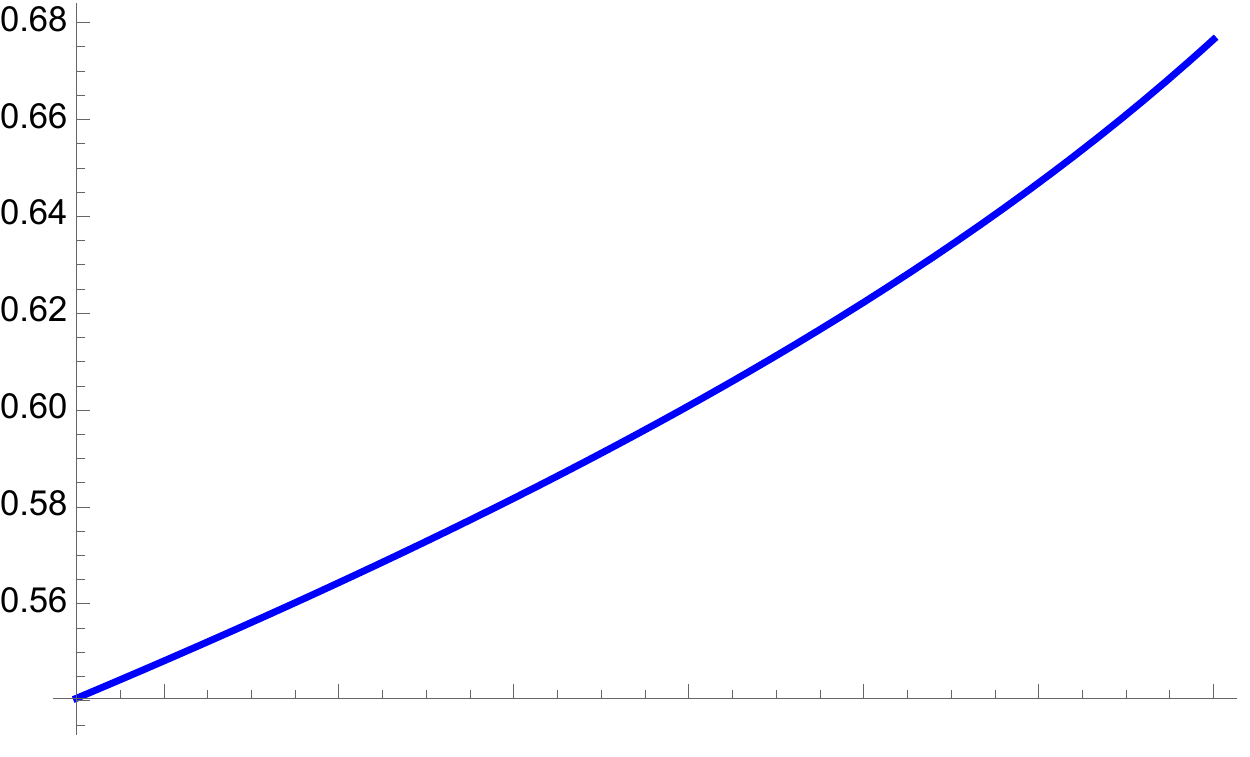}}
\put(14.5,-0.1){\footnotesize{$s$}}
\put(31.5,-0.1){\footnotesize{$s$}}
\put(0,9){\footnotesize{Price}}
\put(17,9){\footnotesize{Price}}
\put(2.3,-0.5){\footnotesize{(a) Illustration of Proposition \ref{gennonmon}}}
\put(2.4,-1.5){\footnotesize{$q=0.15$ so that $s=x$ \textbf{or} $s=v+\varepsilon$}}
\put(18.4,-0.5){\footnotesize{(b) Illustration of Proposition \ref{mon}}}
\put(17.4,-1.5){\footnotesize{$q=1$ so that $s=v+\varepsilon$ \textbf{with certainty}}}
\end{picture}
\bigskip
\vspace{0.5cm}
\begin{center}
\textbf{Figure 8}: 
Nondisclosure price for given $\widehat{v}$ as a function of the external signal $s$\\ {\footnotesize{Numerical example with $p=0.6$ and $\widehat{v}=0.41>0.39=v^B$. We let  $v\sim U[v_{min}=0, v_{max}=1]$, and $\varepsilon \sim U \left[-0.085, 0.085\right]$ so that $\tau=
\frac{12}{(0.085+0.085)^2}=415$, and adjust for the updating explained in footnote \ref{updating}.}}
\end{center}

\end{figure}



\textbf{Monotonicity with certain veracity.} 
Part of the prior disclosure literature (e.g., Acharya, DeMarzo and Kremer 2011; Menon 2020) assumes ``truth-with-noise" structures where the signal equals the true value plus a noise term.\footnote{In Acharya, DeMarzo and Kremer (2011) the manager's value and the external signal are related by the equation: $v=\mu(s)+\sigma(s)z,$, where $z$ is independent of $s$ and  $\mu$ is a strictly increasing deterministic function. For the most of their analysis, $\sigma(s)$ is constant,  equal to $\sigma$ for all $s$. This yields qualitatively similar results as assuming the signal equals for sure the true value plus a noise term. Proof of Proposition \ref{mon} under the exact specification in Acharya, DeMarzo and Kremer (2011) is available upon request.} One could view such structure as a special case of the one in this section where $q=1$, i.e.,  $\phi=\phi_1$ \emph{with certainty}. 
\begin{prop}\emph{\textbf{(Monotonicity with certain veracity)}}
\label{mon}
When the signal's veracity is $\phi=\phi_1$ with certainty so that it is known that $s=v+\varepsilon$, the nondisclosure price is monotonically increasing in $s$ for any given $\widehat{v}$.
\end{prop}
\ni The main driving force behind price monotonicity when $q=1$ is the lack of updating about $\phi$. Intuitively, when investors know that $\phi=\phi_1$ with certainty, there is no place for update in their beliefs about veracity, yielding the monotonically increasing property of the price. This result is graphically illustrated in Figure 8 (b). 

\subsection{Price Fluctuations with Frequent Adjustments} 
\label{pref}

Our model posits that the market price is formed only once, at date 5, after the external news and the manager's disclosure. In capital markets, however, prices are adjusted more frequently. 
This extension considers price adjustments on every date. To do so we return to the truth-or-noise signal structure and focus attention only to the early and late cases. Let $P^{j}_{t}(\Omega)$ be the price  at date $t$ under scenario $j=\{\text{E}, \text{L}\}$ and $\delta \in (0,1)$ be the discount factor faced by the manager.


\textbf{Late disclosure with frequently adjusted market price.} First, consider the late disclosure case. The starting price  is simply the prior expectation, $P_1^L=\mathbb{E}[v]=\mu$. There is no change until date 3, when the external signal is realized and the price  is adjusted to $P^L_3(s)=\mathbb{E}[v\lvert s]=q s+(1-q)\mu$. As one would expect, favorable external news increases the price, and unfavorable one decreases it; i.e.,
$
P^L_{1} \lesseqgtr P^L_3 (s) \quad \Leftrightarrow \quad s \gtreqless \mu.
$
At date 4, if the manager discloses, the price  is $P^L_4(s, v)=v$. 
Otherwise, it  is 
$
P^L_4(s,\emptyset)=P(s,\emptyset)
$. All future payments remain at their date-4 levels. Thus frequent price adjustment only results in scaling of the same payments that the manager had in Section \ref{second} and leaves the manager's decision unaffected.\footnote{When deciding whether to disclose, the manager compares the present value of her current and future payments from disclosure, $\Pi^L(v) = P^L_4(s,v) + \delta P^L_5(s,v) =(1+\delta)P(s,v)$, with those from nondisclosure, $\Pi^L(s, \emptyset) = P^L_4(s,\emptyset) + \delta P^L_5(s,\emptyset)=(1+\delta)P(s,\emptyset)$.}  

\begin{corollary}
\label{pr5}
If at date 4 the manager does not disclose, it holds that $P^L_4(s,\emptyset)<P^L_3(s)$ for any $s$. If at date 4 the manager discloses, it holds that  $P^L_4(s,v) \lesseqgtr P^L_3(s)$ if $s \gtreqless \frac{v-(1-q)\mu}{q}$.
\end{corollary}
\ni When the manager does not respond to the external signal, the price  decreases as investors conclude that the manager may have observed an even lower value. However, the price  may also decrease when the manager responds. 
This happens when the manager anticipates a further decrease in the price if she remains silent. To avoid this, the manager discloses values that, while lower than the date-3 price, exceed the price  that would have prevailed if she had remained silent.


\textbf{Early disclosure with frequently adjusted market price.} 
In the case of early disclosure it is straightforward that $P_1^E=\mathbb{E}[v]=\mu$ and $P_2^E (v)=v$. The nondisclosure price at date 2, $P_2^E(\emptyset)$, resembles the one defined in \eqref{PNDD} but for some conjectured by the investors threshold $\widehat{v}$ instead of $v^B$. (We derive the equilibrium threshold below.) At date 3, the newly arrived external signal changes the market price only if the manager remained silent at date 2, $P_3^E (s, v)=P_2^E(v)=v$ but $P_3^E(s,\emptyset)\neq P_2^E(\emptyset)$. Market prices at date 4 and 5 remain at their date-3 levels.

Do frequent price adjustments affect the manager's disclosure at date 2? An informed manager discloses if the present value of her payoffs from doing so
exceeds that from not.\footnote{These present values are $\Pi^E(v) = \mathbb{E}[P_2^E (v)\lvert v] + \delta \mathbb{E}[ P_3^E (s, v)  \lvert v] + \delta^2 \mathbb{E}[ P_4^E (s, v)  \lvert v] + \delta^3 \mathbb{E}[ P_5^E (s, v) \lvert v] 
= (1+\delta + \delta^2 + \delta^3)v$ and $\Pi^E(\emptyset)  = \mathbb{E}[P_2^E (\emptyset) \lvert v] + \delta \mathbb{E}[ P_3^E (s, \emptyset) \lvert v] + \delta^2 \mathbb{E}[ P_4^E (s, \emptyset) \lvert v] + \delta^3 \mathbb{E}[ P_5^E (s, \emptyset) v]
=  P_2^E(\emptyset) + (\delta + \delta^2 + \delta^3)\mathbb{E}[P(s,\emptyset) \lvert v]$.}
\begin{corollary}
    \label{freqearly}
When the market price is frequently adjusted, there exists a threshold $\widetilde{v}^E \in (v^B, v^E)$ such that the manager discloses (early) if and only if $v > \widetilde{v}^E$. The (early) disclosure threshold $ \widetilde{v}^E$ is increasing in the discount factor $\delta$.
\end{corollary}
\noindent Frequent adjustment of stock prices encourages corporate disclosure, $\widetilde{v}^E < v^E$, because the manager faces a lower short-term price at date-2, relative to the expected market price after the signal is realized---this makes pretending to be uninformed less beneficial (in present value terms) and encourages disclosure. Nevertheless, external news continues to suppress corporate disclosure, $\widetilde{v}^E > v^B$. Since the disclosure threshold is increasing in $\delta$, more impatient managers are more likely to disclose their firm value.\footnote{Extremely impatient managers ($\delta=0$) disclose only  values above $v^B$: they only care about the current period and not about the external news that will reveal in the next period.  Extremely patient managers ($\delta=1$), however, disclose only values above $v^E$.}

It remains to consider the price trend over time. The date-2 disclosure price  $P_2^E (v)=v$ exceeds  the initial price $P_1^E=\mathbb{E}[v]=\mu$ if $v>\mu$. After disclosure of  $v\in [\widetilde{v}^E, \mu)$, the price can decrease. 
This is because the manager expects that the present value of the future price post-silence to be lower. If the manager withholds her information,  the date-2  market price always decreases from its initial level, $P_2^E(\emptyset) <P_1^E$.\footnote{To see why, note that $P_2^E(\emptyset)= P(\emptyset \lvert \widetilde{v}^E)<P(\emptyset \lvert v^B)=v^B<\mu=P_1^E$.} The realization of the external signal affects the date-3 market price in a nontrivial way.
\begin{corollary}
\label{comparisonprice1}
The price following disclosure at date 2, $P_2^E (v)$, is lower than the initial price, $P_1^E$, if the observed and disclosed by the manager value is $v \in [\widetilde{v}^E, \mu)$. 
 There exists $s^{\dag \dag} \in (v_{min},\widetilde{v}^E)$ such that:
 $P_2^E(\emptyset) \geq P_3^E(s, \emptyset)$  if $s\in [v_{min}, s^{\dag \dag}]$ or $s\in [\widetilde{v}^E, \mu]$. However, $P_2^E(\emptyset) \leq P_3^E(s, \emptyset)$  if $s\in [s^{\dag \dag},\widetilde{v}^E]$ or $s\in [\mu, v_{max}]$.
\end{corollary}

\ni  Very favorable external news increases the market price, relative to its date-2 level, and extremely unfavorable news decreases the price. However, our result implies that more neutral news ($s$ in the intermediate region) has an ambiguous effect. While the date-2 nondisclosure price $P_2^E(\emptyset) = P(\emptyset)$ does not depend on the signal, the date-3 price, $P_3^E(s, \emptyset)=P(s,\emptyset)$, is non-monotonic in it (Proposition \ref{pricenon}). Because news that exceeds the disclosure threshold is perceived to be less likely veracious (Proposition \ref{joint}), the market price around $s=\widetilde{v}^E$ drops in a sense that $\lim_{s\to \widetilde{v}^{\EA -}}  P_3^E(s, \emptyset)>\lim_{s\to \widetilde{v}^{\EA +}}  P_3^E(s, \emptyset)$. A signal $s\to \widetilde{v}^{\EA+}$ leads to price decrease from its date-2 level. Conversely, a signal $s\to \widetilde{v}^{\EA-}$ leads to price increase.

\section{Empirical Implications}

Our work has several practical and  empirically testable implications. First, consistent with empirical evidence in Zhang (2006a) and Zhang (2006b), we demonstrate that investors  underreact to external news with uncertain veracity. More importantly, this underreaction exhibits asymmetry: investors are more sensitive to unfavorable news due to their perception of it as more likely to be veracious. The existence of asymmetric price response to external news is empirically documented across diverse settings, ranging from announcements about US trade balances (Aggarwal and Schirm 1998) and macroeconomic news (Goldberg and Grisse 2013) to Federal Open Market Committee decisions (Blot, Hubert, and Labondance 2020), peers' clinical trial outcomes (Capkun, Lou, Otto, and Wang 2022), and Initial Jobless Claims (Xu and You 2023). Although these empirical studies do not address whether the observed phenomenon emerges due to the economic forces in our model---specifically, the nuanced updating of investor beliefs---future research could incorporate  observable proxies for uncertainty in veracity and beliefs to examine this relationship. Additionally, a burgeoning strand of literature explores the asset pricing implications of rumors about companies (e.g., Alperovich, Cumming, Czellar and Groh 2021; Cai, Quan and Zhu 2023). To our knowledge, none of these studies has assessed how the studied effects  vary in the content and positivity of the rumors. We consider this aspect a prospective avenue for future empirical research.

%

Second, we predict that  prices of  nondisclosing firms can be non-monotonic in the sense that an increase in the positivity of external news  may paradoxically lead to a lower market valuation. To the best of our knowledge, this prediction has not yet been empirically documented but could be tested, after controlling for management nondisclosure and firm-related observable characteristics.

Third, we predict that the presence of external news with uncertain veracity reduces the probability of disclosure by firms prior to entering a so-called ``quiet period" in the lead-up to an Initial Public Offering (IPO) or a business quarter end.  We expect this crowding out effect to be stronger whenever the external news is ex ante more likely veracious. Future empirical research may account for the prior probability of news to be accurate (e.g., news originating from social media is less likely veracious than a macroeconomic announcement; prior evidence of spreading false information by a specific mainstream media channel might be indicative as well) and estimate the magnitude of the identified crowding out effect.

Fourth, our predictions about the disclosure behavior of firms that \emph{respond to already released} 
external news differ substantially: Unfavorable (favorable) external news increases (decreases) the probability of corporate disclosure---these effects are exacerbated when the  likelihood of external news to be veracious is higher. Our results are consistent with  Capkun, Lou, Otto, and Wang (2022) who document that the likelihood of voluntary disclosure depends on the positivity of the news about medical trial outcomes conveyed by peers. Moreover, we find that when the external news is sufficiently unfavorable managers may disclose information that is worse than the one revealed in the external news. 
This result is consistent with findings in Sletten (2012) who notes that disclosing bad news after peers' restatement lacks an underlying theory. We hope our model fills this gap.

Fifth, when given the option to choose the timing of their disclosure, firms opt to wait for the release of external news, unless delaying disclosure is  costlier than advancing it (e.g., because cancelling a previously scheduled conference call creates scheduling issues, there is a significant disruption of internal planning and investor relation activities, or the reputation loss for withholding information early on is very high). In that latter case, we predict that good private information is released before the external news and more neutral after that. We also expect that higher likelihood for veracity exacerbates the incentives to delay disclosures. Our result  aligns with recent empirical evidence about clustering of    managerial guidance with good   information before the release of restatements by peers and guidance with bad information after the release  (Sletten 2012).\footnote{For a related result with own-firm information event (e.g., earnings release) rather than external news, see Roychowdhury and Sletten (2012).} 
%
%
Lastly, consistent with Sletten (2012) we also find that frequently adjusted market prices may decrease after corporate disclosures. 

\section{Conclusion}
\label{conclusion}

We examine how uncertainty surrounding the veracity of external news affects investor beliefs, managers' incentives to reveal private information, and market prices. We find that investors perceive favorable external news as less likely veracious, which reinforces their beliefs that managers are hiding low firm values. As a result, there is a possibility of price non-monotonicity where relatively more favorable external news could paradoxically lead to a lower market price. We also find that the timing, likelihood for veracity and content of external news can either encourage or discourage corporate disclosures.

Our analysis has practical implications for investors and firms  deciding whether to disclose private information in an environment with abundant external news with uncertain veracity like traditional media broadcasts, social media posts and tweets etc. The results may provide insights into the findings of recent empirical work in disclosure (e.g., Sletten 2012; Capkun, Lou, Otto, and Wan 2022). Furthermore, our research suggests that market prices may have non-monotonic behavior and non-constant sensitivity to external news---a finding that may be related to empirical asset pricing studies (e.g., Aggarwal and Schirm 1998; Goldberg and Grisse 2013; Blot, Hubert and Labondance 2020; Xu and You 2023). Future empirical research may account for the uncertainty in the veracity of external news when investigating firm disclosure behavior and price reactions.

\onehalfspacing
\newpage
\appendix

\renewcommand{\thelemma}{\thesection.\arabic{lemma}}
\setcounter{lemma}{0}

\renewcommand{\theprop}{\thesection.\arabic{prop}}
\setcounter{prop}{0}


\section{Proofs}

\ni \textbf{Proof of Lemma \ref{lm1}:} Follows directly from the proof in Dye (1985) and Jung and Kwon (1998) and is omitted.

\bs  \ni \textbf{Proof of Lemma \ref{beliefs}:} We consider the joint belief $\Pr(\phi, \kappa \lvert  s, \emptyset)$ separately  for the cases $s<\widehat{v}$ and $s>\widehat{v}$.

\bigskip 
\noindent \underline{Case 1:}
Let the observed signal be $s=s'\in [v_{min},\widehat{v})$. Because $s$ is a continuous variable, we use the limit of measurable intervals  and express 
$\Pr(\phi, \kappa\lvert  s=s', \emptyset)= \lim_{\Delta\to 0} \frac{ \Pr(s\in[s',s'+\Delta],\emptyset,\phi, \kappa)}{\Pr(s\in[s',s'+\Delta], \emptyset)}$ for $\phi \in \{N, V\}$, $\kappa \in \{I,U\}$ and $\Delta>0$ sufficiently small so that $s'+\Delta\in [v_{min},\widehat{v})$. Let us define $A' \equiv [s', s'+\Delta]$ and focus on the numerator. 
We note that:
\bean
\Pr(s\in A',\emptyset,V, I)&=& p \cdot q  \cdot 1
\cdot \Pr(s \in A')  ;
\\
\Pr(s\in A',\emptyset,V, U) &=&(1-p) \cdot 1 \cdot q \cdot \Pr(s \in A')  ;
\\
\Pr(s\in A',\emptyset,N, I) &=& p \cdot G(\widehat{v})  \cdot  (1-q) \cdot  \Pr(s \in A') ;
\\
\Pr(s\in A',\emptyset,N,U) &=&
(1-p) \cdot (1-q) \cdot \Pr(s \in A')    .
\eean
For the denominator in $\Pr(\phi, \kappa\lvert  s=s', \emptyset)$, we have
$
\Pr(s\in A', \emptyset)
= \Pr(s\in A', \emptyset,V, I)
+\Pr(s\in A',\emptyset,V, U) 
+\Pr(s\in A',\emptyset,N, I)+\Pr(s\in A',\emptyset,N, U)=\Big(q+(1-q)\big(1-p+pG(\widehat{v})\big)\Big) \cdot \Pr(s \in A').
$
Now we can express:
\bean
\Pr(U,V \lvert s=s',\emptyset)&=&\lim_{\Delta\to 0} \frac{\Pr(s\in A',\emptyset,V, U)}{\Pr(s\in A',\emptyset)}
\\
&=& \frac{(1-p)q}{q+(1-q)\big(1-p+pG(\widehat{v})\big)} \cdot \lim_{\Delta \to 0} \frac{\hcancel{\Pr(s \in A')}}{\hcancel{ \Pr(s \in A')}};
\\
\Pr(U,N\lvert s=s',\emptyset)&=&\lim_{\Delta\to 0} \frac{\Pr(s\in A',\emptyset,N, U)}{\Pr(s\in A',\emptyset)}
\\
&=& \frac{(1-p)(1-q)}{q+(1-q)\big(1-p+pG(\widehat{v})\big)} \cdot \lim_{\Delta \to 0} \frac{\hcancel{\Pr(s \in A')}}{\hcancel{ \Pr(s \in A')}};
\\
\Pr(I, V \lvert s=s', \emptyset)&=& \lim_{\Delta\to 0}\frac{\Pr(s\in A', \emptyset,V, I)}{\Pr(s\in A',\emptyset)}
\\
&=& \frac{p q }{q+(1-q)\big(1-p+pG(\widehat{v})\big)} \cdot \lim_{\Delta \to 0} \frac{\hcancel{\Pr(s \in A')}}{\hcancel{ \Pr(s \in A')}};
\\
\Pr(I,N\lvert s=s',\emptyset)&=&\lim_{\Delta\to 0} \frac{\Pr(s\in A',\emptyset,N, I)}{\Pr(s\in A',\emptyset)}
\\
&=& \frac{p  G(\widehat{v}) (1- q)}{q+(1-q)\big(1-p+pG(\widehat{v})\big)} \cdot \lim_{\Delta \to 0} \frac{\hcancel{\Pr(s \in A')}}{\hcancel{ \Pr(s \in A')}} .
\eean

\smallskip 
\noindent \underline{Case 2:} 
Let  the observed signal be $s=s''\in(\widehat{v},v_{max}]$.  Defining $A''\equiv [s'',s''+\Delta]$ for $\Delta>0$ sufficiently small, so that $s''+\Delta\in(\widehat{v},v_{max}]$, and following similar steps, 
\bean
\Pr(s\in A'',\emptyset,V, I)&=&0 \cdot  p \cdot G(\widehat{v}) \cdot \Pr(s\in A'');
\\
\Pr(s\in A'',\emptyset,V, U)&=&(1-p) \cdot q\cdot \Pr(s\in A'')  ;
\\
\Pr(s\in A'',\emptyset,N,I)&=& p \cdot G(\widehat{v}) \cdot (1-q)  \cdot \Pr(s\in A'');
\\
\Pr(s\in A'',\emptyset,N,U)&=& (1-p) \cdot (1-q) \cdot  \Pr(s\in A'') .
\eean
Summing up, we have
$\Pr(s\in A'',\emptyset)=\Big(q(1-p)+ (1-q)\big((1-p)+p\cdot G(\widehat{v})\big)\Big) \cdot \Pr(s\in A'')$. 
We can express
\bean
\Pr(U,V \lvert s=s'',\emptyset)&=&\frac{(1-p)q}{q(1-p)+ (1-q)\big((1-p)+p\cdot G(\widehat{v})\big)} \cdot \lim_{\Delta \to 0} \frac{\hcancel{\Pr(s \in A'')}}{\hcancel{ \Pr(s \in A'')}};
\\
\Pr(U,N\lvert s=s'',\emptyset)&=&\frac{(1-p)(1-q)}{q(1-p)+ (1-q)\big((1-p)+p\cdot G(\widehat{v})\big)} \cdot \lim_{\Delta \to 0} \frac{\hcancel{\Pr(s \in A'')}}{\hcancel{ \Pr(s \in A'')}};
\\
\Pr(I, V \lvert s=s'', \emptyset)&=&0;
\\
\Pr(I,N\lvert s=s'',\emptyset)&=&\frac{p G(\widehat{v})(1-q)}{q(1-p)+ (1-q)\big((1-p)+p\cdot G(\widehat{v})\big)} \cdot \lim_{\Delta \to 0} \frac{\hcancel{\Pr(s \in A'')}}{\hcancel{ \Pr(s \in A'')}}. 
\eean
Summarizing yields our result.

We briefly mention that the same argument applies in cases where the first-period disclosure strategy is not characterized by a threshold, simply by modifying the distribution $G$ appropriately.

\bs  \ni \textbf{Proof of Proposition \ref{joint}:} Follows from Lemma \ref{joint}.

\bs \ni \textbf{Proof of Lemma \ref{price}:} 
Follows directly from equation \eqref{PNDgeneral5} and Lemma \ref{beliefs}.

\bs \ni \textbf{Proof of Proposition \ref{sensitivity}:} Follows from Lemma \ref{price}.

\bs \ni \textbf{Proof of Proposition \ref{pricenon}:} Follows from Lemma \ref{price}.


\bs  \ni \textbf{Proof of Proposition \ref{pr2}:} 
The proof that $v^E \geq v^B$ follows from Lemma \ref{lemm:nojumpup} (included in the technical appendix B) and Proposition \ref{pricenon}. To see this note that if $P(s,\emptyset)$ has a discontinuous jump downwards at some $s=\widehat{v}>v^B$, then $\mathbb{E}[P(s,\emptyset)\lvert v]$ has a discontinuous jump downwards at $v=\widehat{v}>v^B$.  We now argue this threshold must be strictly above $v^B$. 

Slightly abusing notation, let $P^{q}(s, v=\widehat{v})$ denote the price when the market conjectures threshold $\widehat{v}$,  given veracity $q$.  We first claim that, at $\widehat{v}=v^{B}$, this function is strictly convex in $q$,  at any $s \neq v^B$.  We also claim that $\frac{d}{dq} \int_{v_{min}}^{v_{max}} P^{q}(s, v=\widehat{v}) g(s) ds$ is equal to 0 at $q=0$.

Now,  note that $P^{q}(v^B, \widehat{v}=v^B )=v^B$.  Also note that $\int_{v_{min}}^{v_{max}} P^{0}(s, \widehat{v}=v^B )g(s)ds=v^B$,  and that putting the two claims above together,  we also have that $\int_{v_{min}}^{v_{max}} P^{q}(s, \widehat{v}=v^B )g(s)ds$ is strictly increasing in $q$,   for all $q > 0$; indeed, for almost every $s$,   $P^{q}(s, \widehat{v}=v^B )$ is strictly convex,  meaning that the expectation over $s$ is strictly convex as well.  Since the derivative is 0 at $q=0$,  we have this function is strictly increasing.  

So consider the manager's incentives.  If the market conjectured a threshold $\widehat{v} =v^B$,  the manager at the threshold would obtain $v^B$ from disclosing,  but: 
\begin{equation*} q \cdot P^{q}(v^B, \widehat{v}=v^B )+(1-q) \int_{v_{min}}^{v_{max}} P^{q}(s, \widehat{v}=v^B )g(s)ds \end{equation*} 
from not disclosing.  However, we have just seen that this is a convex combination of $v^B$ and a term which is strictly larger than $v^B$.  Thus, the manager strictly gains from not disclosing,  showing that we must have $\widehat{v} > v^B$.

To compare $v^E$ and $\mu$, we use that $\lim\limits_{\widehat{v}\to \mu}\mathbb{E}_{s}[P(s,v=\widehat{v})]=\mu$ while $\lim\limits_{\widehat{v}\to\mu}\mathbb{E}_{s}[P(s,\emptyset) \lvert v]<\mu$; to see this point, we analyze the price as a function of the signal:

\begin{equation*} 
P(s, \emptyset)=\frac{q \cdot s(1-p\mathbbm{1}_{s > \widehat{v}})+ (1-q) ((1-p)\mu + p \int_{v_{min}}^{\widehat{v}} v g(v) dv)}{q(1-p\mathbbm{1}_{s > \widehat{v}})+(1-q)(1-p+pG(\widehat{v}))}. 
\end{equation*}

\noindent Note that, if the signal is veracious and the manager's value is $\widehat{v}$, then $(1-p)\mu + p \int_{v_{min}}^{\widehat{v}} v g(v) dv) < \widehat{v}$ whenever $\widehat{v} > v^B$. Thus, at $\widehat{v}=\mu$, the price is less than $\mu$ if the signal is veracious. On the other hand, when $s \sim G$: if $s > \widehat{v}$, then $(1-p)\mu + p \int_{v_{min}}^{\widehat{v}} vg(v)dv < \widehat{v} < s$, meaning that the price would be even larger if the market conjectured such signals could have come from informed managers; that is, 

\begin{equation*} 
P(s, \emptyset) \leq \frac{q\cdot s+ (1-q) ((1-p)\mu + p \int_{v_{min}}^{\widehat{v}} v g(v) dv)}{q(1-p)+(1-q)(1-p+pG(\widehat{v}))}
\end{equation*} 
Now, when we take the expectation of the right hand side of this expression, we have: 

\begin{equation*} 
\frac{q\cdot \mu+ (1-q) ((1-p)\mu + p \int_{v_{min}}^{\widehat{v}} v g(v) dv)}{q(1-p)+(1-q)(1-p+pG(\widehat{v}))} < \mu, \forall \widehat{v}.
\end{equation*}

\noindent Thus, we see that if $\widehat{v} = \mu$, the price is (bounded above by) a convex combination of terms that are equal to $\mu$ and those that are a strictly less than $\mu$, as claimed.

We now return to the two claims above regarding $P^{q}(s, v= \widehat{v})$. The claim on strict convexity involves some tedious algebra, which reveals that $\frac{\partial^{2}}{\partial q^{2}} P^{q}(s, \widehat{v}) > 0$ if and only if: 

\begin{equation*} 
(\mathbbm{1}_{s > \widehat{v}}-1 + G(\widehat{v}))((1-p)(s-\mu) +p \int_{v_{min}}^{\widehat{v}}(s-v)g(v)dv) >0.
\end{equation*}

\noindent We note that: 

\begin{equation*}
0=(1-p)(s-\mu) +p \int_{v_{min}}^{\widehat{v}}(s-v)g(v)dv)  \Leftrightarrow s= \frac{\mu(1-p)+ p \int_{v_{min}}^{\widehat{v}}vg(v)dv}{1-p+p G(\widehat{v})}.
\end{equation*}

\noindent If $\widehat{v}=v^{B}$, then this holds if and only if $s=v^{B}$. So, this term is negative if $s < v^B$ and positive if $s > v^B$. Similarly, $\mathbbm{1}_{s > v^B}-1 + G(v^B)$ is positive for $s > v^B$ and negative for $s < v^B$. Putting this together,  strict convexity holds. 

As for the second claim, we note that we can bring the derivative inside the expectation: 

\begin{multline*}  
\frac{\partial}{\partial q} \biggl\lvert_{q=0}\int_{v_{min}}^{v_{max}}P^{q}(s, v=\widehat{v})g(s)ds \propto  \int_{v_{min}}^{v_{max}}  \biggl((1-p+pG(\widehat{v}))(s(1-p\mathbbm{1}_{s > \widehat{v}})-(1-p)\mu-p \int_{v_{min}}^{\widehat{v}}vg(v)dv)\\ -((1-p)\mu+p \int_{v_{min}}^{\widehat{v}} vg(v)dv)(1-p\mathbbm{1}_{s > \widehat{v}}-((1-p)+pG(\widehat{v}))) \biggr)g(s)ds.
\end{multline*}

\noindent This can be obtained by noting that the denominator of $\frac{\partial}{\partial q} \biggl\lvert_{q=0} P^{q}(s, v=\widehat{v})$ is $(1-p+pG(\widehat{v}))^{2}$ (and hence independent of $s$, meaning it can be dropped without changing the sign of the derivative); the value of the denominator of $P^{q}(s, v= \widehat{v})$ at $q=0$ is $1-p+pG(\widehat{v})$; the derivative of the denominator with respect to $q$ evaluated at $q=0$ is $1-p \mathbbm{1}_{s> \widehat{v}}-(1-p+pG(\widehat{v}))$; the value of the numerator at $q=0$ is $(1-p) \mu+ p \int_{v_{min}}^{\widehat{v}}vg(v)dv$; the derivative at $q=0$ is $s(1-p\mathbbm{1}_{s > \widehat{v}})-(1-p)\mu-p \int_{v_{min}}^{\widehat{v}}vg(v)dv$.  That this expression is equal to 0 follows from the observation that: 

\begin{multline*}
\int_{v_{min}}^{v_{max}} \left(s(1- p \mathbbm{1}_{s > \widehat{v}})-(1-p)\mu- p \int_{v_{min}}^{\widehat{v}}vg(v)dv \right)g(s)ds=0\\=\int_{v_{min}}^{v_{max}} \left((1- p \mathbbm{1}_{s > \widehat{v}}-((1-p)+ p G(\widehat{v}) )\right)g(s)ds, 
\end{multline*}

\noindent completing the proof of the second claim; to see these equalities, note that $1-p\mathbbm{1}_{s>\widehat{v}}=1-p+p\mathbbm{1}_{s\leq \widehat{v}}$.

\bs  \ni \textbf{Proof of Proposition \ref{lm4}:} 
We consider the comparative static, given the threshold equilibrium. Note that in the late case we have the following equation (which also uses the fact that $s$ provides no information about $v$ if non-veracious and is equal to $v$ if veracious): 

\begin{equation*} 
v^{L}(s)  = s \Pr(V \lvert s, \emptyset)+ \mathbb{E}[v \lvert N, \emptyset] \Pr (N \lvert s, \emptyset)
\end{equation*}

\noindent Suppose that $s > v^B$. Then by definition of $v^B$ and Lemma \ref{lemm:thresholdcompstat} (included in the technical appendix B), values that equal to veracious signals are disclosed, i.e., $v=s > v^{L}(s)$. Thus, it must be that $s > \mathbb{E}[v \lvert N, \emptyset]$, since otherwise a manager with value equal to $s$ would have a profitable deviation to remain silent. Since increasing $q$ increases the weight on $s$ and decreases the weight on $\mathbb{E}[v \lvert N, \emptyset]$, we have that the right hand side increases, meaning that $v^{L}(s)$ increases as well. 

On the other hand, if $s < v^B$, the opposite conclusions hold, implying that as the weight on $s$ increases, the threshold decreases. 

Lemma \ref{lemm:thresholdcompstat} and \ref{lemm:uniquethresh} (included in the technical appendix B) imply that $v^L(s)$ is increasing at a rate less than 1, since $\widehat{v}^{L}_{0}(s)$ and $\widehat{v}^{L}_{1}(s)$ both are increasing at a rate less than 1, and $v^L(s)$ is equal to the latter for $s > v^B$ and the former for $s < v^B$. So when $s < v^B$ it holds that $v^L(s)  > s$; however, since the $v^L(s)$ is decreasing in $s$, we also have $v^L(s) < v^B$. The other cases are analogous. 

\bs  \ni \textbf{Proof of Proposition \ref{ssss}:} 
Note that Lemma \ref{lemm:thresholdcompstat} (included in the technical appendix B) shows that  $\widehat{v}^{L}_{1}(s)-\widehat{v}^{L}_{0}(s)$ is decreasing in $s$, which immediately implies the second part of the claim.  
Continuity follows from Lemma \ref{lemm:condthreshold} (included in the technical appendix B), together with the observation that since $v^L(s)=\min\{\widehat{v}^{L}_{1}(s), \widehat{v}^{L}_{0}(s)\}$, and the minimum of two continuous functions is itself continuous.

\bs \ni \textbf{Proof of Proposition \ref{thm:overallthreshold}:}
Note: given any arbitrary $v^{E}$, Lemma \ref{lemm:condthreshold} (included in the technical appendix B) shows the implicit conditions for the thresholds imply that there exists an increasing and continuous function $\widehat{v}^{L}(s)$ such that the manager discloses when $v>\widehat{v}^{L}(s)$. The proof shows that this holds for all $s \neq s^{B}$, and so overall continuity follows from the observation that $\widehat{v}_{1}^{L}(s^{B})= \widehat{v}_{0}^{L}(s^{B})$, given in Lemma \ref{lemm:dyesignal} (included in the technical appendix B). Thus, the late disclosure threshold, defined by the implicit condition, is continuous conditional on $v^{E}$. Existence then follows from noting that a manager at $v=v_{min}$ would never prefer to disclose (and thus would never disclose early), together with continuity of the expected payoffs in $v^{E}$. (This leaves open the possibility that no manager would ever prefer to disclose early, in which case $v^{E}=v_{max}$.)

Uniqueness follows from Lemmas \ref{lemm:uniquethresh} and \ref{lemm:finalpart} (included in the technical appendix B); note in particular that log-concavity of $G$ implies log-concavity of distribution of the manager's value conditional on $v < v^{E}$, delivering uniqueness of the late disclosure thresholds; and furthermore, that given continuity of $v^{L}(s)$, we have the increasing difference property holds and delivers a unique value for $v^{E}$.


\bs \ni \textbf{Proof of  Proposition \ref{gennonmon}:} Let $P_{\tau}(s, \emptyset \mid \widehat{v})$ denote the price function when the market conjectures disclosure threshold $\widehat{v}$ and when the external signal with $\phi=\phi_{1}$ has precision $\tau$.  Write $P_{\infty}(s, \emptyset \mid \widehat{v})$ for the price in the model when $\phi_1\equiv V$, i.e., the veracious signal is perfectly informative,  as in the base model in Section \ref{setting}.  Note that $P_{\tau}(s, \emptyset \mid \widehat{v})$,  as a function of $\tau$,  is continuous at infinity,  since the expectation of $v$ given $s$ and $\phi=V$ is continuous at infinity as a function of $\tau$, and so are the probabilities $\Pr(\phi, \kappa \mid s, \emptyset, \tau)$, with the denominators in all expressions bounded away from 0, since $q  <1$.  Furthermore, the limit is indeed equal to the price function when $\phi_{1} = V$,  since if $\tau \to \infty$ then the veracious signal has no variance and thus we must have $v=s$. On the other hand,  we also have,  by Proposition \ref{pricenon},  that the market price when $\tau\to \infty$  is non-monotone as a function of $s$. 

Pick $s_{1} < s_{2}$ such that $P_{\infty}(s_{1}, \emptyset \mid \widehat{v}) >  P_{\infty}(s_{2}, \emptyset \mid \widehat{v})$,  as  guaranteed by the non-monotonicity in Proposition \ref{pricenon}.  Let $\psi =  P_{\infty}(s_{1}, \emptyset \mid \widehat{v}) -  P_{\infty}(s_{2}, \emptyset \mid \widehat{v})$ and, for $k \in \{1,2\}$, let $\tau_{k}$ be such that 
$|P_{\infty}(s_{k}, \emptyset \mid \widehat{v})-P_{\tau}(s_{k}, \emptyset \mid \widehat{v}) | < \psi/2,$
whenever $\tau > \tau_{k}$,   as is guaranteed to exist by continuity at infinity of $P_{\tau}(s, \emptyset \mid \widehat{v})$.  Let $\widehat{\tau} = \max \{\tau_{1},  \tau_{2}\}$.  Then, because $P_{\tau}(s_{1}, \emptyset \mid \widehat{v}) > P_{\infty}(s_{1}, \emptyset \mid \widehat{v}) - \psi/2= P_{\infty}(s_{2}, \emptyset \mid \widehat{v}) + \psi/2 > P_{\tau}(s_{2}, \emptyset \mid \widehat{v})$, 
we have $P_{\tau}(s_{1}, \emptyset \mid \widehat{v}) >  P_{\tau}(s_{2}, \emptyset \mid \widehat{v})$ if $\tau > \widehat{\tau}$,  as desired.

\bs  \ni \textbf{Proof of Proposition \ref{mon}:} The result is standard and follows from Milgrom (1981), the proof is included for completeness.  By Milgrom (1981),  given a joint distribution $f(s \mid v)$,  the posterior expectation $E[v\lvert s]$ is increasing in $s$,  for any prior over $v$,  if and only if,  for $\tilde{v} > v$, it holds that 
$\frac{f(s \mid \tilde{v})}{f(s \mid v)}$ 
is increasing in $s$.  Recall that $s=v+\varepsilon$,  for $\varepsilon$ distributed according to a log-concave distribution with CDF $F$. We note that for any $s^{*}$ we have $\Pr (s \leq s^{*} \lvert v)=\Pr (v + \varepsilon \leq s^{*} \lvert v)= \Pr( \varepsilon  \leq  s^{*}-v \lvert v)=F \left( s^{*} -v  \right)$.  So, the density of $s$ given $v$ is $f \left( s -v \right)$. Taking $v_{1} > v_{2}$,  the derivative of $f (s-v_{1} )/f(s-v_{2})$ with respect to $s$ is positive if and only if:

\begin{equation*} 
f(s-v_{2})f'(s-v_{1})-f (s-v_{1} ) f'(s-v_{2}) > 0,
\end{equation*}
which holds if and only if
\begin{equation*} 
\frac{f'(s-v_{1})}{f(s-v_{1})}> \frac{f'(s-v_{2})}{f(s-v_{2})}.
\end{equation*}

\noindent Note that $\frac{d}{ds} \log f(s)= \frac{f'(s)}{f(s)}$,  which is decreasing in $s$ by log-concavity.  Increasing $v$ from $v_{2}$ to $v_{1}$ decreases $s-v$,   and therefore increases $\frac{f'(s-v)}{f(s-v)}$.  Hence, the desired inequality holds. 

\bs \ni \textbf{Proof of Corollary \ref{noup}:} Follows from the results in Section \ref{sect:early}, which shows that the early disclosure threshold can only be above $s^{B}$, together with Proposition \ref{pricenon}, which shows that this implies the price can only drop downward.

\bs \ni \textbf{Proof of Corollary \ref{cost}:}  
The proof follows similar steps to Verrecchia (1983) and is omitted.

\bs \ni \textbf{Proof of Corollary \ref{Emore costly}:}  The proof follows from the discussion in the text and is omitted.

\bs \ni \textbf{Proof of Corollary \ref{Lmore costly}:}  The proof follows from the discussion in the text and is omitted.

\bs  \ni \textbf{Proof of Corollary \ref{pr5}:} The comparison between $P^L_4(s,v)$ and $P^L_3(s)$ is straightforward. Here, we only compare $P^L_4(s,\emptyset)$ with $P^L_3(s)$. Recall that, $P^L_3(s)=q s+(1-q)\mu$ and  $P^L_4(s,\emptyset)=\widetilde{v}^L(s)$. Applying the Envelope Proposition, $\frac{d}{d s} \widetilde{v}^L(s)= \frac{1-p}{1-p+pG(\widetilde{v}^L(s))} \cdot q +\frac{pG(\widetilde{v}^L(s))}{1-p+pG(v^L(s))} \cdot\frac{q}{q+(1-q)G(\widetilde{v}^L(s))} >q=\frac{d}{d s} P^L_3(s)$.  

Let $s^o$ be the signal realization that satisfies $s^o=v^L(s=s^o)$. We note that $P^L_3(s=s^o)=q\cdot s^o+(1-q)\mu>s^o=v^L(s=s^o)=P^L_4(s=s^o,\emptyset)$, because $s^o<\mu$. Hence, we have  $P^L_3(s)>P^L_4(s,\emptyset)$ for any  $s \geq s^o$. It remains to show that this inequality holds for $s <s^o$. We note that 
\bean
\min P^L_4(s,\emptyset)&=&P^L_3(s=v_{min},\emptyset)
\\
&=&\frac{(1-p)(1-q)\mu+p G(\widetilde{v}^L(s=v_{min}))\left(1-\frac{q}{q+(1-q)G(\widetilde{v}^L(s=v_{min}))}\right)\mathbb{E}[v\lvert  v\leq \widetilde{v}^L(s=v_{min})]}{1-p+p G(\widetilde{v}^L(s=v_{min}))}
\\
&<&\frac{(1-p)(1-q)\mu+p G(\widetilde{v}^L(s=v_{min}))\left(1-q\right)\mathbb{E}[v\lvert  v\leq \widetilde{v}^L(s=v_{min})]}{1-p+p G(\widetilde{v}^L(s=v_{min}))}
\\
&<&(1-q)\mu
=P^L_3(s=v_{min})
=\min P^L_3(s).
\eean 
Therefore, $P^L_4(s,\emptyset)<P^L_3(s)$ for any $s$.

\bs \ni \textbf{Proof of Corollary \ref{freqearly}:} The equilibrium condition is $\Delta \Pi(\widehat{v}=\widetilde{v}^E) = \Pi(v)- \Pi(\emptyset)=0$. Using the notation in the proof of Proposition \ref{pr2}, we can simplify,
\bean
\Delta \Pi(\widehat{v})  &=& T^B(\widehat{v}) + (\delta + \delta^2 + \delta^3)T^E(\widehat{v}).
\eean
Both $T^B(\widehat{v})$ and $T^E(\widehat{v})$ are decreasing in $\widehat{v}$. Thus, $\Delta \Pi(\widehat{v})$ is also decreasing. Furthermore, 
\bean
\lim_{\widehat{v}\to v_{max}} \Delta \Pi(\widehat{v}) &=& \underbrace{ \lim_{\widehat{v}\to v_{max}} T^B(\widehat{v})}_{<0} + (\delta + \delta^2 + \delta^3)\underbrace{ \lim_{\widehat{v}\to v_{max}} T^E(\widehat{v})}_{<0} <0;
\\
\lim_{\widehat{v}\to v_{min}} \Delta \Pi(\widehat{v}) &=& \underbrace{ \lim_{\widehat{v}\to v_{min}} T^B(\widehat{v})}_{>0} + (\delta + \delta^2 + \delta^3)\underbrace{ \lim_{\widehat{v}\to v_{min}} T^E(\widehat{v})}_{>0} >0.
\eean
Therefore, there exists a unique threshold $\widetilde{v}^E \in (v_{min},v_{max})$ such that the manager discloses if and only if $v > \widetilde{v}^E$. To see that $\widetilde{v}^E \in (v^B, v^E)$ note that
\bean
\lim_{\widehat{v}\to v^B} \Delta \Pi(\widehat{v}) &=& \underbrace{ \lim_{\widehat{v}\to v^B} T^B(\widehat{v})}_{=0} + (\delta + \delta^2 + \delta^3)\underbrace{ \lim_{\widehat{v}\to v^B} T^E(\widehat{v})}_{>0} >0;
\\
\lim_{\widehat{v}\to v^E} \Delta \Pi(\widehat{v}) &=& \underbrace{ \lim_{\widehat{v}\to v^E} T^B(\widehat{v})}_{<0} + (\delta + \delta^2 + \delta^3)\underbrace{ \lim_{\widehat{v}\to v^E} T^E(\widehat{v})}_{=0} <0.
\eean
Lastly, using the Implicit Function Theorem, 
\bean
\frac{\partial}{\partial \delta} v^E \propto \frac{\partial}{\partial \delta} \Delta \Pi(\widetilde{v}^E) = (1+2\delta+3 \delta^2) \underbrace{T^E(\widehat{v}=\widetilde{v}^E)}_{>0} >0. 
\eean

\bs \ni \textbf{Proof of Corollary \ref{comparisonprice1}:}
The first part follows from the discussion in the text and is omitted.
For the second part, we know that
$
P_2^E(\emptyset \lvert \widetilde{v}^E) - P_3^E(s, \emptyset \lvert \widetilde{v}^E ) = P(\emptyset \lvert \widetilde{v}^E) - P(s,\emptyset \lvert \widetilde{v}^E),
$
where
	\bean
	P(\emptyset \lvert \widetilde{v}^E) &=&  \frac{1-p}{1-p + p \cdot G(\widetilde{v}^E)} \cdot \mu + \frac{p\cdot G(\widetilde{v}^E)}{1-p + p\cdot G(\widetilde{v}^E)}  \cdot  \mathbb{E}[v\lvert v \leq \widetilde{v}^E ] 
	\\
P(s,\emptyset \lvert \widetilde{v}^E, s\leq \widetilde{v}^E) &=&  \frac{1-p}{1-p + p \cdot G(\widetilde{v}^E)}[ (1-q  )\cdot \mu 
+  q \cdot s]
\\
&&+ \frac{p\cdot G(\widetilde{v}^E)}{1-p + p\cdot G(\widetilde{v}^E)}
 [ (1- q\cdot \pi(\widetilde{v}^E)) \cdot  \mathbb{E}[v\lvert v \leq \widetilde{v}^E] +q \cdot \pi(\widetilde{v}^E) \cdot s ] 
\\
P(s,\emptyset \lvert \widetilde{v}^E, s > \widetilde{v}^E)
	&=& \frac{1-p}{1-p + p \cdot G(\widetilde{v}^E)}[ (1-q  )\cdot \mu 
+  q \cdot s]
\\
&&+ \frac{p\cdot G(\widetilde{v}^E)}{1-p + p\cdot G(\widetilde{v}^E)}  \cdot  \mathbb{E}[v\lvert v \leq \widetilde{v}^E ].
	\eean
Therefore:
\bean
P_2^E(\emptyset \lvert \widetilde{v}^E, s> \widetilde{v}^E) - P_3^E(s, \emptyset \lvert \widetilde{v}^E , s>\widetilde{v}^E) &\propto & \mu- (1-q) \cdot \mu - q\cdot s  \propto \mu-s,
\eean
which is positive for $s<\mu$ and negative otherwise.
Furthermore,
\bean
 w(s) &\equiv & P_2^E(\emptyset \lvert \widetilde{v}^E, s \leq \widetilde{v}^E) - P_3^E(s, \emptyset \lvert \widetilde{v}^E , s\leq \widetilde{v}^E) 
 \\
 &=& \frac{1-p}{1-p + p \cdot G(\widetilde{v}^E)} \cdot (\mu-s) 
+  \frac{p\cdot G(\widetilde{v}^E)}{1-p + p\cdot G(\widetilde{v}^E)}  \pi( \widetilde{v}^E) (\mathbb{E}[v\lvert v \leq \widetilde{v}^E ]-s)
\eean	
This term is decreasing in $s$. It is immediate that $w(s=v_{min})>0$. 
Furthermore,
\bean
w(s=\widetilde{v}^E)&=& \frac{1-p}{1-p + p \cdot G(\widetilde{v}^E)} \cdot (\mu-\widetilde{v}^E) 
+  \frac{p\cdot G(\widetilde{v}^E)}{1-p + p\cdot G(\widetilde{v}^E)}  \pi( \widetilde{v}^E) (\mathbb{E}[v\lvert v \leq \widetilde{v}^E ]-\widetilde{v}^E).
\eean
From the proof of Proposition \ref{freqearly} recall that $\widetilde{v}^E$ satisfies $\Delta \Pi(\widehat{v}=\widetilde{v}^E) =0$. Rearranging,
\bean
\Delta \Pi(\widehat{v}=\widetilde{v}^E) &=&  (1+ \delta + \delta^2 + \delta^3) (v - \mathbb{E}[P(s,\emptyset) \lvert v=\widehat{v}, \widehat{v}=\widetilde{v}^E]) 
\\
&&+ \mathbb{E}[P(s,\emptyset) \lvert v=\widehat{v}, \widehat{v}=\widetilde{v}^E]-P (\emptyset \lvert \widehat{v}=\widetilde{v}^E).
\eean
Because $\widetilde{v}^E \in (v^B, v^E)$, it holds that $(1+ \delta + \delta^2 + \delta^3) (v - \mathbb{E}[P(s,\emptyset) \lvert v=\widehat{v}, \widehat{v}=\widetilde{v}^E]) <0$. Therefore, it has to be that $\mathbb{E}[P(s,\emptyset) \lvert v=\widehat{v}, \widehat{v}=\widetilde{v}^E]-P (\emptyset \lvert \widehat{v}=\widetilde{v}^E)>0$. Lastly, note that 
\bean
w(s=\widetilde{v}^E)&=& - (\mathbb{E}[P(s,\emptyset) \lvert v=\widehat{v}, \widehat{v}=\widetilde{v}^E]-P (\emptyset \lvert \widehat{v}=\widetilde{v}^E)) <0.
\eean
Therefore, there exists $s^{\dag \dag} \in (v_{min},\widetilde{v}^E)$ such that:
 $P_2^E(\emptyset) \geq P_3^E(s, \emptyset)$  if $s\in [v_{min}, s^{\dag \dag}]$ or $s\in [\widetilde{v}^E, \mu]$. However, $P_2^E(\emptyset) \leq P_3^E(s, \emptyset)$  if $s\in [s^{\dag \dag},\widetilde{v}^E]$ or $s\in [\mu, v_{max}]$.

 \bs

\bs

\section{\large{Technical Appendix}}
\label{THM}

This appendix presents a number of more technical results which are used in the proofs of various parts of the above analysis. For completeness, \textit{all proofs in this appendix allow for the dynamic case with rescheduling costs} but this will also nest the special cases where disclosure can only be early or only late (and so there are no rescheduling costs). Here, we will also introduce some notation to allow for more general strategies beyond those considered in the main text, in particular those that may allow for mixed strategies.
Note that an arbitrary strategy for the manager can be represented as 

\begin{equation*} 
\sigma^{E} : [0,1] \rightarrow \Delta(\{Disclosure, Silence\}),  \sigma^{L} : [0,1] \times [0,1] \rightarrow \Delta(\{Disclosure, Silence\}),
\end{equation*} 

\noindent where $\sigma^{E}$ specifies the early disclosure decision while $\sigma^{L}$ specifies the late disclosure decision; the former depends only on the firm value, whereas the latter depends on both the firm value as well as the external signal.  In particular,  since the manager's strategy is non-veracious after disclosure occurs,  the second period strategy does not need to condition on the first period action (i.e.  we take this to be ``silence'' by assumption). Slightly abusing notation let the induced equilibrium price function be  $P^{\sigma^{E},\sigma^{L}}(s,\emptyset)$.

\subsection{Preliminaries}  
We start with basic, more preliminary observations. First, for the sake of completeness, we present a result of independent interest to the direct setting at hand: that with a log-concave distribution and intermediate probability that the manager is informed, $\mathbb{E}[v \lvert v \leq \overline{v}]$ from the benchmark model is increasing at a rate less than 1.

\begin{lemma} \label{lemm:BB}
In the benchmark model without external signals, if $g$ is log-concave, then $\mathbb{E}[v \lvert v \leq \overline{v}]$ is increasing in $\overline{v}$ at a rate less than 1. 
\end{lemma}

\begin{proof}[Proof of Lemma \ref{lemm:BB}] Computing the derivative of $\mathbb{E}[v \lvert v \leq \overline{v}]$ with respect to $\overline{v}$, we see that it is:

\begin{equation*} 
\frac{pg(\overline{v})(\overline{v}-\mathbb{E}[v \lvert v \leq \overline{v}])
}{(1-p+pG(\overline{v}))}
\end{equation*}

If $\overline{v} < v^{B}$, then this expression is negative, and hence less than 1 for all $p$ such that $\overline{v} < v^{B}$. Otherwise, for all $\overline{v}$, the derivative of this expression with respect to $p$ is proportional to: 

\begin{equation*} 
(1-p+pG(\overline{v})) \left(g(\overline{v})(\overline{v}-\mathbb{E}[v \lvert v \leq \overline{v}]) -pg(v) \frac{\partial}{\partial p}\mathbb{E}[v \lvert v \leq \overline{v}]\right)+pg(\overline{v})(\overline{v}-\mathbb{E}[v \lvert v \leq \overline{v}])(1-G(\overline{v})).
\end{equation*}

If $\overline{v} > \mathbb{E}[v \mid v \leq \overline{v}]$, then it is immediate that this expression is positive for all $p$: this observation that uses that $\mathbb{E}[v \lvert v \leq \overline{v}]$ is decreasing in $p$ (see also Kartik, Lee and Suen 2019 and references cited therein). Thus, since the slope is increasing in $p$, it is less than the slope in the case that $p=1$ (which Bagnoli and Bergstrom 2005 imply  is less than 1), we have the slope is less than 1 for all $p \in (0,1)$, as claimed.
\end{proof}

 Second, we show that \emph{late} disclosure must be characterized by a threshold which depends on $s$ and, third, that beliefs about veracity have no impact on market price only for a single signal. Our last results show that that the price (as a function of the signal) is increasing with a slope less than 1, which flattens as the market assigns lower probability to veracity. The last point makes use of an auxiliary game which we introduce and discuss later in the proof.

\begin{lemma}  \label{lemm:latethresh}
If $\sigma^{L}$ is an equilibrium late disclosure strategy, then $\sigma^{L}$ is characterized by a threshold, $v^{L,\sigma_{E}}(s)$ such that the manager discloses whenever $s >v^{L,\sigma_{E}}(s)$. Furthermore, for any $\sigma^{E}$ and any $s$, there exists an equilibrium disclosure threshold $v^{L,\sigma_{E}}(s)$. 
\end{lemma}

\begin{proof}[Proof of Lemma \ref{lemm:latethresh}] 
Standard; the manager's payoff from nondisclosure is constant in equilibrium,  whereas the payoff from disclosure is increasing in the manager's value.
\end{proof}

\begin{lemma}  \label{lemm:dyesignal}
Given any arbitrary $\sigma^{E}$,  there exists a single signal,  which we denote $s^{B}$,  such that the expectation of $v$ conditional on $s^{B}$ and nondisclosure is constant in the belief about veracity.   
\end{lemma}

\begin{proof}[Proof of Lemma \ref{lemm:dyesignal}] 
The signal can be solved for in closed form; if the market believes the veracious signals are never disclosed: 

\begin{equation} 
s^{B} = \frac{qs^{B} +(1-q) \left((1-p) \mu + p \int_{v_{min}}^{v_{max}} v \sigma^{E}(v)g(v)dv \right)}{q+(1-q)(1-p)+p \int_{v_{min}}^{v_{max}} \sigma^{E}(v)g(v)dv}.  \label{eq:dyedef}
\end{equation} 

\noindent We can then solve for $s^{B}$ as: 

\begin{equation*} 
s^{B}=  \frac{(1-p)\mu+ p \int_{v_{min}}^{v_{max}} v \sigma^{E}(v)g(v)dv}{(1-p)+p\int_{v_{min}}^{v_{max}} \sigma^{E}(v)g(v)dv}.
\end{equation*} 

We note that, varying the market belief that veracious signals are disclosed is equivalent to changing the value of $q$ in equation (\ref{eq:dyedef}).  But as we have seen,  the value for $s^{B}$ does not depend on $q$,  and furthermore this value is always uniquely defined. 
\end{proof}

\subsubsection{The Auxiliary Game} 
We describe the auxiliary game which we use in our analysis, as an intermediate step toward equilibrium. This game is identical to our dynamic model, but if the manager does not disclose early, then the late-disclosure strategy is assumed to be exogenous whenever the signal is veracious. That is, we assume that if $\phi=V$, the manager discloses with probability $\theta$ if informed. We denote $\widehat{\sigma}^{L,\theta}$ as a candidate late disclosure strategy in the auxiliary game.  When discussing the auxiliary game, we restrict to cases where the threshold $\widehat{v}_{\theta}^{L}(s)$ is continuous in $s$. We note that this property holds by assumption if disclosure can only be early---in the case with dynamic disclosure, this property is established in Lemma \ref{lemm:condthreshold}. 

Most general properties of the model can be derived by considering the auxiliary game first, and then determining which case is applicable. 
We provide three results which show that the price function in the auxiliary game is well-behaved---it is increasing in $s$ with a slope less than 1, and it is flatter when $\theta$ is higher. We use these properties to show equilibrium existence in the analysis of dynamic disclosure as well as when disclosure must be early. 

\begin{lemma} 
\label{lemm:increasingproof} 
Let $P^{\sigma^{E},\sigma^{L}, \theta}(s,\emptyset)$ correspond to the price function in the auxiliary game when the manager uses strategies $\sigma^{E}$ and $\sigma^{L}$.  $P^{\sigma^{E},\sigma^{L}, \theta}(s,\emptyset)$ is increasing, for every $\theta$.
\end{lemma}

\begin{proof} 
The price is given by the following, assume disclosure is above a threshold $\widehat{v}$ (which may depend on $s$) and given early disclosure strategy $\sigma^{E}$:

\begin{equation}  \frac{ q(1-p\theta)s +(1-q(1-p\theta)) \left((1-p)\mu + p \int_{v_{min}}^{\widehat{v}} v \sigma^{E}(v)g(v) dv \right)}{q(1-p\theta) +(1-q(1-p\theta)) \left((1-p)\int_{v_{min}}^{v_{max}} g(v)dv + p \int_{v_{min}}^{\widehat{v}}  \sigma^{E}(v) g(v) dv \right)} .   \label{eq:priceexpressionproof}
\end{equation}

If $\widehat{v}$ does not vary with $s$, then the observation that the price is increasing in $s$ is immediate. The same holds if the manager either always discloses or never discloses following some signal $s$. So, the last case to consider is when the price (\ref{eq:priceexpressionproof}) defines the indifference condition for the threshold manager, $v_{\theta}^{L}(s)$. So suppose  $\widehat{v}=v_{\theta}^{L}(s)$ and that $v_{\theta}^{L}(s)$ equals (\ref{eq:priceexpressionproof}). For simplicity take $\theta=0$, as this simplifies notation but does not change the argument. We rewrite this as: 
\begin{eqnarray*}
&&\widehat{v}^{L}_{0}(s)\left(q+(1-q)\left( (1-p) +p \int_{v_{min}}^{v_{0}^{L}(s)}\sigma^{E}(v)g(v)dv \right)\right)
\\
&=& qs +(1-q) \left((1-p)\mu + p \int_{v_{min}}^{\widehat{v}^{L}_{0}(s)} v \sigma^{E}(v)g(v) dv \right)
\end{eqnarray*}

Write $\widehat{v}_{0}^{L}(s+\Delta)=\widehat{v}_{0}^{L}(s)+\tilde{\Delta}$. We claim that if $\Delta > 0$ then $\tilde{\Delta} > 0$. Subtracting the definition of $\widehat{v}_{0}^{L}(s)$ from the corresponding definition of $\widehat{v}_{0}^{L}(s+\Delta)$ We have: 

\begin{multline*}
\widehat{v}_{0}^{L}(s)\left((1-q)p\int_{v_{0}^{L}(s)}^{v_{0}^{L}(s)+\tilde{\Delta}}\sigma^{E}(v) g(v) dv\right)+\tilde{\Delta} \left(q+(1-q)\left( (1-p) +p \int_{v_{min}}^{v_{0}^{L}(s)+\tilde{\Delta}}\sigma^{E}(v)g(v)dv \right)\right)\\ = q\Delta+(1-q)p\int_{\widehat{v}_{0}^{L}(s)}^{\widehat{v}_{0}^{L}(s)+\tilde{\Delta}} v\sigma^{E}(v) g(v) dv.
\end{multline*}

\noindent So: 

\begin{eqnarray*}
&&\tilde{\Delta} \left(q+(1-q)\left( (1-p) +p \int_{v_{min}}^{v_{0}^{L}(s)+\tilde{\Delta}}\sigma^{E}(v)g(v)dv \right)\right) \\
&=&  q\Delta + (1-q)p\int_{\widehat{v}_{0}^{L}(s)}^{\widehat{v}_{0}^{L}(s)+\tilde{\Delta}}(v-\widehat{v}_{0}^{L}(s))\sigma^{E}(v)g(v)dv. 
\end{eqnarray*}

\noindent Now, $|v-\widehat{v}_{0}^{L}(s)| < |\tilde{\Delta}|$ for all $v$ in the integral on the right hand side of this equation. Replicating the argument from Lemma \ref{lemm:condthreshold}, we have that $\tilde{\Delta}$ cannot approach 0 at a slower rate than $\Delta$. It follows that the integral on the right hand side becomes negligible in the limit as $\Delta \rightarrow 0$ (since it approaches 0 at the rate of $\Delta^{2}$). Since the factor multiplying $\tilde{\Delta}$ is positive and since $q, \Delta$ are positive as well, we therefore conclude that $\tilde{\Delta}$ is positive. 
\end{proof}

\begin{lemma} \label{lemm:PriceFunc}
Let $P^{\sigma^{E},\sigma^{L}, \theta}(s,\emptyset)$ correspond to the price function in the auxiliary game when the manager uses strategies $\sigma^{E}$ and $\sigma^{L}$.  The slope of the price function is less than 1, for any $\theta$.
\end{lemma}

\begin{proof}[Proof of Lemma \ref{lemm:PriceFunc}]
As with the proof of Lemma \ref{lemm:increasingproof}, we note that if the disclosure threshold is fixed given a signal realization $s$ (e.g., if disclosure is always early), then this holds by inspection of (\ref{eq:priceexpressionproof}). Thus, in the proof we focus on the case where the manager must be indifferent at the price.  We first show that $ \widehat{v}^{L}_{\theta}(s)$ has a slope less than 1.  Consider the implicit condition

\begin{equation} 
0= \widehat{v}^{L}_{\theta}(s)- c^{L} -  \frac{ q(1-p\theta)s +(1-q(1-p\theta)) \left((1-p)\mu + p \int_{v_{min}}^{\widehat{v}^{L}_{\theta}(s)} v \sigma^{E}(v)g(v) dv \right)}{q(1-p\theta) +(1-q(1-p\theta)) \left((1-p)\int_{v_{min}}^{v_{max}} g(v)dv + p \int_{v_{min}}^{\widehat{v}^{L}_{\theta}(s)}  \sigma^{E}(v) g(v) dv \right)} .  \label{eq:ratio}
\end{equation}

\noindent In particular, the argument that the ratio in Equation (\ref{eq:ratio}) defines the expected value upon nondisclosure follows the same argument as in Proposition \ref{joint}.  Suppose that $s$ increases to $s + \Delta$ for $\Delta$ small.  Then the change in $\widehat{v}^{L}_{\theta}(s)$, say $\widetilde{\Delta}$, is equal to the change in the ratio defined in (\ref{eq:ratio}).  Since $\widehat{v}^{L}_{\theta}(s)$ is continuous,  we know that this change must be small as well.  If $s$ increases by $\Delta$, this ratio becomes: 

\begin{equation*} 
\frac{ q(1-p\theta)(s+ \Delta) +(1-q(1-p\theta)) \left((1-p)\int_{v_{min}}^{v_{max}} v g(v)dv + p \int_{v_{min}}^{\widehat{v}^{L}_{\theta}(s)+\widetilde{\Delta}} v\sigma^{E}(v) g(v) dv \right)}{q(1-p\theta) +(1-q(1-p\theta)) \left((1-p)\int_{v_{min}}^{v_{max}} g(v)dv + p \int_{v_{min}}^{\widehat{v}^{L}_{\theta}(s)+\widetilde{\Delta}}\sigma^{E}(v) g(v) dv \right)}
\end{equation*}

Noting that we must have $\widetilde{\Delta}$ small,  we obtain the following approximation for $\widetilde{\Delta}$,  using (\ref{eq:ratio}) (which becomes increasingly accurate as $\Delta \rightarrow 0$): 

\begin{equation*} 
\widetilde{\Delta} \approx  \frac{ q(1-p\theta) \Delta +(1-q) p \widetilde{\Delta} \widehat{v}^{L}_{\theta}(s)\sigma^{E}(\widehat{v}^{L}_{\theta} (s)) g(\widehat{v}^{L}_{\theta} (s))  }{q(1-p\theta) +(1-q) \left((1-p)\mu + p \int_{v_{min}}^{\widehat{v}^{L}_{\theta}(s)+\widetilde{\Delta}}g(v) \sigma^{E}(v) dv \right)}.
\end{equation*}

\noindent (An upper bound of this quantity can be obtained by taking $\widetilde{\Delta}=0$ in the denominator,  and the argument will still work) Note that the denominator is greater than $q$.  Thus, without the $\widetilde{\Delta}$ terms in the numerator on the right hand side, $\widetilde{\Delta}$ would increase by less than $(q/q) \cdot \Delta$; adding these terms in,  we see the necessary increase in $\widetilde{\Delta}$ is even smaller.  Thus, $\widetilde{\Delta} < \Delta$.  The result follows from noting that the implicit conditional also defines the rate of change of the price as a function of the signal,  provided this implicit condition is satisfied;  if note, then this threshold does not adjust and an identical argument can be used (simply treating $\widetilde{\Delta}=0$). 
\end{proof}

\begin{lemma}  \label{lemm:slopelemma}
The slope of $P^{\sigma^{E},\sigma^{L}, \theta}(s,\emptyset)$ is decreasing in $\theta$. 
\end{lemma}

\begin{proof}[Proof of Lemma \ref{lemm:slopelemma}] 
Inspecting (\ref{eq:ratio}), we see that as $\theta$ increases,  less weight is placed on $s$ and more is placed on a term that is independent of $s$,  implying that the price as a function of $s$ is flatter when $\theta$ is larger. 
\end{proof}

Putting together Lemmas  \ref{lemm:dyesignal}, and \ref{lemm:slopelemma}  we see that varying $\theta$ essentially ``twists'' the price through $s^{B}$.

\subsection{Early Disclosure Must Be a Threshold} \label{sect:early}

In this section, we show that the early disclosure decision \emph{must} be characterized by a threshold.  The argument is more involved than the case of late disclosure, since we need to take into account the possibility of option value. 

\begin{lemma}  \label{lemm:nojumpup}
If $\sigma^{E}$ prescribes that the manager discloses early with probability 1 when the value is $v$, then $v \geq s^B$. 
\end{lemma}

\begin{proof}[Proof of Lemma \ref{lemm:nojumpup}]
Suppose to the contrary, that some $v < s^B$ is disclosed with probability 1 under $\sigma^{E}$. Note that a manager with value $v=v_{min}$ gets payoff $v_{min}$ from disclosing, and a higher payoff from never disclosing (due to the positive probability of being uninformed); the same is true for all managers with values sufficiently close to $v_{min}$. So, if $\{v^{\sigma^{E}}\}$ is the set of manager values which disclose under $\sigma^{E}$, and $v^{*}=\text{inf } \{v^{\sigma^{E}}\}$, then we can find a sequence $v_{\Gamma}^{n} \rightarrow \widetilde{v}$ such that $v_{\Gamma}^{n}$ prefers not disclosing early whereas $\widetilde{v}$ does, for some $\widetilde{v}$. Note that, since $\widetilde{v} < s^B$, we have $\widehat{v}^{L}_{1}(\widetilde{v}) > \widehat{v}^{L}_{0}(\widetilde{v})$ (where $\widehat{v}^{L}_{\theta}(\widetilde{v})$ is as in Lemma B.4.). On the other hand, the expected payoff conditional on the signal being non-veracious is constant in the manager's value.

So, compare the manager's payoffs under $\widetilde{v}$ and $v_{\Gamma}^{n}$ for $n$ sufficiently large. If the signal is non-veracious, the manager obtains $\int_{v_{min}}^{v_{max}} \max \{P^{\sigma^{E}, \sigma^{L}}(s,\emptyset), v\}g(s)ds$, which is the same in the limit as $n \rightarrow \infty$ for $v_{\Gamma}^{n}$ as $\widetilde{v}$. If the signal is veracious, then since the market conjectures $v_{\Gamma}^{n}$ do not disclose and $\widetilde{v}$ does, and $\widetilde{v} < s^B$, we have that the payoff from nondisclosure, if the signal is veracious, is larger for the manager of type $\widetilde{v}$ than $v_{\Gamma}^{n}$. And by definition, the payoff from disclosure in the limit as $n \rightarrow \infty$ of $v_{\Gamma}^{n}$ is $\widetilde{v}$. 

Putting this together, we conclude that the payoff from nondisclosure is strictly higher in the limit for $\widetilde{v}$ than $v_{\Gamma}^{n}$ managers, whereas the payoff from disclosure is the same. But since the $v_{\Gamma}^{n}$  managers do at least as well from not-disclosing as disclosing, the $\widetilde{v}$ manager must do strictly better. This suggests a profitable deviation, in contradiction to the hypothesis that $\sigma^{E}$ was an equilibrium strategy. 
 \end{proof}

The previous argument shows that managers cannot disclose early with probability 1 if $v < s^B$. This does not yet imply that the expected payoff from not disclosing early can only equal $v- c^{E}$ if $v \geq s^B$; while the argument does show there cannot be any pure strategy equilibrium with this property, we must also consider the case of mixed strategies, where the manager is indifferent between disclosure decisions within some range.

\begin{figure}[t]
\setlength{\unitlength}{.5cm}
\begin {picture}(15.5,15.5)\thicklines
\put(2,0){\includegraphics[width=0.9\textwidth]{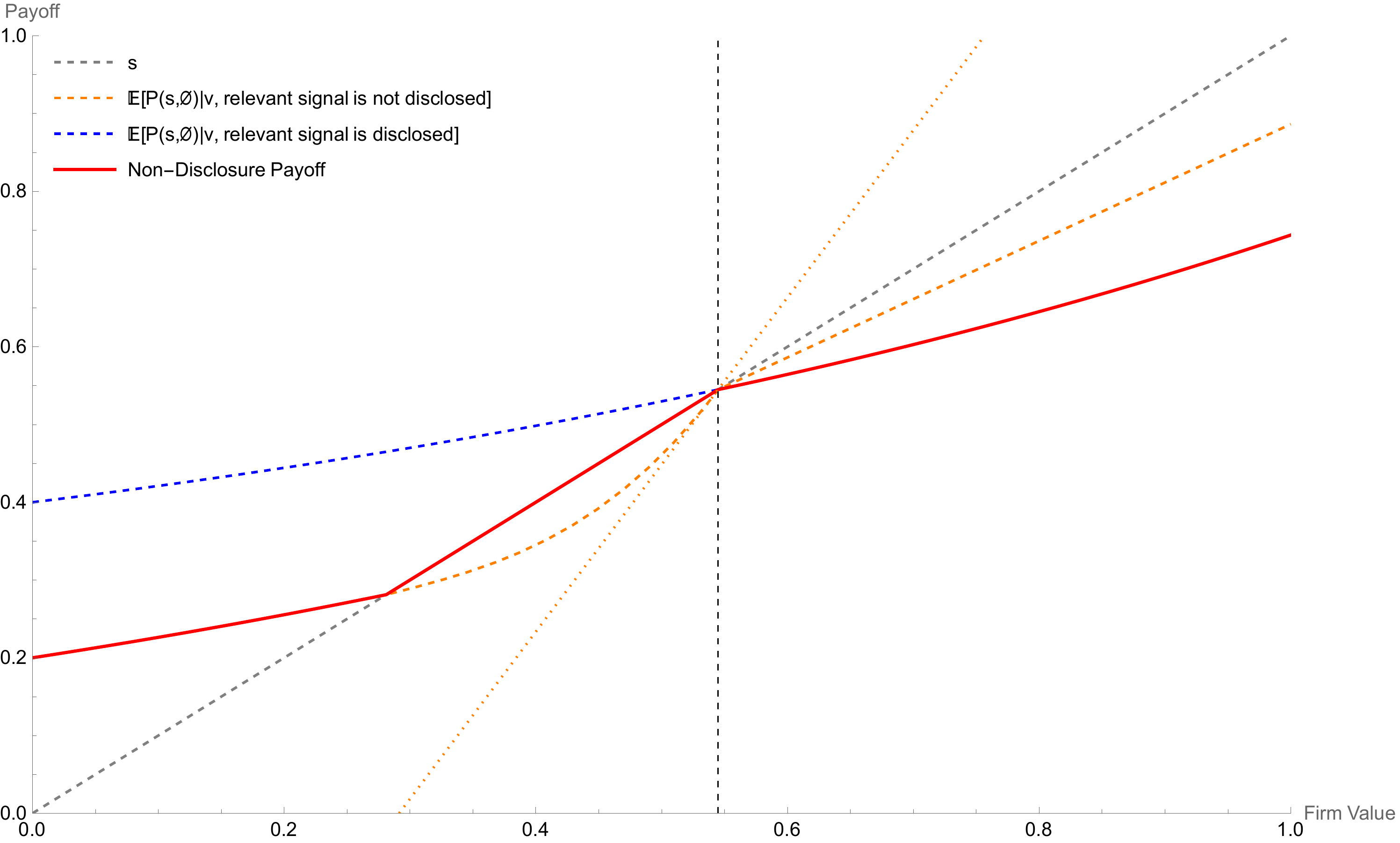}}

\put(16.5,16.8){\color{black}\small{$s^B$}}

\end{picture}

\footnotesize{\textbf{Figure B.1}: Graphical Proof of Lemma \ref{lemm:nomix} with $v_{min}=0$, $v_{max}=1$; if there are two points such that the price does not depend on the market's conjecture regarding whether a signal that is veracious would have been disclosed by the manager---which must be the case if the manager is indifferent between decisions at some value below $v^{*}< s^B$---then there must be some point in $(v^{*}, s^B)$ at which the slope of the expected price given $v$ is greater than 1.  }

\end{figure}

\begin{lemma}  
\label{lemm:nomix}
Suppose the manager is willing to disclose early given a value of $v$.  Then $v > s^{B}$. 
\end{lemma}

\begin{proof}[Proof of Lemma \ref{lemm:nomix}]
Again we prove this result by contradiction. Note that the argument from Lemma \ref{lemm:nojumpup} implies that, if some $v < s^B$ obtains payoff $v-c^{E}$ in equilibrium, then this manager cannot disclose with probability 1; so, if managers are disclosing with positive probability, then they must be following a mixed strategy. And, if they follow a mixed strategy, their payoff from remaining silent must be equal to their payoff from disclosing, $v-c^{E}$.  

If the market conjectures a manager with value $v$ discloses with probability $\theta$, then the expected payoff is: 

\begin{equation*} 
\pi^{\sigma^{E}}(v, \theta)=q \frac{qv(1- p \theta)+(1-q) \int_{v_{min}}^{v_{max}} v \sigma^{E}(v)g(v)dv}{q(1- p \theta)+(1-q)\int_{v_{min}}^{v_{max}} \sigma^{E}(v)g(v)dv } + (1-q)\mathbb{E}[P^{\sigma^{E}}(s,\emptyset)].
\end{equation*}

\noindent Let $v^{*}$ be the highest $v$ such that $\pi^{\sigma^{E}}(v, 0)$ equals $v-c^{E}$; again, since $\pi^{\sigma^{E}}(v, 0)$ is larger than $v_{min}-c^{E}$, but also some manager is willing to disclose for $v < s^B$, it holds that $v^{*} \in (v_{min}, s^B)$. Now, by the definition of $s^B$, $\pi^{\sigma^{E}}(v, 0)=\pi^{\sigma^{E}}(v, 1)$, since the expectation is independent of $v$ and the price given $s=s^{B}$ is independent of $\theta$; for all $v < s^B$, since a pure strategy would (by the argument from Lemma \ref{lemm:nojumpup}) imply a profitable deviation, it follows that, for \emph{every} $v \in (v^{*}, s^B)$, the equilibrium expected payoff from non-disclosure is equal to $v-c^{E}$. For $v \geq s^B$, types must disclose with probability 1 in equilibrium, since the maximum possible payoff from nondisclosure increases at a rate less than 1, so that we cannot have $v-c^{E} \leq \pi^{\sigma^{E}}(v,\theta)$ for $v > s^{B}$ (under the conjecture of the proof).

To summarize, these arguments imply that: 

\begin{itemize} 
\item At $v^{*}$, the manager's payoff when the market conjectures the manager does not disclose is equal to $v^{*} - c^{E}$. 

\item At $s^B$, the manager's payoff when the market conjectures the manager discloses is equal to $s^B - c^{E}$, and is also equal to the expected price if the market were to conjecture the manager does not disclose. 
\end{itemize}

Together, these two bulletpoints show that the function $\pi^{\sigma^{E}}(v,0)$  is equal to $v-c^{E}$ for \emph{two} values of $v$; thus there must be some $\widetilde{v}$ where the slope of:

\begin{equation*} 
q \cdot \frac{q \cdot v+(1-q) \int_{v_{min}}^{v^{L}(v)} \widetilde{v} \cdot \sigma^{E}(\widetilde{v})g(\widetilde{v})dv}{q+(1-q)\int_{v_{min}}^{v^{L}(v)} \sigma^{E}(\widetilde{v})g(\widetilde{v})d\widetilde{v} },
\end{equation*}

\noindent is greater than 1; but we have argued previously (Lemma \ref{lemm:PriceFunc}) that this function divided by $q$ (i.e., the price function) has a slope less than 1. Thus, the slope when multiplying by $q$ is even smaller, and in particular still everywhere less than 1. This contradiction establishes that if $v$ is such that the manager's payoff from early disclosure is equal to the expected payoff from not disclosing early, it must be that $v \geq s^B$. 
\end{proof}

\noindent The previous two Lemmas imply that the manager can only possibly be willing to disclose early if $v$ is such that $v > s^{B}$.

\begin{lemma}  \label{lemm:dyesignalsfirst}
Suppose $\sigma^{E}, \sigma^{L}$ induce a price function $P^{\sigma^{E}, \sigma^{L}}(s, \emptyset)$ which is increasing with a slope less than 1. If $\sigma^{E}$ is part of an equilibrium, then it is characterized by a threshold (possibly at the boundary), say $v^{E}$ above which the manager discloses early and below which the manager does not disclose early. Furthermore, if the manager is indifferent between disclosing early and not at $v^{E}$, then this threshold determines an equilibrium.   
\end{lemma}

\begin{proof}[Proof of Lemma \ref{lemm:dyesignalsfirst}] 
We only consider equilibria where early disclosure occurs with positive probability,  since otherwise the threshold trivially is at $v_{max}$.  Clearly,  if the manager discloses early,  then the payoff is $v-c^{E}$, which increases in $v$ at a rate of 1.  We show that the payoff from nondisclosure increases  in $v$ at a rate less than 1.  Assume that $v$ is such that $v- c^{E} < P^{\sigma_{E}, \sigma_{L}}(v_{max}, \emptyset)$,  since otherwise the manager would always disclose early.  

Let $\widetilde{v}$ be the infimum over manager values which are willing to disclose early.  By the previous,  we have that $\widetilde{v} \geq s^{B}$.  We show that all values with $v > \widetilde{v}$ strictly prefer to disclose early.  There are two cases to consider: (i) signal $s$ is uninformative; (ii) signal $s$ is informative. 

Consider the first case.  The manager's payoff, as a function of the signal $s$, is: 

\begin{equation*} 
\max\{v-c^{L}, P^{\sigma^{E}, \sigma^{L}}(s, \emptyset)\}. 
\end{equation*}

\noindent Since the signal is uninformative,  $P^{\sigma^{E}, \sigma^{L}}(s, \emptyset)$ is independent of $v$,  so that this expression is weakly increasing in $v$ for all $s$;  thus,  $\int_{v_{min}}^{v_{max}} \max\{v-c^{L}, P^{\sigma^{E}, \sigma^{L}}(s, \emptyset)\} g(s) ds$ is also increasing in $v$, but at a rate less than $1$ (since the integral does not change in the event that $P^{\sigma^{E}, \sigma^{L}}(s) > v-c^{L}$);  this involves the manager disclosing when the signal is below $w^{L}(v)$ and not when it is above $w^{L}(v)$.  In particular,  as long as $v- c^{E} < P^{\sigma_{E}, \sigma_{L}}(v_{max}, \emptyset)$,  and since $v -c^{E} > v-c^{L}$ (otherwise disclosing early is strictly dominated), we know that there are a range of signal realizations such that the manager would \emph{not} disclose late. 

Now consider the second case.  Since $v > s^{B}$,  we therefore have that, no matter what inference the market makes about the probability of veracity following signal $v$,   the market price is lower than $v$ itself.  Indeed, at $v=s^{B}$, then the market price is constant in belief about veracity.  Now,  the change in the price following a \emph{veracious signal} for a manager with type $\widetilde{v}$ versus type $v$ is bounded by the change in the price if non-veracious signals are not-disclosed.  Since the rate of change of the price for such signals in the auxiliary game is less than 1,  we have that the manager's payoff increases at a rate less than 1 if they are in a region where they do not disclose late following veracious signals.  If this condition does not hold, then the payoff from not disclosing early increases at a rate equal to exactly 1. 

Putting this together, we see that the total rate of change in the payoff from not disclosing early is a convex combination of a function with slope at most 1 and a function with slope strictly less than 1.   Importantly,  since the weights on each of these terms themselves do not vary with the signal (since veracity is independent of $v$),  we have that the overall rate of change in the payoff from not disclosing early is increasing at a rate strictly less than 1.  It follows that if some manager with value $v'$ weakly prefers to disclose early,  then any manager with $v > v'$ strictly prefers to disclose early,  and conversely if a manger with value $v'$ weakly prefers to not disclose early,  then any manager with value $v < v'$ strictly prefers to not disclose early.  Thus, all equilibria $\sigma^{E}$ are characterized by thresholds. 
\end{proof}

\subsection{Late Thresholds} 
This section presents our technical results on the late disclosure threshold, $v^{L}(s)$. We show that standard arguments from past work can be applied to characterize the late disclosure threshold. We also show that we can take the late disclosure threshold to be continuous in the external signal. Finally, we show that the disclosure thresholds in the auxiliary game also determine the thresholds in the original game, with market beliefs about veracity ``switching'' at $s^{B}$. 

We start by showing that the auxiliary game has well-behaved equilibrium disclosure strategies. Toward the first step, we follow the usual procedure of showing that higher types benefit from (late) disclosure more than lower types. When considering dynamic disclosure, these results make use of our earlier results that the early disclosure equilibrium must be characterized by a threshold. We denote this threshold $v^{E}$. We consider the expected payoff of the manager as a function of an arbitrary late-disclosure threshold,  $\overline{v}$; without loss, we take this to be less than $v^{E}$,  since the market infers that all managers with values $v > v^{E}$ disclosed early.  Define:

\begin{equation} 
H(s, \overline{v}, \theta):=\mathbb{E}[v \lvert s,  v \leq \overline{v}, \theta],
\end{equation}

\noindent as the equilibrium nondisclosure payoff given a threshold in the auxiliary game (given $\theta$), where we note that we do not need to consider the early disclosure strategy because $\overline{v} \leq v^{E}$.

\begin{lemma} \label{lemm:secondperiodDD}
The function $H(s, \overline{v}, \theta)- \overline{v}$ is decreasing in $\overline{v}$.
\end{lemma}

\begin{proof}[Proof of Lemma \ref{lemm:secondperiodDD}] 
 Note that: 
\begin{equation} 
H(s, \overline{v}, \theta)=\overbrace{s \Pr(\phi=V \lvert s,  v \leq  \overline{v}, \theta)}^{(i)} \\ + \overbrace{\mathbb{E}[ v \lvert \phi =N,  s,  v \leq  \overline{v}, \theta] \Pr(\phi =N \lvert  s,  v \leq  \overline{v}, \theta)}^{(ii)}  \label{eq:decomp}
\end{equation}

\noindent by the Law of Iterated Expectations.  Note that $H(s, \overline{v}, \theta)$ is differentiable in $ \overline{v}$,  since $ \overline{v}$ only shows up as the right endpoint in the integrals which define these conditional expectations.  We show that $\frac{d}{d \overline{v}} H(s, \overline{v}, \theta) < 1$.  

Consider (ii): we show that it is increasing at a rate less than 1. First, write the expression as $ \mathbb{E}[v \mathbbm{1}_{\phi =N} \lvert  s,  v \leq  \overline{v}, \theta].$ 
Note that $s$ only influences the probability that $\phi=N$, and in particular does not influence the distribution over $v$ (since if $\phi=N$ then $s$ conveys no information about $v$). Furthermore, $\theta$ has no direct influence on the manager's strategy on the event that $\phi=N$, by definition; adjusting this probability if necessary, we can drop this from the conditioning event. Furthermore, $\mathbb{E}[v \mathbbm{1}_{\phi =N} \lvert     v \leq  \overline{v}, \theta]$ increases at a rate less than $\mathbb{E}[v \lvert    v \leq  \overline{v}, \theta]$; indeed, both $\mathbb{E}[v \mathbbm{1}_{\phi =N} \lvert    v \leq  \overline{v}, \theta]$ and $\mathbb{E}[v (1-\mathbbm{1}_{\phi =N}) \lvert   v \leq  \overline{v}, \theta]$ are increasing in $\overline{v}$, and $\mathbb{E}[v \lvert  v \leq  \overline{v}, \theta]$ is the sum of these two.  Furthermore,  by Lemma \ref{lemm:BB}, since $g$ is log-concave, we have that $\mathbb{E}[v \lvert  v \leq  \overline{v}]$ increases at a rate less than 1.  

On the other hand, consider (i). We can write this term as 
$
 s\frac{q(1-p \theta)}{q(1-p \theta)+ (1-q)((1-p)+p\int_{v_{min}}^{\overline{v}}  g(v)dv)}. 
$

\noindent Upon inspection, this is \emph{decreasing} in $\overline{v}$, since as $\overline{v}$ increases the denominator increases and the numerator is constant. 

Putting this together, we have shown that $ \frac{d}{d \overline{v}}H(s, \overline{v}, \theta)$ is the sum of a term with slope less than 1 and a term which is non-increasing; thus, $ \frac{d}{d \overline{v}}H(s, \overline{v}, \theta) <1$.\end{proof} 

\noindent Since Lemma \ref{lemm:latethresh} states that any late-disclosure equilibrium must be characterized by a threshold (given $s$), Lemma \ref{lemm:secondperiodDD} implies that the equilibrium strategy in the auxiliary game is unique conditional on $\theta$ even when disclosure costs are present, provided log-concavity holds. Indeed, the fact that $H(s, \widehat{v}, \theta) - \widehat{v}$ is strictly decreasing in $\widehat{v}$ means that either (i) there exists a unique $\widehat{v}$ solving $H(s, \widehat{v}, \theta)-\widehat{v}=-c^{L}$,  or (ii) $H(s, \widehat{v}, \theta)-\widehat{v} \geq -c^{L}$ for all $\widehat{v}$; in the latter case, no manager types disclose. Note that in particular we cannot have $H(s, 0, \theta) \leq -c^{L}$, since even if $c^{L}=0$, we assume the manager is uninformed with positive probability, meaning that the left hand side must be positive.

It remains to show that $\theta$ is pinned down in equilibrium, since these results leave open the possibility for multiplicity in the case that the market makes different conjectures about the manager's strategy regarding disclosure. Toward ruling this out,\footnote{Note that this can only be ruled out for late disclosure, but not early disclosure, as discussed in the main text.} define $\widehat{v}^{L}_{\theta}$ as the unique second period threshold in the auxiliary game (i.e., when the market conjectures that veracious signals are disclosed by informed with probability $\theta$). 

Let $v^{L}(s)$ denote a candidate equilibrium late disclosure threshold in the original game---in Lemma \ref{lemm:latethresh}, this was denoted by $v^{L, \sigma^{E}}(s)$, but we suppress the implicit potential dependence on $\sigma^{E}$ since results earlier in this Appendix imply this is a threshold.

\begin{lemma} \label{lemm:interiorvs}
There exists some $\alpha \in [0,1]$ such that $v^{L}(s)= \alpha \widehat{v}^{L}_{1}(s) + (1- \alpha) \widehat{v}^{L}_{0}(s)$ is an arbitrary threshold. 
\end{lemma}

\begin{proof}[Proof of Lemma \ref{lemm:interiorvs}]
Letting $P(s, \emptyset)$ denote the equilibrium price function, we note that $P(s, \emptyset)$ must itself be a convex combination of $H(s, \overline{v},1)$ and $H(s, \overline{v}, 0)$, since it is the expectation over the event that the manager discloses a veracious signal, given all other parameters. Assume $H(s, v^{L}(s),1) \geq H(s, v^{L}(s), 0)$; analogous arguments hold when this inequality is flipped. Then  $H(s, v^{L}(s),1) \geq P(s, \emptyset) \geq H(s, v^{L}(s), 0)$. So if we have $v^{L}(s) - c^{L} > H(s, v^{L}(s),1)$, then $v^{L}(s) - c^{L} 
 > P(s, \emptyset)$, and if $v^{L}(s)-c^{L} <  H(s, v^{L}(s),1)$, then $v^{L}(s) - c^{L} < P(s, \emptyset)$.  So, since the conditional expectations are all increasing in $v^{L}(s)$, if the manager is indifferent between late disclosure decision---i.e., $v^{L}(s)-c^{L}= P(s, \emptyset)$---then by this argument, $v^{L}(s)$ cannot be outside of the interval defined by $\widehat{v}^{L}_{1}(s)$ and $\widehat{v}^{L}_{0}(s)$.
\end{proof}

By Lemma \ref{lemm:interiorvs}, any candidate threshold must be in between the ``extreme'' thresholds in the auxiliary game, $\widehat{v}^{L}_{1}(s)$ and $\widehat{v}^{L}_{0}(s)$, so for the moment we consider the behavior of these thresholds \emph{as a function of} $s$. Note that it is immediate that $\widehat{v}^{L}_{\theta}(s)$ is increasing in $s$ for any $\theta\in[0,1]$; indeed, this follows since $H(s, \widehat{v}, \theta) - \widehat{v}$ is increasing in both $s$ and decreasing in $\widehat{v}$, so that if $s$ increases $\widehat{v}$ must increase as well, in order to make the manager indifferent between the disclosure decisions. The next two Lemmas present some additional properties of these functions.

\begin{lemma}  \label{lemm:condthreshold} The thresholds
$\widehat{v}^{L}_{1}(s)$ and $\widehat{v}^{L}_{0}(s)$ are continuous in $s$.
\end{lemma} 

\begin{proof}[Proof of Lemma \ref{lemm:condthreshold}]
Consider $\widehat{v}^{L}_{1}(s)$, since the argument for $\widehat{v}^{L}_{0}(s)$ is analogous. Recall that $H(s, \widehat{v},1)$ is differentiable in $\widehat{v}$, since $\widehat{v}$ only appears as the endpoint of an integral in the expression defining it, and in particular Lemma \ref{lemm:secondperiodDD} shows that the slope increases at a rate less than 1. On the other hand, the threshold is defined implicitly as the solution to the equation $-c^{L}=H(s, \widehat{v},1) - \widehat{v}$. Since the derivative of the right hand side with respect to $\widehat{v}$ is non-zero, the implicit function theorem applies for any range of $s$ with $\widehat{v}^{L}_{1}(s) \in (v_{min},v_{max})$. Noting that we can never have $\widehat{v}^{L}_{1}(s)=v_{min}$ (since the payoff from nondisclosure is strictly positive in any threshold equilibrium, since $p <1$), we point out that if in fact $\widehat{v}^{L}_{1}(s)=v_{max}$, since this is constant at any $s$ above which this holds, we still have continuity at any possible $s$. That is, the implicit function still applies in this case, so that $\widehat{v}^{L}_{1}(s)$ is continuous at the lowest $s$ such that $\widehat{v}^{L}_{1}(s)=v_{max}$; and, since it is constant above this $s$, we have it is continuous at all $s$ as well. Thus, $\widehat{v}^{L}_{1}(s)$ is continuous in $s$, as claimed. 
\end{proof}

\begin{lemma}  \label{lemm:thresholdcompstat}
The thresholds $\widehat{v}^{L}_{1}(s)$ and $\widehat{v}^{L}_{0}(s)$ are increasing in $s$ at a rate less than 1.   Furthermore,  $\widehat{v}^{L}_{1}(s) - \widehat{v}^{L}_{0}(s)$ is decreasing in $s$.  
\end{lemma}

\begin{proof}[Proof of Lemma \ref{lemm:thresholdcompstat}]
Note that: 

\begin{equation*} \frac{ \partial}{\partial s}H(s,\widehat{v},0)  >  \frac{ \partial}{\partial s}H(s,\widehat{v},1)  ~~~~~ \frac{ \partial}{\partial \widehat{v}}H(s,\widehat{v},1) > \frac{ \partial}{\partial \widehat{v}}H(s,\widehat{v},0) .
\end{equation*} 

\noindent Indeed,  if $s$ is never disclosed when it is veracious (i.e., $\theta=0$),  then upon seeing nondisclosure, the probability it is the true value is higher.  Thus,  an increase in $s$ has more of an impact on the expectation, whereas an increase in $\widehat{v}$ has less of an impact.  Thus, for a fixed change in $s$,  in order to ensure that $H(s,\widehat{v},\theta) - \widehat{v} =c^{L}$,  the change in the left hand side is smaller when $\theta=1$ than when $\theta=0$,  and furthermore that $\widehat{v}$ does not need to increase as much in order to ensure equality holds.  Thus,  $\widehat{v}^{L}_{1}(s)$ increases by less than $\widehat{v}^{L}_{0}(s)$. 
\end{proof}

\subsection{Uniqueness}
We now put the above arguments together to present our results on uniqueness, starting with the late disclosure decision. Recall that we denote a candidate late-disclosure threshold function by $v^{L}(s)$.

\begin{lemma}  \label{lemm:uniquethresh}
The unique second period equilibrium strategy, conditional on the manager disclosing early above $v^{E}$, involves threshold $\widehat{v}^{L}_{1}(s)$ for $s > s^B$ and $\widehat{v}^{L}_{0}(s)$ for $s < s^B$.
\end{lemma}

\begin{proof}[Proof of Lemma \ref{lemm:uniquethresh}]
The proof essentially follows from putting together Lemmas \ref{lemm:interiorvs} and \ref{lemm:thresholdcompstat}. First, note that Lemma \ref{lemm:interiorvs} and the definition of $v^B$ immediately implies that $\widehat{v}^{L}(s^B)=v^B$. We consider two cases separately:

\begin{itemize} 
\item Suppose that $s > s^B$. In this case, we note that by Lemma \ref{lemm:thresholdcompstat}, we have that $s > \widehat{v}^{L}_{0}(s) > \widehat{v}^{L}_{1}(s)$, and by Lemma \ref{lemm:interiorvs} the only candidate values for $v^{L}(s)$ satisfy $\widehat{v}^{L}_{0}(s)  \geq v^{L}(s) \geq  \widehat{v}^{L}_{1}(s)$. From inspection, we therefore see that, for any possible conjecture of the market regarding the manager's behavior following an informative signal, the signal $s$ is above the market threshold. It follows that the only possible outcome in equilibrium is that the manager discloses informative signals. As a result, the unique equlibrium threshold is $v^{L}(s) =  \widehat{v}^{L}_{1}(s)$. 

\item Suppose that $s < s^B$. Again by Lemma \ref{lemm:thresholdcompstat}, we have that $s < \widehat{v}^{L}_{0}(s) < \widehat{v}^{L}_{1}(s)$. Again from inspection, we see that for any possible conjecture of the market regarding the manager's behavior following an informative signal, the signal $s$ is below the disclosure threshold. It follows that the only possible outcome in equilibrium is that the manager does not disclose informative signals. As a result, the unique equilibrium threshold is $v^{L}(s) =  \widehat{v}^{L}_{0}(s)$.
\end{itemize}

\noindent Putting these observations together completes the proof. \end{proof}

At this point, we have shown that we can restrict to equilibria where the first-period disclosure decision is characterized by a threshold, and furthermore, where the second-period disclosure threshold (function) is unique conditional on the first-period disclosure threshold. Lemma \ref{lemm:condthreshold} implies that this threshold is continuous in the external signal. 

We now show that the nondisclosure payoff increases in the threshold at a rate less than 1, the last ingredient in our uniqueness argument. 

\begin{lemma} \label{lemm:finalpart} 
Suppose the market assumes that, when the manager is indifferent between disclosure decisions, disclosure occurs with probability $\widetilde{\theta}$. There exists at most one value $v^{E}$ such that the manager is indifferent between disclosing early and not when the value is $v^{E}$,  provided that $v^L(s)$ is itself an equilibrium disclosure threshold given the corresponding threshold $v^{E}$. 
\end{lemma}

Note that this shows the equilibrium is unique up to the choice of $\widetilde{\theta}$; this completes the proof of the Proposition since we take equilibrium to involve $\widetilde{\theta}=0$. See the discussion in the main text in Section \ref{freqd} for a discussion of this assumption.

\begin{proof}[Proof of Lemma \ref{lemm:finalpart}] 
We first consider how the equilibrium price function changes in the \emph{early} disclosure threshold.  Using our assumption that $v^L(s)$ is itself an equilibrium disclosure decision,  we can rewrite equation (\ref{eq:ratio}) with $v^{E}$ explicitly written in. Let 
\begin{equation*} 
r(s,\widehat{v}^{L}_{0}(s),v^{E},\widetilde{\theta})=\mathbbm{1}_{s >\max\{\widehat{v}^{L}_{0}(s),v^{E}\}}+\widetilde{\theta}\mathbbm{1}_{s=\max\{\widehat{v}^{L}_{0}(s),v^{E}\}},
\end{equation*}
refer to the probability that an informative signal would be disclosed. This yields:

\begin{equation} 
0= \widehat{v}^{L}_{0}(s)- c^{L} -  \frac{ qs(1-p\cdot r(s,\widehat{v}^{L}_{0}(s),v^{E},\widetilde{\theta})) +(1-q) \left((1-p)\int_{v_{min}}^{v_{max}} v g(v)dv + p \int_{v_{min}}^{\max\{\widehat{v}^{L}_{0}(s),v^{E}\}} v g(v) dv \right)}{q(1-r(s,\widehat{v}^{L}_{0}(s),v^{E},\widetilde{\theta})) +(1-q) \left((1-p)\int_{v_{min}}^{v_{max}} g(v)dv + p \int_{v_{min}}^{\max\{v^{E},\widehat{v}^{L}_{0}(s)\}} g(v) dv \right)} .  
\end{equation}

\noindent Indeed, recall that this derivation held for an arbitrary first period strategy, and therefore holds assume the first-period strategy is used in the first period (which, unlike when this expression was first introduced, we have now shown must be the case).  We consider two-cases separately:

\begin{itemize} 
\item If $\widehat{v}^{L}_{0}(s)$ is such that $v^{E} >\widehat{v}^{L}_{0}(s)$, an increase in $v^{E}$ has no impact on $\widehat{v}^{L}_{0}(s)$; 

\item If $\widehat{v}^{L}_{0}(s)$ is such that $v^{E} <\widehat{v}^{L}_{0}(s)$, then the increase in $v^{E}$ increases $\widehat{v}^{L}_{0}(s)$ at a rate less than 1; the argument is identical to the proof from Lemma \ref{lemm:thresholdcompstat}, since $v^{E}$ takes the same role as $\widehat{v}^{L}_{0}(s)$ in the ratio in equation (\ref{eq:ratio}). 
\end{itemize}

\noindent  Note that an identical argument applies to $\widehat{v}^{L}_{1}(s)$, meaning that this will also increase at a rate less than 1.  So consider the equilibrium price function.  Define $v^B$ precisely as it was defined previously,  namely as the intersection point of the thresholds in the two auxiliary games.    If $\min\{\widehat{v}^{L}_{1}(s), \widehat{v}^{L}_{0}(s)\} < v^{E}$,  then it is exactly the same as if $v^{E}=v_{max}$.  If $\min\{\widehat{v}^{L}_{1}(s), \widehat{v}^{L}_{0}(s)\}  > v^{E}$, then it is equal to $\mathbb{E}[v \lvert v \leq v^{E} ]$ (in other words, if the intersection point is less than $v^{E}$, then $v^{E}$ does not matter and we can set it equal to $v_{max}$; if it is less than it, the price function is in a range where it is increasing at a rate less than 1, since the argument is identical to the late disclosure case).

With this in mind, we note the following: 

\begin{itemize} 
\item Consider the manager's expected payoff from not disclosing early if the second-period signal is informative. Note that, given $\widetilde{\theta}$, the payoff of the manager in this event is a convex combination, with weight $\widetilde{\theta}$, between $\lim_{v \rightarrow^{-} v^{E}} P(v, \emptyset)$ and $\lim_{v \rightarrow^{+} v^{E}} P(v, \emptyset)$; but since both of these increase at a rate less than 1 (by the previous argument in this proof), the overall rate of change is less than 1.

\item On the other hand, suppose the signal is uninformative.  In this case,  the fact that the manager's payoff increases at a rate less than 1 follows an identical argument from Lemma \ref{lemm:dyesignalsfirst};  the manager's payoff given a signal $s$ is $\max\{v^{E}-c^{E}, P(s, \emptyset)\}$;  as long as $v^{E}-c^{E}$ is less than or equal to  $P(v_{max}, \emptyset)\}$. 

On the other hand, suppose that $v^{E}-c^{E} \geq P(v_{max}, \emptyset)$.   Then in this case, since $s < v_{max}$ with probability 1, the manager would strictly prefer to disclose.  Therefore,  the indifference condition can only be satisfied if $v^{E}$ is such that $v^{E} -c^{E} < P(v_{max}, \emptyset)$.

Putting these observations together,  consider $\int_{v_{min}}^{v_{max}} \max\{v^{E}-c^{E}, P(s, \emptyset)\} g(s)ds$.  Since there is positive probability that the signal is such that $v^{E}-c^{E} < P(s, \emptyset)$,  this cannot increase at a rate faster than 1. 

\end{itemize}

Now crucially, the probability that the signal is informative is constant and independent of $v^{E}$.  Therefore,  if $r_{a}(v^{E})$ is the rate of change when the signal is veracious and $r_{i}(v^{E})$ is the rate of change when the signal is non-veracious,  the total rate of change is $q r_{a}(v^{E}) + (1-q) r_{i}(v^{E})$.  Since $r_{a}(v^{E}) \leq 1$ and $r_{i}(v^{E}) < 1$ by the above arguments,  provided that $v^{E}$ is in a range such that indifference can possibly be satisfied (i.e., $v^{E} - c^{L} < P(v_{max}, \emptyset)$. Therefore,  the total payoff from nondisclosure cannot increase at a rate greater than or equal to 1 as $v^{E}$ increases.   Since the payoff in the event that the manager discloses early increases at a rate \emph{equal} to 1,   we have that there can only be one threshold $v^{E}$ such that the manager is indifferent between disclosing early and not. \end{proof}

\subsection{Additional Discussion}  \label{app:summary}

\noindent We conclude with a discussion of several properties that are of further interest, beyond the scope of the results at hand: 

\begin{itemize}
\item[1.] The results in Section \ref{sect:early} show that \emph{no} equilibrium---mixed or otherwise---can involve early disclosure below $v^{B}$. 

\item[2.] For early disclosure, there exist multiple threshold equilibria that can arise by assuming different tiebreaking rules of the manager in the case of indifference. As discussed in the main text, this corresponds to the expected price intersecting the 45 degree line at different points between the left and right limit of the expected price function at the discontinuity. Our analysis implies all equilibria (including mixed) are of this form if disclosure is only early. This follows from Lemma \ref{lemm:dyesignalsfirst}, which leaves open the possibility of this form of multiplicity.

\item[3.] In the knife-edge cases where disclosure is only late and disclosure costs are zero, uniqueness can be obtained without log-concavity; the usual argument for uniqueness applies to the case of late disclosure. Log-concavity is primarily invoked to ensure that $H(s, \overline{v}, \theta)-\overline{v}$ is decreasing in $\overline{v}$, which in turn is used to obtain certain natural properties of $v^{L}(s)$ (for instance, that it is increasing, and continuous by the implicit function theorem). 

\item[4.] In our setting, disclosure dynamics can be generated by disclosure costs, a feature which in itself may be of independent interest. That is, for costs $c^{L}> c^{E}$, disclosure may occur both early and late. We highlight some technical issues relevant more generally to future work on dynamic costly disclosure with this in mind.

To obtain unique late-disclosure thresholds with dynamic disclosure, we require log-concavity of $G$ \emph{conditional} on non-(early) disclosure. While vacuous if there is no early disclosure at all---or, for that matter, if late disclosure is not allowed either---in the general case this step is non-trivial and our proof requires the restriction to continuous late-disclosure thresholds (since continuity implies that the slope is less than 1). In other words, given a threshold $v^{E}$, our proof shows uniqueness of a disclosure threshold $v^{L}(s)$ that is continuous and increasing in $s$---and visa versa. That said, we do not rule out exotic early-disclosure strategies which induce non-log-concave conditional distributions, supported by the manager switching between different $v^{L}(s)$ functions depending on the signal realization (with these switches generating discontinuities). While this cannot occur if disclosure is only early or only late, we also conjecture that this possibility cannot emerge even more generally. 
\end{itemize}

\newpage

\singlespacing


\begin{thebibliography}{20}

\bibitem{} Abrams, E., Libgober, J., List, J.  2023. Research registries and the credibility crisis: an
empirical and theoretical investigation. Working paper.

         \bibitem{} Acharya, V., DeMarzo, P., Kremer, I. 2011. Endogenous information flows and the clustering of announcements. \emph{American Economic Review} 101, 2955-2979.
         
              


\bibitem{} Aggarwal, R., Schirm, D. 1998. Asymmetric impact of trade balance news on asset prices. \emph{Journal of International Financial Markets, Institutions and Money} 8, 83-100. 

    \bibitem{} Ahern, K. Sosyura, D. 2015. Rumor has it: Sensationalism in financial media. \emph{Review of Finanacial Studies} 28(7), 2050–2093.


    \bibitem{} Alperovych, Y. Cumming, D., Czellar, V., Groh, A. 2021.
M \& A rumors about unlisted firms.
\emph{Journal of Financial Economics}
Volume 142 (3),
1324-1339.

\bibitem{} Bagnoli, M., Bergstrom, T. 2005. Log-concave probability and its applications. \emph{Economic Theory} 26(2), 445-469.


\bibitem{} Banerjee, S., Davis, J. Gondhi, N. 2020. The man(ager) who knew too much. Working paper.

\bibitem{} Banerjee, S., Green, B. 2015. Signal or noise? Uncertainty and learning about whether other traders are informed. \emph{Journal of Financial Economics} 117(2), 398-423.


\bibitem{} Barlevy, G., Veronesi, P. 2003. Rational panics and stock market crashes. \emph{Journal of Economic Theory} 110, 234-263.

\bibitem{} Barr, K. 2022. How much did Twitter's verification chaos cost insulin maker Eli Lilly and Twitter itself? \emph{Gizmodo}.


 \bibitem{} Bond, P., Goldstein, I., Prescott, E. 2010. Market-based corrective actions. \emph{Review of Financial Studies} 23 (2), 781–820
 

\bibitem{} Blot, C., Hubert, P., Labondance, F. 2020. The asymmetric effects of monetary policy on stock price bubbles. Working paper.


    \bibitem o	Cai, W., Quan, X.,  Zhu, Z. 2023. Rumors in the sky: Corporate rumors and stock price synchronicity. \emph{International Review of Financial Analysis}, forthcoming.
    
    \bibitem{} Capkun, V., Lou, Y., Otto, C.A. and Wang, Y. 2022. Do firms respond to peer disclosures? Evidence from disclosures of clinical trial results. \emph{The Accounting Review}, forthcoming

    \bibitem{} Clarkson, P., Joyce, D., Tutticci, I. 2006. Market reaction to takeover rumour in Internet Discussion Sites. \emph{Accounting \& Finance} 46, 31-52
    
	
	

\bibitem{} Davis, F., Khadivar, H., Walker, T. 2021.
Institutional trading in firms rumored to be takeover targets.
\emph{Journal of Corporate Finance} 66.

		\bibitem{} Dye, R. 1985. Disclosure of non-proprietary information. \emph{Journal of Accounting Research} 23, 123-145.
  
	\bibitem{} Dye, R., Sridhar, S. 1995. Industry-wide disclosure dynamics. \emph{Journal of Accounting Research}, 33, 157-174.
 
 
 
	\bibitem{} Einhorn, E. 2018. Competing information sources. \emph{The Accounting Review}, 93, 151-176.
	\bibitem{} Frenkel, S., Guttman, I., Kremer, I. 2020.  The effect of exogenous information on voluntary disclosure and market quality. \emph{Journal of Financial Economics}, 138(1), 176-192.
	\bibitem{} Frenkel, S., Guttman, I., Kremer, I. 2023. Disclosure with informed market. Work-in-progress.

 \bibitem{} Ganuza, J., Penalva, J. 2010. Signal orderings based on dispersion and the supply of private information in auctions. \emph{Econometrica} 78 (3), 1007-1030.

 


\bibitem{}  Gennotte, G., Leland, H. 1990. Market liquidity, hedging, and crashes. \emph{American Economic Review} 30, 999-1021.



 
 
	
 

\bibitem{} Goldberg, S. 2013. Time variation in asset price responses to macro announcements. \emph{Federal Reserve Report, New York}.

\bibitem{} Goncalves, D., Libgober, J., Willis, J. 2023. Retractions:
Learning from information about information. Working paper.
 
 
	\item{} Grossman, S.J. 1981. The informational role of warranties and private disclosure about product quality. \emph{Journal of Law and Economics} 24, 461-484.

 \item{} Guttman, I., Marinovic, I. 2018. Debt contracts in the presence of performance manipulation. \emph{Review of Accounting Studies} 23, 1005-1041.

 \item{} Harrwell, D. 2022. A fake tweet sparked panic at Eli Lilly and may have cost Twitter millions. \emph{The Washington Post}.


\bibitem{} Johnson, J., and Myatt, D. 2006. On the simple economics of advertising, marketing, and product design. \emph{American Economic Review} 96 (3), 756-84.
		\bibitem{} Jung, W., Kwon, Y. 1988. Disclosure when the market is unsure of information endowment of managers. \emph{Journal of Accounting Research} 26, 146-153.
  

\item{} Kanodia, C., Singh, R., Spero, A. 2005. Imprecision in accounting measurement: can it be value enhancing? \emph{Journal of Accounting Research} 43, 487-519.

\bibitem{} Kapferer, J. 2013. Rumors: uses, interpretations, and images. \emph{Transaction Publishers}.

  \bibitem{} Kartik, N., Lee, F., Suen, W. 2019. A proposition on Bayesian updating and applications to communication games. Working paper.

  \bibitem{} Kelly, H. 2023. How to avoid falling for misinformation, fake AI images on social media. \emph{The Washington Post}.

  \bibitem{} Kimmel, A. 2013. Rumors and rumor control: a manager's guide to understanding and combatting rumors. \emph{Taylor \& Francis}.
  
		\bibitem{} Kremer, I. and Schreiber, A. and Skrzypacz, A. 2021. Disclosing a random walk. Working paper. 


  \bibitem{} Kohlbeck, M., Vakilzadeh, H. 2020. False news determinants and its association with financial reporting quality. Working paper.

  \bibitem{} Lewis, T., Sappington, D. 1994. Supplying information to facilitate price discrimination. \emph{International Economic Review}, 35(2), 309–327. 
  
\bibitem{} Liu, B., Moss, A. 2023. The role of accounting information in an era of fake news. Working paper.

\bibitem{} Marett, K., Joshi, K. 2009. The decision to share information and rumors: Examining the role of motivation in an online discussion forum. \emph{Communications of the Association for Information Systems} 24. 

\bibitem{} Maynard, M. 2008. United Airlines shares fall on false report of bankruptcy. \emph{The New York Times}.

\bibitem{} Menon, R. 2020. Voluntary disclosures when there is an option to delay disclosure. \emph{Contemporary Accounting Research} 37, 829-856.


	\bibitem{} Milgrom, P. 1981. Good news and bad news: representation propositions and applications. \emph{The Bell Journal of Economics }12, 380-391.

 

 \bibitem{} Ottaviani, M., Sorensen, P. 2006. 
Professional advice,
\emph{Journal of Economic Theory},
126 (1),
120-142.
   \bibitem{} Roychowdhury, S., Sletten, E. 2012. Voluntary disclosure incentives and earnings informativeness. \emph{The Accounting Review} 87(5), 1679-1708.

   \bibitem{} Schaper, D. 2008. Bankruptcy rumor sparks United Airlines sell-Off. \emph{National Public Radio}.

   \bibitem{} Schindler, M. 2007. Rumors in financial markets. \emph{John Wiley and Sons.}
   
	\bibitem{} Sletten, E. 2012. The effect of stock price on discretionary disclosure. \emph{Review of Accounting Studies} 17. 
 


\bibitem{} Veronesi, P. 1999. Stock market overreaction to bad news in good times: a rational expectations
equilibrium model. \emph{Review of Financial Studies} 12(5), 975-1007.

 \bibitem{} Verrecchia, R. 1983. Discretionary disclosure. \emph{Journal of Accounting and Economics} 5, 179-194.

 \bibitem{} Xu, N., You, Y. 2023. Main Street’s pain, Wall Street’s gain. Working paper.

 \bibitem{} Zhang, X. 2006a. Information uncertainty and stock returns. \emph{Journal of Finance} 61 (1), 105–137. 
 
 \bibitem{} Zhang, X. 2006b. Information uncertainty and analyst forecast behavior. \emph{Contemporary Accounting Research} 23(2), 565–590. 
	
\end{thebibliography}
\end{document}